\documentclass[%
 reprint,
 amsmath,amssymb,
 aps,
 pra,
]{revtex4-2}

\usepackage{lipsum}
\usepackage{graphicx}% Include figure files
\usepackage{dcolumn}% Align table columns on decimal point
\usepackage{bm}% bold math
\usepackage{float}

\usepackage{enumitem}
\usepackage{amsthm}
\usepackage{physics}
\usepackage{mathtools}
\usepackage{xcolor}
\usepackage{braket}
\usepackage{leftindex}
\usepackage{MnSymbol}

\usepackage{bbm}
\usepackage[normalem]{ulem}

\usepackage[noline]{algorithm2e}
\RestyleAlgo{ruled}
\SetKwInput{KwData}{Input}
\SetKwInput{KwResult}{Output}
 
\usepackage{array,makecell}
\usepackage{multirow}
\def\thickhline{\noalign{\hrule height.8pt}}

\usepackage{ragged2e}
\newcolumntype{P}[1]{>{\raggedright\arraybackslash\hspace{0pt}}p{#1}}

\usepackage[toc,page]{appendix}

\usepackage{tikz}
\usetikzlibrary{quantikz2}

\NewDocumentEnvironment{alignb}{b}{%
  \begin{align*}
  \refstepcounter{equation} #1 \tag{\theequation}
  \end{align*}
}{\ignorespacesafterend}

\newcommand{\bigo}[1]{\mathcal{O}\left({#1}\right)}

\newsavebox{\mstrut}
\newcommand{\bbra}[1]{%
    \sbox{\mstrut}{\(#1\)}%
    \mathinner{\left\langle\kern-0.25\ht\mstrut\left\langle{#1}\right|\right.\!}%
}
\newcommand{\kett}[1]{%
    \sbox{\mstrut}{\(#1\)}%
    \mathinner{\!\left.\left|{#1}\right\rangle\kern-0.25\ht\mstrut\right\rangle}%
}
\newcommand{\bbrakett}[2]{%
    \sbox{\mstrut}{\(#1\)}%
    \sbox{\mstrut}{\(#2\)}%
    \mathinner{\left\langle\kern-0.25\ht\mstrut\left\langle{#1}\right.\right.}\!\! | \!\!\mathinner{\left.\left.{#2}\right\rangle\kern-0.25\ht\mstrut\right\rangle}%
}
\newcommand{\vecc}[1]{%
    \text{vec}(#1)%
}
\newcommand{\kettbbra}[2]{
    \kett{#1} \! \bbra{#1}
}

\newcommand{\exptsuper}[3]{%
    \sbox{\mstrut}{\(#1\)}%
    \sbox{\mstrut}{\(#2\)}%
    \sbox{\mstrut}{\(#3\)}%
    \mathinner{\left\langle\kern-0.25\ht\mstrut\left\langle{#1}\right.\right|}\! {#2} \!\mathinner{\left|\left.{#3}\right\rangle\kern-0.25\ht\mstrut\right\rangle}%
}

\newtheorem{theorem}{Theorem}
\newtheorem{lemma}{Lemma}
\newtheorem{corollary}{Corollary}
\newtheorem{proposition}{Proposition}

\theoremstyle{definition}

\usepackage[toc,page]{appendix}
\usepackage[tight]{minitoc}
\usepackage{tocloft}

\doparttoc % Tell to minitoc to generate a toc for the parts
\faketableofcontents % Run a fake tableofcontents command 
\renewcommand \partname{}
\usepackage{hyperref}
\hypersetup{
	colorlinks = true,
    linkcolor = [rgb]{0.25, 0.34, 0.63},
	citecolor = [rgb]{0.13,0.55,0.13},
	urlcolor = [rgb]{0.25, 0.34, 0.63}}
\usepackage[capitalize]{cleveref}
\expandafter\providecommand\csname href@noop\endcsname[2]{#2} % To remove bib URL warnings

\begin{document}

\title{Quantum simulation in the Heisenberg picture via vectorization}

\author{Shao-Hen Chiew}
 \email{shao.chiew@epfl.ch}
\author{Armando Angrisani}%
\author{Zoë Holmes}%
\author{Giuseppe Carleo}%
\affiliation{%
Institute of Physics, École Polytechnique Fédérale de Lausanne (EPFL), CH-1015 Lausanne, Switzerland, and \\
Center for Quantum Science and Engineering, EPFL, CH-1015 Lausanne, Switzerland
}%

\date{\today}% It is always \today, today,

\begin{abstract}
We present a general framework for simulating quantum systems in the Heisenberg picture on quantum hardware. Based on the vectorization map, our framework fully exploits the mapping between operators and quantum states, allowing any task defined on Heisenberg operators to be mapped to standard Schrödinger-picture tasks that are naturally accessible via quantum computers and simulators.
This yields new or improved protocols for tasks such as operator sampling, the computation of OTOCs/superoperator expectation values and their higher order moments, two-point correlators, and operator stabilizer and entanglement entropies.
Our approach is also amenable to implementation, as it inherits the structure and resource requirements of the (forward and time-reversed) Schrödinger-picture quantum simulation problem. We demonstrate this by proposing implementations of our framework for a 2D problem on digital and analog quantum simulators, taking into account device connectivity constraints.
\end{abstract}

\maketitle

\section{Introduction} \label{sec:intro}

Quantum computers are devices that process and extract information directly through the laws of quantum mechanics. Their operation can be most naturally described in the Schrödinger picture, where quantum computations evolve a quantum state (stored in quantum memory) according to the Schrödinger equation, followed by measurements of fixed observables. From this viewpoint, quantum simulation follows as a direct application \cite{feynman1982simulating}, enabling the study of interacting quantum systems, with applications ranging from many-body physics and quantum chemistry to materials science \cite{mcardle2020quantum,bauer2020quantum,fauseweh2024quantum,eisert2025mind}.

This view of quantum computation naively appears at odds with the equivalent Heisenberg picture, where observables evolve backwards in time, while states remain fixed. Nonetheless, many questions concerning dynamics such as operator growth, transport, and scrambling are most naturally formulated in the Heisenberg picture \cite{khemani2018operator,von2018operator,nahum2018operator,roberts2018operator,qi2019quantum,xu2024scrambling}, via operatorial properties such as out-of-time-ordered correlators (OTOCs) \cite{larkin1969quasiclassical,maldacena2016bound,roberts2017chaos}, dynamical correlators \cite{baez2020dynamical}, and operator analogues of entanglement \cite{prosen2007operator} and stabilizer entropies \cite{dowling2025magic}. Owing to their central role, a range of experimental protocols and algorithms have been developed to access certain quantities such as OTOCs individually \cite{yoshida2019disentangling,landsman2019verified,schuster2022many,sundar2022proposal,green2022experimental,xu2024scrambling,swingle2016measuring,schuster2023operator,google2025observation}.

This observation has also been fruitfully exploited by classical methods, for instance via matrix product operators \cite{hemery2019matrix,xu2020accessing} or by directly back–propagating observables in the Pauli basis \cite{rall2019simulation, aharonov2022polynomial, fontana2023classical, rudolph2023classical, schuster2024polynomial, gonzalez2024pauli, angrisani2024classically, angrisani2025simulating, martinez2025efficient, rudolph2025pauli, teng2025leveraging}. However, such methods encounter bottlenecks under general dynamics involving high entanglement and/or magic\ \cite{dowling2024magic}, analogous to classical simulation methods based on the Schrödinger picture. This prompts the question of whether the advantages conferred by quantum computers for quantum simulation can also be carried over to the Heisenberg picture.

\begin{figure*}
    \centering
    \includegraphics[width=1\linewidth]{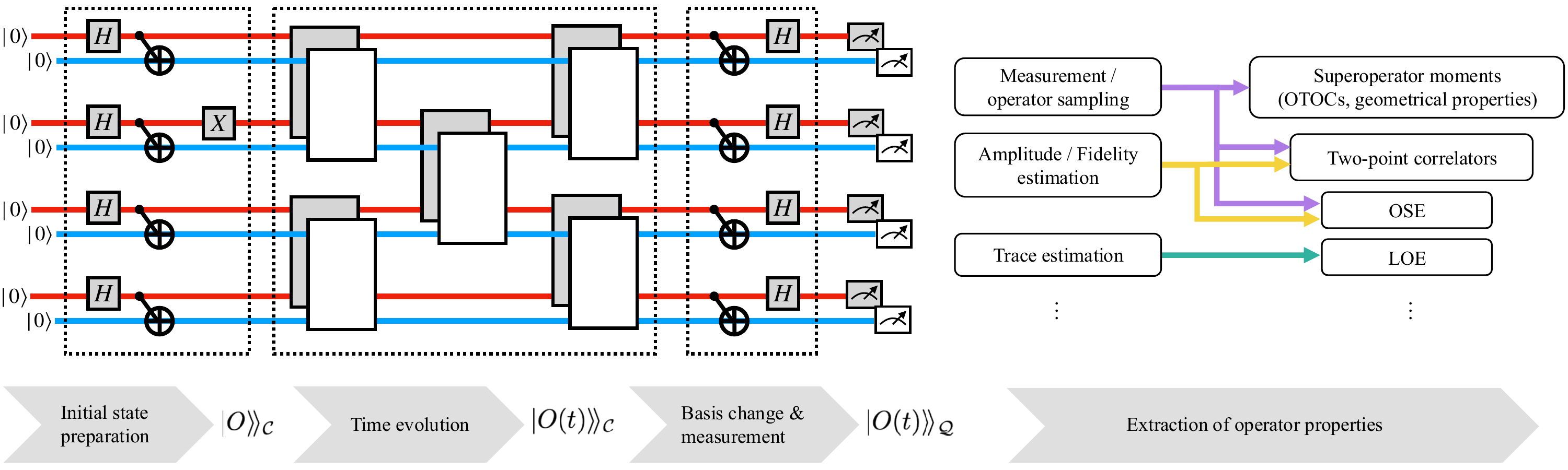}
    \caption{\textbf{Schematic summarizing our algorithmic framework.} Via the circuit schematic on the left, samples of the $2n$-qubit state $\kett{O(t)}_\mathcal{Q}$ are prepared via an initialize-and-evolve protocol, illustrated for the case of initializing and ending in the Pauli basis $\mathcal{P}$. Subsequently, by extracting information from it via measurements and other subroutines such as amplitude, fidelity, or trace estimation, a plethora of tasks in the Heisenberg picture can be accomplished naturally.}
    \label{fig:main}
\end{figure*}

In this work, we introduce a general framework for quantum simulation in the Heisenberg picture using quantum computers. Central to our approach are the vectorization and transfer matrix formalisms, which we exploit to encode $n$-qubit Heisenberg operators $O(t)$ as $2n$-qubit quantum states $\kett{O(t)}$ in a fully structure-preserving manner. As a consequence, \textit{any} task naturally defined in the Heisenberg picture -- including sampling from the Pauli distribution of an operator, the computation of OTOCs, operator geometrical properties, two-point correlators, operator stabilizer and entanglement entropies -- can be mapped to tasks in the Schrödinger picture via the state $\kett{O(t)}$. This enables us to propose new algorithms for their solution that leverage existing subroutines such as expectation estimation, fidelity/amplitude estimation, and more.

A sketch of our framework is shown in Fig.~\ref{fig:main}. Concretely, it consists of first preparing and time-evolving on a quantum computer to obtain samples of the state $\kett{O(t)}$, followed by the extraction of information regarding $O(t)$ from them via measurements and various protocols. These steps are described generally in Section~\ref{sec:main_vect}. The following sections then describe applications of this procedure for the computation of various quantities such as operator stabilizer (Section~\ref{sec:ose}) and entanglement entropies (Section~\ref{sec:loe}), OTOCs/superoperator expectation values (Section~\ref{sec:superop}), and two-point correlators (Section~\ref{sec:2pc_main}).

Within our framework, the quantum simulation of $O \rightarrow O(t)$ inherits much of the resource requirements of the Schrödinger-picture Hamiltonian simulation problem. This renders our algorithms amenable to implementation, including on current-generation quantum computers and simulators limited by connectivity constraints. To demonstrate this, we describe an experimental proposal to simulate Heisenberg time evolution and study operator growth under a 2D lattice Hamiltonian using devices constrained to 2D connectivity, with only a small constant overhead in circuit depth compared to the Schrödinger-picture Hamiltonian simulation problem. 

Our framework can be viewed as a generalization and formalization of certain existing protocols that require access to time-reversal, and is related to Bell sampling \cite{hangleiter2024bell} and certain state and Hamiltonian learning protocols \cite{huang2021information,king2024exponential,zhao2024learning,castaneda2025hamiltonian}. For instance, under the lens of the vectorization and superoperator formalisms, existing protocols for the computation of OTOCs and operator size simply reduce to expectation estimation problems on $\kett{O(t)}$. Further exploitation and generalization of this observation then yields novel algorithms for other quantities such as two-point correlators and operator stabilizer and entanglement entropies.

Finally, we study an existing class of $n$-qubit protocols \cite{swingle2016measuring,mi2021information,schuster2022many,schuster2023learning,cotler2023information,schuster2023operator,google2025observation,algorithmiq} that halves qubit count at the expense of doubling circuit depth, and is able to solve a subset of the above tasks with matching or worse sample complexity. Through vectorization, we relate them to our framework, which can intuitively be understood as an `unfolding' of its tensor network diagram. This yields generalizations and new algorithms for tasks such as sampling from the matrix elements of an operator.

\vspace*{0.8cm}

%% TOC stuff START
\begingroup
\vspace*{-0.75cm}
\makeatletter
\def\mtc@hook@beforeinputfile{\vskip 1.75\baselineskip} % <-- space after CONTENTS
\mtcprepare[c]
\mtcsetformat{parttoc}{tocrightmargin}{2.55em plus 1fil} % <-- left-aligned (ragged-right) titles

\renewcommand\numberline[1]{\hb@xt@\@tempdima{#1.\hfil}}

\renewcommand{\mtcgapafterheads}{2.0\baselineskip}
\renewcommand*\l@section{%
  \addvspace{0.6\baselineskip}% <-- space between sections
  \@dottedtocline{1}{0.1em}{2.5em}%
}
\renewcommand*\l@subsection{\@dottedtocline{2}{2.6em}{1.5em}}
\mtcsetformat{parttoc}{dotinterval}{10000}
\renewcommand{\ptctitle}{CONTENTS}
\renewcommand{\ptifont}{\small\bfseries}
\renewcommand{\ptcSfont}{\small\rmfamily}

\noptcrule
\makeatother

\vspace*{-1.5cm}
\part{}
\parttoc
\vspace*{-0.5cm}
\endgroup
%% TOC stuff END

\section{Preliminaries} \label{sec:prelim}

We begin by reviewing the Heisenberg picture, and the vectorization, superoperator and transfer matrix formalisms, setting up notation along the way. Readers familiar with them can skip directly to Section~\ref{sec:main_vect}.

\subsection{Heisenberg picture} \label{sec:background}

In the Heisenberg picture, observables evolve as:
\begin{equation}
O \longrightarrow O(t) \equiv \mathcal{E}^\dagger(O),
\end{equation}
where $\mathcal{E}^\dagger$ denotes the adjoint (dual) of the quantum channel $\mathcal{E}$ with respect to the Hilbert--Schmidt (HS) inner product $\tr(A^\dagger B)$. This map is characterized by the requirement that expectation values match those in the Schr\"odinger picture:
\begin{equation}
\tr(\mathcal{E}(\rho)\,O)
\;=\;
\tr(\rho\,\mathcal{E}^\dagger(O))
\qquad
\text{for all states }\rho.
\end{equation}
As a special case, under unitary dynamics $\mathcal{E}(\cdot)=U(\cdot)U^\dagger$ lead to $O(t) = U^\dagger O U$.
A useful property of the adjoint map is that unitality and trace preservation are dual to each other: $\mathcal{E}^\dagger$ is unital (i.e., identity-preserving) if and only if $\mathcal{E}$ is trace-preserving, and $\mathcal{E}^\dagger$ is trace-preserving if and only if $\mathcal{E}$ is unital. Since, by definition, a quantum channel $\mathcal{E}$ is trace-preserving, it follows that its adjoint $\mathcal{E}^\dagger$ is unital, i.e. $\mathcal{E}^\dagger(\mathbb{I})=\mathbb{I}$.

Given any complete operator basis over $n$ qubits $\mathcal{Q} = \{ Q_k\}_{k=1}^{4^n}$ that is orthogonal with respect to the Hilbert-Schmidt (HS) inner product, i.e., $\text{tr}(Q_k^\dagger Q_l) = N \delta_{kl}$, where $N \in \mathbb{R}$ is the normalization factor, $O$ can be expanded as $O = \sum_k c_k Q_k$, where:  
\begin{equation} \label{eq:op_amp_def}
    c_k \equiv \frac{1}{N} \tr(Q_k^\dagger O)
\end{equation}
are referred to as the \textit{operator amplitudes of $O$} in the basis $\mathcal{Q}$. Particularly important to us are the computational operator basis $\mathcal{C} \equiv \{ \ketbra{i}{j}: i, j = 0,1 \}^{\otimes n}$ and the Pauli basis $\mathcal{P} \equiv \{I,X,Y,Z\}^{\otimes n}$. In the latter case, the amplitudes can also be interpreted as a dynamical two-point correlation function at infinite temperature $c_k(t) = \langle P_k O(t)  \rangle_\rho$, where $\rho = \mathbb{I}/2^n$ is the infinite temperature thermal state. 

The normalized absolute-squared values of $c_k$:
\begin{equation} \label{eq:op_prob_def}
    p(Q_k, O) \equiv \frac{\abs{c_k}^2}{\sum_k \abs{c_k}^2} 
\end{equation}
form a probability distribution, referred to as the \textit{operator distribution of $O$} in the basis $\mathcal{Q}$. In the Pauli basis, we refer to $p(P_k, O)$ as \textit{Pauli probabilities}. 

The $4^n$ operator amplitudes completely characterize a Heisenberg operator $O(t)$. Under unitary dynamics generated by a geometrically local Hamiltonian, the support of an initially local Pauli operator grows within a lightcone delimited by the Lieb-Robinson bound. Analogous to the notions of wavefunction entanglement and magic in the Schrödinger picture, operator entanglement \cite{prosen2007operator} and magic \cite{dowling2025magic} can be defined in the Heisenberg picture, respectively dictating the spreading of operators in space, and the complexity of the operator within its support. They can be quantified by operator entanglement and stabilizer entropies that are non-linear functions of $c_k$'s.

In the absence of complete characterization, coarse-grained statistical properties of $O$ are invaluable. An example of such a quantity is the out-of-time-ordered correlator (OTOC), defined as:
\begin{equation} \label{eq:def_otoc}
    \text{OTOC}(O(t), P, Q) \equiv \frac{1}{2^n}\tr(O(t)^\dagger P^\dagger O(t) Q).
\end{equation}
When $P,Q$ are taken to be equal and spatially local, the OTOC measures the degree of their overlap with $O(t)$, and therefore how much of $O$ has spread to that position in space-time. As a central quantity in studies of information scrambling and operator growth \cite{von2018operator,nahum2018operator,roberts2018operator,qi2019quantum,fisher2023random,xu2024scrambling,zhang2025quantum}, it is experimentally measurable given the ability to implement time-reversed dynamics, via protocols closely resembling Loschmidt echoes \cite{swingle2016measuring,yoshida2019disentangling,landsman2019verified,schuster2022many,sundar2022proposal,green2022experimental,xu2024scrambling,google2025observation}. 

\subsection{Vectorization}

The (normalized) vectorization map with respect to an orthogonal operator basis $\mathcal{Q}$ and an orthonormal state basis $\{\ket{k}\}_{k=1}^{4^n}$, defined as:
\begin{equation} \label{eq:vect_map}
    \kett{O}_\mathcal{Q} \equiv \sum_{k=1}^{4^n} \frac{c_k}{\sqrt{\sum_i \abs{c_i}^2}} \ket{k} \in \mathcal{H}',
\end{equation}
where $c_k$ are the operator amplitudes Eq.~(\ref{eq:op_amp_def}), is a bijective map between normalized $n$-qubit operators $O \in \mathcal{L}(\mathcal{H})$ to $2n$-qubit pure states $\kett{O} \in \mathcal{H}'$, where $\mathcal{H}$ and $\mathcal{H}'$ denote the $n$- and $2n$-qubit Hilbert spaces respectively. It is also an isometry, since:
\begin{equation} \label{eq:isometry}
    \leftindex_{\mathcal{Q}}{\bbrakett{O_1}{O_2}}_\mathcal{Q} =  \frac{\tr(O_1^\dagger O_2)}{\sqrt{\tr(O_1^\dagger O_1) \tr(O_2^\dagger O_2)}}.
\end{equation}
As a result, objects in $\mathcal{L}(\mathcal{H})$ can be mapped to those in $\mathcal{H}'$ and vice versa in a structure- and geometry-preserving way. Table~\ref{tab:comparison} of Appendix~\ref{app:background_prelim} shows examples of correspondences resulting from this mapping.

By definition, basis operators map to computational basis states (up to unit phase) when vectorized in the same basis, $\kett{Q_k}_\mathcal{Q} \propto \ket{k}$. For instance, when vectorized in the Pauli basis $\mathcal{P}$, i.e, when we take $\mathcal{Q} = \mathcal{P}$ in Eq.~(\ref{eq:vect_map}), single-qubit Pauli operators (and $\mathbb{I}$) map to computational basis states:
\begin{align*}
    \kett{{\mathbb{I}}}_\mathcal{P} = \ket{00}, &\quad \kett{X}_\mathcal{P} = \ket{01}, \\
    \kett{Z}_\mathcal{P} = \ket{10}, &\quad \kett{Y}_\mathcal{P} = -i\ket{11},
\end{align*}
but map to Bell states in the computational operator basis $\mathcal{C}$:
\begin{align*}
    \kett{{\mathbb{I}}}_\mathcal{C} = \frac{1}{\sqrt{2}}(\ket{00} + \ket{11}), &\quad \kett{X}_\mathcal{C} = \frac{1}{\sqrt{2}}(\ket{01} + \ket{10}), \\
    \kett{Z}_\mathcal{C} = \frac{1}{\sqrt{2}}(\ket{00} - \ket{11}), &\quad \kett{Y}_\mathcal{C} = \frac{i}{\sqrt{2}}(-\ket{01} + \ket{10}).
\end{align*}

It will also be convenient to order the doubled space $\mathcal{H}'$ spatially as $\mathcal{H}' =\mathcal{H}_L \otimes \mathcal{H}_R = (\bigotimes_{i=1}^n \mathcal{H}_L^i ) \otimes (\bigotimes_{i=1}^n \mathcal{H}_R^i )$, where $\mathcal{H}_L$ and $\mathcal{H}_R$ are each of dimension $2^n$, and each $\mathcal{H}_L^i$ or $\mathcal{H}_R^i$ has local dimension 2 (illustrated e.g. in Fig.~\ref{fig:doubled_space}). In $\mathcal{C}$, the useful `ricochet identity':
\begin{equation} \label{eq:ricochet}
    A \otimes B \kett{O}_\mathcal{C} = \kett{AOB^\mathsf{T}}_\mathcal{C}
\end{equation}
holds, with $A$ and $B$ acting on $\mathcal{H}_L$ and $\mathcal{H}_R$ respectively.

Due to the completeness of $\{\kett{Q_k}_\mathcal{Q}\}_{k=1}^{4^n} = \{\ket{k} \}_{k=1}^{4^n}$, one can always switch between bases of vectorization via:
\begin{equation} \label{eq:vect_basis_change}
    R_{\mathcal{Q}, \mathcal{Q}'} \kett{O}_\mathcal{Q} = \kett{O}_{\mathcal{Q}'},
\end{equation}
where $R_{\mathcal{Q}, \mathcal{Q}'}$ is a unitary that also satisfies:
\begin{align}
    R_{\mathcal{Q} , \mathcal{Q}'} &= (R_{\mathcal{Q}' , \mathcal{Q}})^\dagger, \\
    R_{\mathcal{Q}_1 , \mathcal{Q}_3} &= R_{\mathcal{Q}_2 , \mathcal{Q}_3} R_{\mathcal{Q}_1 , \mathcal{Q}_2}, \label{eq:vect_basis_composition}
\end{align}
for bases $\mathcal{Q}_1, \mathcal{Q}_2, \mathcal{Q}_3$. Importantly, the basis change unitary between $\mathcal{C}$ and $\mathcal{P}$ can be shown to be the Bell basis transformation:
\begin{equation} \label{eq:rotation_cp}
    R_{\mathcal{C},\mathcal{P}} = \bigotimes_{i=1}^n (\text{H}_{i_L} \cdot \text{CNOT}_{i_L,i_R}).
\end{equation}
It is a Clifford unitary, consisting of only the Hadamard and CNOT gates. Moreover, the CNOTs act on disjoint pairs $(i_L,i_R)$ and hence can be simultaneously applied.

\subsection{Superoperators and transfer matrices}

Denote a general linear superoperator as $\mathcal{A}(\cdot) \in \mathcal{L}( \mathcal{L}(\mathcal{H}^{n}))$. It can always be decomposed in the operator-sum representation in any basis $\mathcal{Q}$ as:
\begin{equation} \label{eq:superop_def}
    \mathcal{A}(\cdot) = \sum_{kl} f_{kl} Q_k (\cdot) Q_l^\dagger,
\end{equation}
where $f_{kl} \in \mathbb{C}$. It can also be associated with a unique linear operator $M^\mathcal{A}_\mathcal{Q}$ acting on $\kett{O}_\mathcal{Q}$, i.e.,:
\begin{equation} \label{eq:def_tm}
    \kett{\mathcal{A}(O)}_\mathcal{Q} = M^\mathcal{A}_\mathcal{Q} \kett{O}_\mathcal{Q},
\end{equation}
where $M^\mathcal{A}_\mathcal{Q} \in \mathbb{C}^{4^n \cross 4^n}$ is called the \textit{transfer matrix} of $\mathcal{A}$ in the basis $\mathcal{Q}$. It is the matrix representation of $\mathcal{A}$ in the basis $\mathcal{Q}$, with matrix elements:
\begin{equation} \label{eq:tm_elements}
    [M^\mathcal{A}_\mathcal{Q}]_{ij} = \frac{1}{N}\tr(Q_i^\dagger \mathcal{A}(Q_j)).
\end{equation}
When $\mathcal{Q} = \mathcal{P}$, it is also commonly called the Pauli Transfer Matrix (PTM) of $\mathcal{A}$.

Its basis can be changed by conjugating it with the same unitary appearing in Eq.~(\ref{eq:vect_basis_change}), i.e.
\begin{equation} \label{eq:TM_basis_change}
    M^\mathcal{A}_{\mathcal{Q}'} = R_{\mathcal{Q}', \mathcal{Q}}^\dagger M^\mathcal{A}_\mathcal{Q} R_{\mathcal{Q}', \mathcal{Q}}.
\end{equation}
Transfer matrices also inherit many properties of the superoperator such as self-adjointness, unitarity, and commutation.

The computational operator basis $\mathcal{C}$ is particularly useful for manipulating superoperators and transfer matrices, due to the identity Eq.~(\ref{eq:ricochet}). For instance, $\mathcal{A}$ has the decomposition Eq.~(\ref{eq:superop_def}) if and only if:
\begin{equation} \label{eq:tm_cb_identity}
    M^\mathcal{A}_{\mathcal{C}} = \sum_{kl} f_{kl} Q_k \otimes Q^*_l.
\end{equation}
As an example, consider the superoperator $\mathcal{U}(\cdot) \equiv U^\dagger(\cdot)U$ representing Heisenberg evolution by the unitary $U$, with transfer matrix in $\mathcal{Q}$ denoted $M^\mathcal{U}_\mathcal{Q}$, which is also unitary. In $\mathcal{C}$, Eq.~(\ref{eq:tm_cb_identity}) then implies that this unitary takes the simple product form: 
\begin{equation} \label{eq:tm_unitary}
    M^\mathcal{U}_\mathcal{C} = U^\dagger \otimes U^\mathsf{T},
\end{equation}
which is separable over $\mathcal{H}_L$ and $\mathcal{H}_R$.

\section{Algorithmic Framework} \label{sec:main_vect}

\begin{table*}[ht]
\renewcommand{\arraystretch}{1.4}
\setlength{\tabcolsep}{6pt}
\begin{tabular}{P{4.5cm}|P{5.75cm}|P{5.7cm}}
\thickhline
\textbf{Properties of $O(t)$} & \textbf{Properties of $\kett{O(t)}_\mathcal{Q}$} & \textbf{Method of extraction} \\
\thickhline
Operator distribution in $\mathcal{Q}$ & Probability distribution in computational basis & Measurement in computational basis \\
\hline
OTOCs / Statistical moments of self-adjoint $\mathcal{A}$, $\tr(O(t)^\dagger\mathcal{A}^k(O(t)))/N$ & Expectation value of Hermitian transfer matrices $M^\mathcal{A}_\mathcal{Q}$, $\leftindex_{\mathcal{Q}}{\bbra{O(t)} (M^{\mathcal{A}}_\mathcal{Q})^k \kett{O(t)}}_\mathcal{Q}$ & Measurement in eigenbasis of $M^\mathcal{A}_\mathcal{Q}$ and evaluation of Monte Carlo estimator; Section~\ref{sec:superop} \\
\hline
Operator amplitudes $\tr(Q_k^\dagger O(t))/N$, inner products $\tr(O_1^\dagger O_2)/N$ & Amplitude in computational basis, inner product between states & Amplitude/inner product estimation; Section~\ref{sec:2pc_main}\\
\hline
Operator Stabilizer Entropy (OSE) & Inverse participation ratio of probability distribution of $\kett{O(t)}_\mathcal{P}$ & Measurement and fidelity/amplitude estimation; Section~\ref{sec:ose} \\
\hline
Local Operator Entanglement (LOE) & Subsystem entanglement entropy of $\kett{O(t)}_\mathcal{Q}$ & Multivariate trace estimation; Section~\ref{sec:loe} \\

\thickhline
\end{tabular}
\caption{\textbf{Summary of quantities that can be computed under our framework.} Summary of properties of $O(t)$, the corresponding properties encoded in $\kett{O(t)}_\mathcal{Q}$ as a result of the operator-state mapping, and methods for their extraction on a quantum computer. See also Table~\ref{tab:algorithms_index} for an index of introduced algorithms.}
\label{tab:features_Ot}
\end{table*}

Central to our approach is the vectorization map of Eq.~(\ref{eq:vect_map}), which allows $n$-qubit operators $O(t) = U^\dagger O U$ (of unit HS norm) to be encoded as $2n$-qubit pure states $\kett{O(t)}_\mathcal{Q}$ in the doubled Hilbert space $\mathcal{H}'$. By preparing $\kett{O(t)}_\mathcal{Q}$ on a quantum computer or simulator, the encoded Heisenberg operator $O(t)$ is realized as a physical object that can be manipulated and studied. Information on $O(t)$ contained in $\kett{O(t)}_\mathcal{Q}$ (and variations of this encoding, depending on the task at hand) can subsequently be extracted via measurements and other protocols. This approach is summarized in Fig.~\ref{fig:main}. 

In the following, we discuss these steps in detail. Throughout our work, we take $O$ to be Hermitian observables of unit HS norm, allowing us to drop the dagger in expressions such as Eqs.~(\ref{eq:op_amp_def}) and (\ref{eq:def_otoc}).

\subsection{Preparation of encoded operators} \label{sec:preparation}

In this subsection we describe how $\kett{O(t)}_\mathcal{Q}$ can be prepared on a quantum computer. Before describing the preparation of $\kett{O(t)}_\mathcal{Q}$ for an arbitrary basis $\mathcal{Q}$, we begin by presenting Pauli-basis vectorization $\mathcal{Q} = \mathcal{P}$ as this is perhaps simpler to follow and relevant to most applications we will discuss. It is possible to prepare $\kett{O(t)}_\mathcal{P}$ via (1) preparing the $\kett{O}_\mathcal{P}$ in the Pauli basis, (2) rotating from the Pauli basis to the computation basis, (3) implementing the time evolution in the computational basis and then (4) transforming back to the Pauli basis: 
\begin{equation*}
    \ket{0...0} ~\xlongrightarrow{(1)}~ \kett{O}_\mathcal{P} ~\xlongrightarrow{(2)}~ \kett{O}_\mathcal{C} ~\xlongrightarrow{(3)}~ \kett{O(t)}_\mathcal{C} ~\xlongrightarrow{(4)}~ \kett{O(t)}_\mathcal{P} .
\end{equation*}
These steps are schematically summarized in the circuit diagram of Fig.~\ref{fig:main}. More concretely these steps are implemented as follows. 

Step (1) is the preparation of $\kett{O}_\mathcal{P}$, starting from a product state such as $\ket{0...0}$. When $O$ is a single Pauli operator $\kett{O}_\mathcal{P}$  is simply a computational basis state. For example, from Eq.~(\ref{eq:vect_map}), if $O = X$, then $\kett{O}_\mathcal{P} = \ket{01}$. More generally, when $O$ is a linear combination of $s$ Pauli operators, we have that:
\begin{equation*}
    O = \sum_{k=1}^s c_k P_k \longrightarrow \kett{O}_\mathcal{P} = \sum_{k=1}^s c_k \ket{k} \, .
\end{equation*}
This state can be prepared using existing general state preparation algorithms that take $\bigo{ns}$ resources \cite{gleinig2021efficient,zhang2022quantum,li2024nearly,vilmart2025resource}, which is efficient when $s = \bigo{\text{poly}(n)}$. 

In general we expect the time evolution to be simpler to implement in the computational basis and so in step (2) we transform from the Pauli basis $\mathcal{P}$ into the computational basis $\mathcal{C}$.
This is carried out by the unitary $R_{\mathcal{P},\mathcal{C}}$ which is simply the Bell basis transformation shown in Eq.~(\ref{eq:rotation_cp}), and can easily be implemented as a local, transversal Clifford circuit in depth 2.

Step (3) carries out time evolution in the computational basis $\mathcal{C}$. It is implemented by the transfer matrix of $\mathcal{U}$ in $\mathcal{C}$, which takes the tensor product form Eq.~(\ref{eq:tm_unitary}). It can be implemented by applying $U^\dagger$ and $U^\mathsf{T}$ in $\mathcal{H}_L$ and $\mathcal{H}_R$ respectively, in parallel. On a digital quantum computer where $U$ is decomposed as a sequence of elementary gates $U = \prod_{l=1}^L U_l$, this is achieved by reversing the order of the gates and implementing $U_l^\dagger \otimes U_l^\mathsf{T}$: i.e.,
\begin{equation}
\prod_{l=1}^L U_l \longrightarrow \prod_{l=L}^1  U_l^\dagger \otimes U_l^\mathsf{T} \, . 
\end{equation}
We note that $U^\dagger \otimes U^\mathsf{T}$ can be implemented in the same depth as $U$ and also that queries to the transposed unitary $U^\mathsf{T}$ can be avoided, at the price of having to implement $U$, $O$, and $U^\dagger$ sequentially -- see Appendix~\ref{app:details_state_prep} for more details. Finally, in Step (4) we rotate the state back into the Pauli basis using $R_{\mathcal{C},\mathcal{P}}$ - this is just the inverse of the Bell transformation again.

To generalize the above steps to arbitrary $\mathcal{Q}$, we (1) start by preparing $\kett{O}_{\mathcal{Q}_i}$ in a basis in which $O$ is sparse, (2) transform the basis of vectorization to $\mathcal{C}$, (3) carry out time-evolution in the computational basis, and finally (4) transform to the desired basis $\mathcal{Q}$. This is summarized as:
\begin{equation} \label{eq:state_prep}
    \kett{O(t)}_\mathcal{Q} = R_{\mathcal{C}, \mathcal{Q}} M^\mathcal{U}_{\mathcal{C}} R_{\mathcal{Q}_i, \mathcal{C}} \kett{O}_{\mathcal{Q}_i} \, .
\end{equation}
The preparation of $\kett{O(t)}_\mathcal{Q}$ can thus always be composed of the following steps:
\begin{enumerate}[label=(\roman*)]
    \item Preparation of initial encoded operator $\kett{O}_{\mathcal{Q}_i}$ in a basis $\mathcal{Q}_i$.
    \item Vectorization basis change between $\mathcal{Q}_i$ and $\mathcal{C}$ ($\kett{O}_{\mathcal{C}} = R_{\mathcal{Q}_i,\mathcal{C}} \kett{O}_{\mathcal{Q}_i}$), and between $\mathcal{C}$ and $\mathcal{Q}$ ($\kett{O(t)}_{\mathcal{Q}} = R_{\mathcal{C}, \mathcal{Q}}, \kett{O(t)}_{\mathcal{C}}$).
    \item Time evolution $\kett{U^\dagger O U}_{\mathcal{C}} = M^U_{\mathcal{C}} \kett{O}_{\mathcal{C}}$ in $\mathcal{C}$.
\end{enumerate}
Each step is efficient, provided that respectively (i) the initial operator $O$ is sparse in $\mathcal{Q}_i$, (ii) the initial and final bases $\mathcal{Q}_i$ and $\mathcal{Q}$ have tensor product structures over small partitions, and (iii) the implementations of $U^\dagger$, and $U^\mathsf{T}$ or $U$ are efficient -- this is shown precisely in Appendix~\ref{app:details_state_prep}.

Beyond depth preservation, we further remark on the following facts regarding $M^\mathcal{U}_\mathcal{C}$:
\begin{itemize}
    \item It inherits the connectivity/entanglement structure of $U$ over spatial indices: trivially, each gate $U_l^\dagger \otimes U^\mathsf{T}_l$ is only supported on $\bigotimes_{i \in \text{supp}(U_l)} \mathcal{H}^i_L \otimes \mathcal{H}^i_R$.

    \item It inherits the Cliffordness of $U$: $U$ being Clifford implies that both $U^\dagger$ and $U^\mathsf{T}$, and therefore $U^\dagger \otimes U^\mathsf{T}$ is also Clifford (as Clifford gates are closed under inversion and complex-conjugation).
\end{itemize}
Heisenberg time evolution therefore inherits the resource requirements of simulating $U^\dagger$, and $U$ or $U^\mathsf{T}$ via our framework. 

In Section~\ref{sec:implementation_proposal}, we further detail an implementation of this procedure in practice, for the simulation of a spin model defined on a 2D lattice. Taking into account realistic hardware connectivity constraints, we illustrate the practical amenability of our approach on current-generation hardware.

Finally, beyond the main vectorization encoding $O \rightarrow \kett{O}_\mathcal{Q}$, we will also study an alternative but related $n$-qubit encoding $\rho(O)$ that requires only $n$ qubits. This is detailed in Section~\ref{sec:n_qubit_rand} and Appendix~\ref{app:n_qubit_details}.

\subsection{Sampling and extracting information} \label{sec:extracting_information}

Having prepared the state $\kett{O(t)}_\mathcal{Q}$, one can now begin to extract information regarding $O(t)$ from it. Importantly, Eq.~(\ref{eq:vect_map}) enables tasks or properties involving Heisenberg operators to be mapped to those involving quantum states (and vice versa). By effectively realizing this mapping on a quantum simulator, we gain access to Heisenberg operators $O(t)$ via the state $\kett{O(t)}_\mathcal{Q}$. This central observation allows us to apply existing quantum algorithms and subroutines to extract various properties in the Heisenberg picture, summarized in Table~\ref{tab:features_Ot}; see also Table~\ref{tab:algorithms_index} for an index of introduced algorithms. Under this lens, it also becomes clear that certain existing approaches for the computation of quantities such as the OTOC and operator geometrical properties \cite{nahum2018operator,roberts2018operator,qi2019quantum} reduce to instances of this protocol. The relation between our framework and existing methods are discussed in Appendix~\ref{app:relation_existing}.

A particularly important subroutine is the ability to directly sample from the operator distribution $p(Q_k, O)$. This is a direct consequence of Eq.~(\ref{eq:vect_map}), which tells us that the wavefunction amplitudes/probabilities of $\kett{O(t)}_{\mathcal{Q}}$ in the computational basis encode the operator amplitudes/probabilities of $O(t)$ in $\mathcal{Q}$. Among other uses, sampling from $p(Q_k, O)$ enables the estimation of any quantity of the form:
\begin{alignb} \label{eq:computable_general_form}
    F(O(t)) &= \sum_k p(Q_k, O) f(Q_k, O(t)) \\
    &= \mathop{\mathbb{E}}_{k \sim p(Q_k, O)}[f(Q_k, O(t))]
\end{alignb}
via Monte Carlo estimation, where $f(Q_k, O(t)) \in \mathbb{R}$ is a general non-linear function of $O(t)$ and/or $Q_k$ to be computer on either a classical or a quantum computer. This observation allows us to construct estimators for various quantities naturally defined in the Heisenberg picture via the empirical average $F(O(t)) \approx \frac{1}{M} \sum_{i=1}^M f(Q_i, O(t))$, using a small number of samples $M$ that scales only polynomially with system size $n$, as long as $f(Q_k, O(t))$ can be evaluated efficiently (either on a classical or quantum computer), and the variance of this estimator is bounded.

In fact, we will see that the case where $f(Q_k, O(t)) = f(Q_k)$ -- i.e., when $f$ is a function of $Q_k$ but not $O(t)$ -- coincides with the expectation values of self-adjoint superoperators $\langle\mathcal{A}\rangle_O$, which map to the expectation values of observables on $\kett{O(t)}_\mathcal{Q}$, Eq.~(\ref{eq:def_superop_expt}). Furthermore, as vectorization preserves the notions of self-adjointness and commutation, we can exploit the basic but powerful observation that expectation values of commuting observables can be simultaneously estimated. This enables the simultaneous estimation of many superoperator expectation values $\{\langle\mathcal{A}_i\rangle_O\}_i$ (Theorem~\ref{theorem:simul_otoc}). 

The Pauli basis $\mathcal{P}$ is especially useful, since many physically relevant properties can be written in the form of Eq.~(\ref{eq:computable_general_form}) with $\mathcal{Q} = \mathcal{P}$. For instance, the OSE takes this form, with $f(P_k, O(t))$ a non-linear function of both $P_k$ and $O(t)$ (Theorem \ref{theorem:ose}). Also, as we will discuss in Section~\ref{sec:superop}, the case where $f(P_k, O(t))=f(P_k)$ is a function of $P_k$ alone coincides with \textit{all} possible geometrical properties of $O(t)$. It is therefore convenient to generate and store a compact empirical Pauli distribution $p(P_k, O(t)) \approx \frac{\text{no. of samples of $k$}}{{\text{total no. of samples}}}$ which can subsequently be used in an \textit{a posteriori} manner to simultaneously compute the quantities $\{F_1(O(t)), ..., F_M(O(t))\}$ via Monte Carlo, where $F_i(O(t))$'s are properties of $O(t)$ taking the form of Eq.~(\ref{eq:computable_general_form}).

The above discussion sets the stage for Sections~\ref{sec:ose}--\ref{sec:2pc_main}, where we detail applications of our framework for a number of tasks, and propose novel/improved algorithms that accomplish them. Subsequently, Section~\ref{sec:implementation_proposal} discusses concrete proposals for their implementation on current-generation quantum computers. Finally, in Section~\ref{sec:n_qubit_rand}, we study an alternative $n$-qubit approach able to solve a subset of the previously described tasks, and its connection to the main framework. Further details are delegated to the Appendices, including generalizations to finite temperature (App.~\ref{app:finite_temp}), open-system dynamics (App.~\ref{app:channel_gen}), and qudit systems (App.~\ref{app:qudit}).

\section{Algorithm for Operator Stabilizer Entropies} \label{sec:ose}

The $\alpha$-Operator Stabilizer Entropy (OSE) of $O$ \cite{dowling2025magic} is defined as the $\alpha$-Rényi entropy of its Pauli distribution:
\begin{alignb} \label{eq:def_ose}
    M^{(\alpha)}(O) &\equiv \frac{1}{1-\alpha} \log \left[ \sum_{k=1}^{4^n} \left( \frac{\text{tr}(P_k O)}{2^n} \right)^{2 \alpha} \right],
\end{alignb}
where $\alpha \geq 0$. As a magic monotone, it quantifies magic resources generated by non-Clifford dynamics, and is the operator analogue of the state-based stabilizer Rényi entropy \cite{leone2022stabilizer}. The term within the logarithm, the operator stabilizer purity:
\begin{equation}
    P^{(\alpha)}(O) \equiv \sum_{k=1}^{4^n} \left( \frac{\text{tr}(P_k O)}{2^n} \right)^{2 \alpha}
\end{equation}
can be identified as a linearized entropy \cite{haug2024efficient,dowling2025magic}. 

To compute $P^{(\alpha)}$, observe that it can be viewed as the inverse participation ratio of the probability distribution of the state $\kett{O(t)}_\mathcal{P}$, when measured in the computational basis. Denoting the Pauli probabilities by the shorthand $p_k =  \left({\tr(P_kO(t))}/{2^n} \right)^{2}$ and writing $P^{(\alpha)} = \sum_k p_k ~ p^{\alpha-1}_k$ for integer-valued $\alpha > 1$, we observe immediately that $P^{(\alpha)}$ can be written in the Monte Carlo form Eq.~(\ref{eq:computable_general_form}) with $f(P_k,O(t)) \equiv p^{\alpha-1}_k$. Given access to samples of $\kett{O(t)}_\mathcal{P}$, we can therefore compute $P^{(\alpha)}$ via a Monte Carlo algorithm that makes use of the abilities to (1) sample from the Pauli distribution of $O(t)$, and (2) estimate $p^{\alpha-1}_k$ through estimators of the Pauli probabilities $p_k = \abs{\!\! ~_{\mathcal{P}} \langle \! \langle P_k \kett{O(t)}_\mathcal{P}}^2$ such as the SWAP test.

The above discussion yields the following quantum algorithm, which is the operator analogue of existing Schrödinger-picture quantum algorithms \cite{haug2024efficient}:
\begin{enumerate}
    \item Perform independent computational basis measurements on $M \geq 2\log(4/\delta)/\epsilon^2$ copies of $\kett{O(t)}_\mathcal{P}$, yielding measurement results $\{k_1, ..., k_M\}$.

    \item For each $i = 1,..., M$, estimate the $(\alpha-1)^{\text{th}}$ power of the fidelities $p^{\alpha-1}_{k_i}$ via SWAP tests between $\kett{O(t)}_\mathcal{P}$ and $\ket{k}$, using $N \geq 2(\alpha-1)\log(4/\delta)/\epsilon^2$ samples of $\kett{O(t)}_\mathcal{P} \otimes \ket{k}$ in total, yielding the estimates $\{\hat{p}^{\alpha-1}_{k_1}, ..., \hat{p}^{\alpha-1}_{k_M}\}$. 

    \item Output the empirical average $\hat{P}^{(\alpha)} = \frac{1}{M} \sum_{i=1}^M \hat{p}^{\alpha-1}_{k_i}$.
\end{enumerate}
We therefore have the following result:
\begin{theorem}[$\alpha$-stabilizer purities] \label{theorem:ose}
There is an algorithm that computes the $\alpha$-stabilizer purity $P^{(\alpha)}(O)$ (where $\alpha \in \mathbb{N}$) to additive error $\epsilon$ and success probability at least $1-\delta$, using $N_P=\bigo{\alpha \log(1/\delta)/\epsilon^2}$ samples of $\kett{O}_\mathcal{P}$, measured in the computational basis.
\end{theorem}

Details of the algorithm and the proof of Theorem~\ref{theorem:ose} are delegated to Appendix~\ref{app:ose}, where we also describe alternatives of the estimator of $p_{k_i}^{\alpha-1}$ (replacing step 2), such as a variant of the Hadamard test that takes only $2n+1$ qubits and queries $U^\dagger \otimes U^\mathsf{T}$.

Due to the logarithm, the estimation of the OSE $M^{(\alpha)}$ to additive precision $\epsilon$ via the above algorithm translates to a precision of order $\bigo{e^{-M^{(\alpha)}}\epsilon}$ in $P^{(\alpha)}$, resulting in a sample complexity $N_M=\bigo{\alpha e^{M^{(\alpha)}}\log(1/\delta)/\epsilon^2}$, scaling exponentially with $M^{(\alpha)}$. The computation of $M^{(\alpha)}$ is therefore only efficient when $M^{(\alpha)} = \bigo{\log(n)}$, an expected limitation shared with algorithms for (Schrödinger-picture) stabilizer/entanglement entropies \cite{tarabunga2023many,haug2024efficient,quek2024multivariate,liu2025quantum}.

When one wishes to track the value of $M^{(\alpha)}$ at a specific time $t^*$ (e.g., at long times or in infinite-time averages after transient dynamics), we expect our approach to outperform classical methods whenever the dynamics exhibits \textit{magic backflow}. By this we mean scenarios in which the OSE of $O(t)$ grows to $M^{(\alpha)} = \mathcal{O}(n)$ at intermediate times before returning to $\mathcal{O}(\log n)$ at the target time $t^*$. Backflow phenomena have been extensively studied in the context of operator spreading~\cite{von2018operator, von2022operator}, and since global operators are particularly susceptible to generating high magic, these analyses naturally extend to the magic setting.

In particular, our methods depend only on the amount of magic at the desired time $t^*$, whereas classical approaches must keep track of the full time evolution of the Heisenberg operators over the entire interval from $0$ to $t^*$. In general, classical techniques such as Pauli~\cite{rall2019simulation,aharonov2022polynomial, beguvsic2023simulating, fontana2023classical, shao2023simulating, rudolph2023classical, schuster2024polynomial, angrisani2024classically, gonzalez2024pauli, lerch2024efficient, cirstoiu2024fourier, angrisani2025simulating, fuller2025improved, rudolph2025pauli, angrisani2025simulating, teng2025leveraging, rudolph2026thermal} and Majorana~\cite{miller2025simulation,alam2025fermionic,alam2025programmable,d2025majorana, facelli2026fast} propagation are provably efficient only for restricted classes of noiseless circuits—for example, those that maintain consistently low magic throughout the evolution~\cite{lerch2024efficient,shao2025pauli,d2025majorana}, or certain scrambling circuits for specific tasks such as estimating expectation values~\cite{angrisani2024classically}. For more general tasks, such as approximating the Heisenberg-evolved observable in operator norm, it has been shown that one must retain an exponential number of Pauli strings in $O(t)$~\cite{dowling2025magic}. Consequently, Heisenberg operators with extensive OSE cannot be simulated both faithfully and efficiently using Pauli propagation.

\section{Algorithm for Local Operator Entanglement} \label{sec:loe}
The Local Operator Entanglement (LOE) of an operator $O$ \cite{prosen2007operator} is the operator-space analogue of the entanglement entropy for quantum states. Defined as the $\alpha$-Rényi entropy of the (squared) operator Schmidt coefficients of $O$ across a spatial bipartition, it has been studied in a variety of contexts \cite{prosen2007operator,dubail2017entanglement}, and determines the classical simulatability of operators via tensor network methods \cite{prosen2007operator,jonay2018coarse}.

\begin{figure}[]
    \centering
    \includegraphics[width=.6\linewidth]{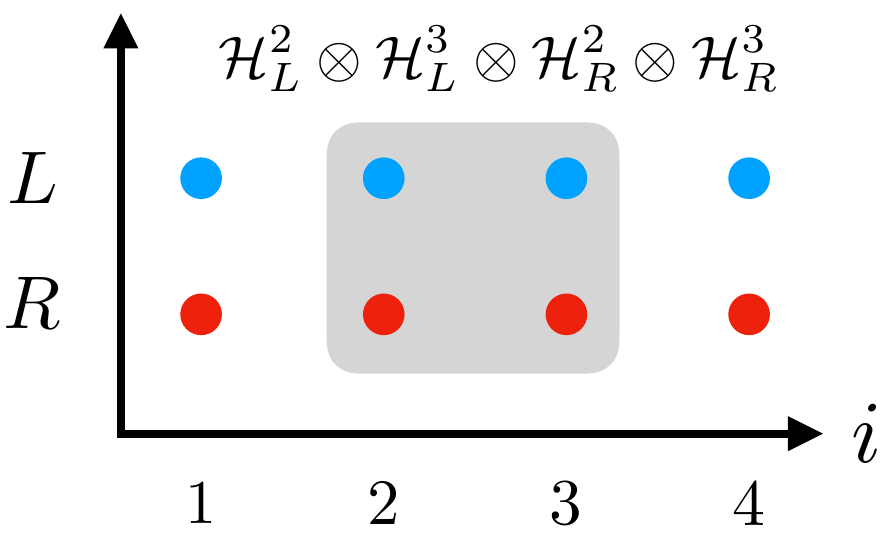}
    \caption{\textbf{Schematic of $\mathcal{H}'$ for $n=4$}. The labeled $4$-qubit subsystem (mapping to a $2$-site operator) is enclosed in gray.}
    \label{fig:doubled_space}
\end{figure}

The LOE of $O$ can be defined naturally through its vectorization $\kett{O}_\mathcal{Q}$ (in any basis $\mathcal{Q}$). Let $A:B$ denote a bipartition of $\{1,...,n\}$, and define the reduced density matrix:
\begin{equation*}
    \rho_A \equiv \tr_{B}\left(\kett{O}_\mathcal{Q} \leftindex_{\mathcal{Q}}{\bbra{O}}\right) \in \bigotimes_{i\in A} (\mathcal{H}^i_L \otimes \mathcal{H}^i_R),
\end{equation*}
obtained by tracing out the subsystem $B \equiv \bigotimes_{i\in B} (\mathcal{H}^i_L \otimes \mathcal{H}^i_R)$, illustrated in Fig.~\ref{fig:doubled_space} as the qubits highlighted in gray for the choice $A=\{2,3\}$, $B=\{1,4\}$. The $\alpha$-LOE is then the Rényi entropy of $\rho_A$:
\begin{equation}
    E^\alpha_A(O) \equiv \frac{1}{1-\alpha} \log(\tr(\rho_A^\alpha)),
\end{equation}
where $\alpha \geq 0$. Within the logarithm is the $\alpha$-purity $\tr(\rho_A^\alpha)$, which also defines the linearized LOE $E^\alpha_{\text{lin}, A}(O) \equiv 1-\tr(\rho_A^\alpha)$. We consider integer-valued $\alpha$ in this work.

Through $\kett{O}_\mathcal{Q}$, $E^\alpha_{\text{lin}, A}(O)$ can be estimated via existing quantum algorithms for purities and entropies \cite{ekert2002direct,johri2017entanglement,elben2018renyi,huang2020predicting,rath2021quantum,quek2024multivariate}. In particular, for $\alpha=2$ we have:
\begin{theorem}[Linear entropy]\label{theorem:loe_combined}
Consider the task of estimating $E^2_{\mathrm{lin},A}(O)$ to additive error $\epsilon$ with success probability at least $1-\delta$. There exists a quantum algorithm that accomplishes this using $\bigo{\log(1/\delta)/\epsilon^2}$ copies of the state $\kett{O(t)}_{\mathcal Q}$.
\end{theorem}

Concretely, the algorithm performs the destructive SWAP test \cite{barenco1997stabilization, ekert2002direct} on $N$ independent pairs of $\kett{O(t)}_{\mathcal Q}$: for each trial, subsystem $A$ of the pair is measured in the Bell basis via the transversal Clifford unitary of Eq.~(\ref{eq:rotation_cp}). Classical post-processing of the outcomes then produce an estimate of $E^\alpha_{\mathrm{lin},A}(O)$. It requires entangled measurements across $\leq 4n$ qubits, and attains a sample complexity that is independent of partition size. It can be generalized to larger integers $\alpha$ by generalizing the SWAP to a permutation involving $\alpha$ copies of $\kett{O}_\mathcal{Q}$ \cite{brun2004measuring,johri2017entanglement}. 

This task can also be adapted for connectivity-constrained hardware via a constant-depth multivariate trace estimation algorithm \cite{quek2024multivariate}, which takes $N$ samples of the state $V' \kett{O(t)}_\mathcal{Q}^{\otimes 2} \ket{0}^{\otimes \lfloor \alpha/2 \rfloor}$, where $V'$ is a constant-depth circuit consisting of $\bigo{n}$ three-qubit gates. This renders the estimation of the LOE feasible on connectivity-constrained devices such as those possessing a square grid topology; this point is further discussed in Section~\ref{sec:implementation_proposal}.

Similar to the OSE, taking the logarithm yields $\alpha$-Rényi entropies $E^\alpha_A$, but leads to an exponential dependence on $E^\alpha_A$, rendering its estimation to additive precision expensive unless $E^\alpha_A = \bigo{\log(n)}$ \cite{tarabunga2023many,haug2024efficient,quek2024multivariate,liu2025quantum}. Mirroring the discussion of Section~\ref{sec:ose} on classical simulatability in low-magic conditions, we expect our approach for the estimation of the LOE to outperform classical simulation methods that are inefficient under dynamics that generate extensive entanglement, such as those based on tensor networks \cite{prosen2007operator,jonay2018coarse,hemery2019matrix,xu2020accessing}.

\section{Algorithms for OTOCs and Superoperator expectation values} \label{sec:superop}
Coarse-grained statistical information contained in $O$ encompasses quantities such as the OTOC, average operator size, boundary position, and so on, which are important probes of information scrambling and operator growth. In this section, we exploit the vectorization map for their computation on quantum computers. 

Working in a suitably chosen basis -- the Pauli basis -- we show that the ability to draw samples from the Pauli distribution of an operator enables the \textit{simultaneous} extraction of \textit{all} its geometrical properties (to be defined precisely), such as diagonal OTOCs, size, boundary positions, propagation velocities, and so on (visualized in Fig.~\ref{fig:lightcone_geometry}). More generally, by mapping the problem of computing superoperator expectation values to the usual expectation estimation problem on quantum computers, we exploit the commutation between superoperators for the simultaneous estimation of many expectation values at once, which improves upon existing approaches requiring individual estimation \cite{google2025observation,zhang2025quantum}.

The results in this section admit natural generalizations to finite-temperature \cite{viswanath1994recursion,maldacena2016bound,xu2024scrambling} and open-system \cite{swingle2018resilience,zhang2019information,zanardi2021information,schuster2023operator} situations. They are described in Appendices~\ref{app:finite_temp} and \ref{app:channel_gen} respectively.

\subsection{Self-adjoint superoperators and OTOCs} \label{subsec:self-adjoint}
Self-adjoint superoperators are superoperators satisfying $\mathcal{A} = \mathcal{A}^\dagger$. Their transfer matrices $M^\mathcal{A}_\mathcal{Q}$ inherit this property, i.e., they are Hermitian matrices $M^\mathcal{A}_\mathcal{Q} = (M^\mathcal{A}_\mathcal{Q})^\dagger ~ \forall \mathcal{Q}$. They also inherit the notion of commutativity (c.f. Appendix~\ref{app:superop}); defining the commutator between superoperators $\mathcal{A}, \mathcal{B}$ to be the superoperator $[\mathcal{A}, \mathcal{B}](\cdot) \equiv (\mathcal{A}\mathcal{B}-\mathcal{B}\mathcal{A})(\cdot)$, we have, $\forall \mathcal{Q}$,
\begin{equation} \label{eq:commutation_inheritance}
    [\mathcal{A}, \mathcal{B}] = 0 \iff [M^\mathcal{A}_\mathcal{Q}, M^\mathcal{B}_\mathcal{Q}] = 0.
\end{equation}

Due to the above properties, $\mathcal{A}$ can be interpreted as a `superobservable' that encodes coarse-grained properties of Heisenberg operators $O(t)$. That is, analogous to Schrödinger-picture situations where expectation values $\langle A \rangle_\rho = \tr(A \rho)$ of self-adjoint observables $A$ over quantum states $\rho$ carry coarse-grained information, we define expectation values over self-adjoint superoperators $\mathcal{A}$ as:
\begin{equation} \label{eq:def_superop_expt}
    \langle \mathcal{A}\rangle_O \equiv \frac{1}{2^n}\tr(O^\dagger \mathcal{A}(O)).
\end{equation}
Similarly, the variance of $\mathcal{A}$ is defined as $\textup{Var}(\mathcal{A})_O \equiv \langle \mathcal{A}^2\rangle_O - \langle \mathcal{A}\rangle_O^2$, and the $k^{\text{th}}$ moment as $\langle \mathcal{A}^k \rangle_O = \tr(O^\dagger \mathcal{A}^k(O))/2^n$.

Subsequently, the vectorization map allows $\langle \mathcal{A}^k \rangle_{O(t)}$ (where $k \in \mathbb{N}$) to be expressed as:
\begin{equation} \label{eq:super_expt_vect}
    \langle \mathcal{A}^k \rangle_{O(t)} = \leftindex_{\mathcal{Q}}{\bbra{O(t)} (M^{\mathcal{A}}_\mathcal{Q})^k \kett{O(t)}}_\mathcal{Q},
\end{equation}
where $\mathcal{Q}$ can be suitably chosen. In other words, $\langle \mathcal{A}^k \rangle_{O(t)}$ can be accessed on a quantum computer by measuring $\kett{O}_\mathcal{Q}$ in the eigenbasis of $M^\mathcal{A}_\mathcal{Q}$. This reduces the computation of $\langle \mathcal{A}^k \rangle_O$ on quantum computers to the usual expectation estimation problem, and together with commutation, allows us to readily make use of existing results and strategies such as term-grouping, measurement optimization, and so on \cite{wecker2015progress,arrasmith2020operator,hamamura2020efficient,crawford2021efficient}. Indeed, following the discussion of Section~\ref{sec:extracting_information}, they coincide with quantities in the Monte Carlo form Eq.~(\ref{eq:computable_general_form}) with $f(Q_k,O(t)) = f(Q_k)$ only depending on $Q_k$.

The simplest instance of a superoperator expectation value is the OTOC, Eq.~(\ref{eq:def_otoc}). Writing: 
\begin{align*}
    \text{OTOC}(O(t), Q_k, Q_l) &\equiv \frac{1}{2^n}\tr(O(t)^\dagger Q_k^\dagger O(t) Q_l) \\
    &= \leftindex_{\mathcal{C}}{\bbra{O(t)} Q_k \otimes Q_l^\mathsf{T} \kett{O(t)}}_\mathcal{C}  \\
    &= \langle \mathcal{A}_{kl} \rangle_{O(t)},
\end{align*}
where the ricochet identity Eq.~(\ref{eq:ricochet}) and the cyclicity of the trace was used throughout, the OTOC takes the form of the expectation value of the superoperators $\mathcal{A}_{kl}(\cdot) = Q^\dagger_k (\cdot) Q_l$ (consisting of a single term in its operator-sum decomposition). We observe that it can be directly computed by measuring $\kett{O(t)}_\mathcal{C}$ in the eigenbasis of $Q_k \otimes Q_l$. If $Q_k$ and $Q_l$ are Pauli operators, it suffices to perform local Pauli measurements on $\kett{O}_\mathcal{C}$. This procedure requires $\bigo{1/\epsilon^2}$ samples of $\kett{O}_\mathcal{C}$ to achieve additive error $\epsilon$. Higher-order OTOCs of the form $\tr((QO(t)^k))/2^n$ \cite{roberts2017chaos,vallini2024long,fava2025designs,dowling2025free} can also be computed by applying the unitary $QO(t)$ repeatedly before measurement, similar to existing protocols \cite{google2025observation}.

In this language, it also becomes clear that we can exploit the notion of superoperator commutativity, Eq.~(\ref{eq:commutation_inheritance}), to simultaneously estimate OTOCs by measuring in their common eigenbasis. An application of Hoeffding's inequality and the union bound then immediately yields:
\begin{theorem}[Commuting superoperators]\label{theorem:simul_otoc}
Consider the task of simultaneously estimating $M$ OTOCs $\{ \tr(O(t) P_{i} O(t) Q_{i})/2^n\}_{i=1}^M$ to error $\epsilon$ and success probability at least $1-\delta$. Further suppose that $[P_i \otimes Q_i, P_j \otimes Q_j] = 0$ for all pairs $(i,j)$. There is an algorithm that achieves this using $\bigo{\log(M/\delta)/\epsilon^2}$ samples of $\kett{O(t)}_\mathcal{C}$, measured in the common eigenbasis of $\{P_i \otimes Q_i\}_{i=1}^M$.
\end{theorem}

The structure of the common eigenbasis of Theorem~\ref{theorem:simul_otoc} generally depends on the details of the set $\{P_i \otimes Q_i\}_{i=1}^M$. In Appendix~\ref{app:general_commutation}, we further show that such eigenbases can be systematically classified as (i) those requiring entanglement between $\mathcal{H}_R$ and $\mathcal{H}_L$, and (ii) those that do not, depending on the local structure of their commutation/anticommutation (Lemma~\ref{lemma:anticomm_basis}). As a consequence, OTOCs belonging to the former case can always be jointly computed using the $n$-qubit algorithm to be described in Section~\ref{sec:n_qubit_rand}, providing a connection between the two classes of algorithms. Finally, beyond grouping superoperator terms $\mathcal{A}_i$'s into commuting groups, it is also possible to group them into anticommuting groups (c.f. Lemma~\ref{app:anticommutation_condition}), in which case we can make use of unitary grouping strategies for measurement reduction \cite{izmaylov2019unitary,zhao2020measurement}.

Repeating the protocol of Theorem~\ref{theorem:simul_otoc} but with the order of $U^\dagger$ and $U$ interchanged to prepare $\kett{O(-t)}$, we immediately also obtain (due to the cyclicity of trace):
\begin{corollary} \label{corr:backward_otoc}
Consider the task of simultaneously estimating $M$ OTOCs $\{ \tr(O P_{i}(t) O Q_{i}(t))/2^n\}_{i=1}^M$ to error $\epsilon$ and success probability at least $1-\delta$. Further suppose that $[P_i \otimes Q_i, P_j \otimes Q_j] = 0$ for all pairs $(i,j)$. There is an algorithm that achieves this using $\bigo{\log(M/\delta)/\epsilon^2}$ samples of $\kett{O(-t)}_\mathcal{C}$, measured in the common eigenbasis of $\{P_i \otimes Q_i\}_{i=1}^M$.
\end{corollary}

Next, consider a general self-adjoint superoperator $\mathcal{A}$ of the form Eq.~(\ref{eq:superop_def}), with $M$ terms in its decomposition (that need not commute). Working in the basis $\mathcal{C}$ and applying Eq.~(\ref{eq:tm_cb_identity}) then yields:
\begin{equation} \label{eq:superop_term_by_term}
    \langle \mathcal{A} \rangle_{O(t)} = \sum_{kl}^M f_{kl} \leftindex_{\mathcal{C}}{\bbra{O(t)} Q_k \otimes Q_l^* \kett{O(t)}}_\mathcal{C}.
\end{equation}
The most straightforward approach to estimate $\langle \mathcal{A} \rangle_{O(t)}$ is to estimate it term-by-term, by measuring $\kett{O}_\mathcal{C}$ in the eigenbases of each $Q_k \otimes Q_l$ term. This requires $N \propto \left( \sum_{kl} \abs{f_{kl}} \sqrt{\text{Var}(\mathcal{A}_{kl})_{O(t)}} \right)^2$ samples, via the optimal measurement allocation strategy that minimizes error \cite{crawford2021efficient}.

However, if some terms mutually commute, they can be simultaneously estimated from a single set of measurements in their common eigenbasis. That is, group terms as:
\begin{equation} \label{eq:superoperator_term_group}
    \mathcal{A} = \sum_{i=1}^K \mathcal{A}^i = \sum_{i=1}^K \sum_{j=1}^{K_i} \mathcal{A}_{j}^i 
\end{equation}
so that each $\mathcal{A}^i$ only contains mutually commuting terms, i.e., $[\mathcal{A}^i_{j}, \mathcal{A}^i_k]=0$ for all $j,k$. $n_i$ measurements on $\kett{O(t)}_\mathcal{Q}$ are then taken in the eigenbasis of $M^{\mathcal{A}^i}_\mathcal{Q}$ for each $i=1,...,K$. The number of samples is then reduced to $N \propto \left(\sum_{i=1}^K \sqrt{\text{Var}(\mathcal{A}^i)} \right)^2$, using the optimal measurement allocation strategy \cite{rubin2018application,crawford2021efficient}. Finally, when all terms commute, a further reduction to $N \propto \text{Var}(\mathcal{A})_{O(t)}$ is possible. Related discussions are delegated to Appendix~\ref{app:superop_comparison}.

\subsection{Operator geometrical properties} \label{sec:geometrical_prop}

Geometrical properties of Heisenberg operators are properties such as their size/support, bulk and boundary volume, shape, velocity, and so on, visualized in Fig.~\ref{fig:lightcone_geometry}. They are central in the context of operator hydrodynamics and growth \cite{khemani2018operator,von2018operator,nahum2018operator,roberts2018operator,qi2019quantum,xu2024scrambling}, where experimental measurements are invaluable.

\begin{figure}[h]
    \centering
    \includegraphics[width=.75\linewidth]{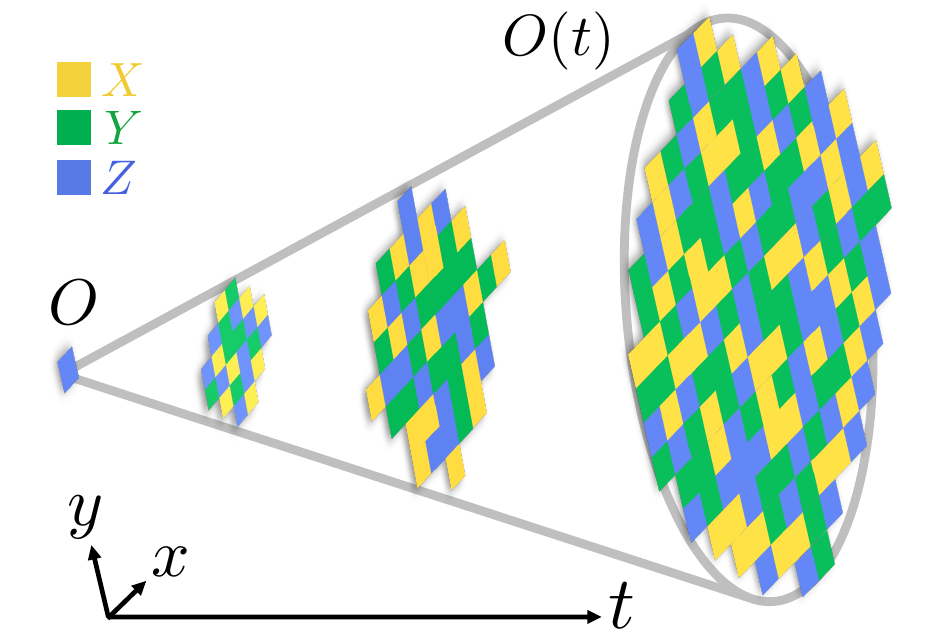}
    \caption{\textbf{Operator growth in the Heisenberg picture.} Dynamics of initially localized operator on a 2D lattice; a single Pauli term of $O(t)$ is visualized. Operator geometrical properties -- such as its surface area, volume, etc. -- are functions of the local occupancies of each of the three `flavors' of Pauli operators $\{X, Y, Z\}$, and are therefore expectation values of superoperators diagonal in the Pauli basis, $\langle \mathcal{D} \rangle_{O(t)}$. }
    \label{fig:lightcone_geometry}
\end{figure}

By definition, they are linear functions depending on the spatial occupancy of operators, and can therefore be expressed as superoperators that are linear combinations of projectors onto Pauli operators $P_k \tr(P_k (\cdot))$. In other words, any geometrical property of $O(t)$ can be written as the superoperator expectation value $\langle \mathcal{D} \rangle_{O(t)} = \tr(O(t) \mathcal{D}(O(t)))/2^n$, where $\mathcal{D}$ is a superoperator diagonal in $\mathcal{P}$ with operator-sum representation:
\begin{equation} \label{eq:superop_diagonal}
    \mathcal{D}(\cdot) = \sum_k f_k P_k (\cdot) P_k.
\end{equation}

Importantly, it can be shown (Lemma~\ref{lemma:diag_superop}, Appendix~\ref{app:superop}) that diagonal superoperators of the form Eq.~(\ref{eq:superop_diagonal}) coincide with superoperators with diagonal Pauli Transfer Matrices (PTMs), i.e., their PTMs take the form:
\begin{equation}
    M^{\mathcal{D}}_\mathcal{P} = \sum_{k=1}^{4^n} \lambda_k \ketbra{k}{k} = \sum_{k=1}^{4^n} \lambda_k \kett{P_k}_\mathcal{P} \leftindex_{\mathcal{P}}{\bbra{P_k}}.
\end{equation}
Consequently, expectation values of diagonal superoperators $\mathcal{D}$ can always be computed as $\langle \mathcal{D} \rangle_O = \leftindex_{\mathcal{P}}{\bbra{O(t)} M^{\mathcal{D}}_\mathcal{P} \kett{O(t)}}_\mathcal{P}$ by sampling from $\kett{O}_\mathcal{P}$ in the computational basis (or equivalently, sampling from $\kett{O}_\mathcal{C}$ in the Bell basis). Denoting $\norm{\mathcal{D}} \equiv \max_{\kett{O}_\mathcal{Q} \in \mathcal{H}'} \norm{M^{\mathcal{D}}_\mathcal{Q} \kett{O}_\mathcal{Q}}$ the operator/spectral norm of $\mathcal{D}$, we have:
\begin{corollary}[Diagonal superoperators] \label{corr:diagonal_superop}
Consider the task of estimating $\langle \mathcal{D} \rangle_O$ to additive error $\epsilon$ and success probability at least $1-\delta$, where $\mathcal{D}$ is a superoperator diagonal in the Pauli basis. There is an algorithm that achieves this using $\bigo{\norm{\mathcal{D}}^2 \log(1/\delta)/\epsilon^2}$ samples of $\kett{O(t)}_\mathcal{P}$, measured in the computational basis.
\end{corollary}

Similarly, specializing Theorem~\ref{theorem:simul_otoc} to diagonal OTOCs yields:
\begin{corollary}[Estimating all diagonal OTOCs] \label{corr:diagonal_otoc}
Consider the task of simultaneously estimating all $4^n$ quantities $\{\tr(O(t) P O(t) P)/2^n\}_{P \in \mathcal{P}_n}$ to additive error $\epsilon$ and success probability at least $1-\delta$. There is an algorithm that achieves this using $\bigo{\log(4^n/\delta)/\epsilon^2}$ samples of $\kett{O(t)}_\mathcal{P}$, measured in the computational basis.
\end{corollary}

In other words, the Pauli distribution of $O(t)$ -- accessed by sampling from $\kett{O(t)}_\mathcal{P}$ -- contains all information on its geometrical properties at time $t$. From a practical point of view, this approach is efficient and experimentally friendly, only requiring measurements from a fixed set of experimental conditions. This motivates our discussion in Section~\ref{sec:extracting_information} on efficiently obtaining and storing an empirical probability distribution in practice -- a single set of samples/experimental setting suffices for the extraction of any geometrical property at time $t$ in an \textit{a posteriori} manner, using the Monte Carlo estimator Eq.~(\ref{eq:computable_general_form}) with $f[P_k, O(t)] = \lambda_k$ computed classically. 

Examples of useful geometrical properties encoded in the Pauli basis that can be simultaneously accessed are as follows. 

~\\ \noindent \textbf{Operator size} \cite{nahum2018operator,roberts2018operator,qi2019quantum}: Define $s(P_k) \in \{0,...,n\}$ as the number of non-identity terms present in $P_k$, e.g. $s(XI...I) = 1$, $s(XY) = 2$. Its corresponding superoperator $\mathcal{S}(\cdot)$ is:
\begin{alignb} \label{eq:size_superop}
    \mathcal{S}(\cdot) &= \frac{3n}{4}(\cdot) - \frac{1}{4}\sum_{P \in \mathcal{P}_{\text{loc}}} P (\cdot) P \\
    &= \frac{1}{2^n} \sum_{k =1}^{4^n} s(P_k) P_k \tr(P_k (\cdot)),
\end{alignb}
where $\mathcal{P}_{\text{loc}}$ is the set of $3n$ size-1 Pauli operators. $\mathcal{S}$ is diagonal in $\mathcal{P}$ (and therefore has a diagonal PTM with entries $\bra{k}M^\mathcal{S}_\mathcal{P}\ket{k} = s(P_k)$), so $\langle \mathcal{S} \rangle_{O(t)}$ can be estimated most efficiently using samples of $\kett{O(t)}_\mathcal{P}$. This approach coincides with an existing protocol \cite{hu2023quantum}. Its variance $\delta \mathcal{S} = \text{Var}(\mathcal{S})_{O(t)}$ is also a relevant quantity in the context of open system dynamics \cite{schuster2023operator}. Further details are delegated to Appendix~\ref{app:superop_comparison}. \\

\noindent \textbf{Operator density/boundary position} \cite{khemani2018operator,von2018operator,nahum2018operator,fisher2023random}: Define $\delta_{\text{RHS}=x}(P_k) \in \{0,1\}$ as the function that takes value 1 when the rightmost non-identity term of $P_k$ (its right boundary) is located at position $x$, and 0 otherwise, e.g. $\delta_{\text{RHS}=1}(XI...I) = 1$, $\delta_{\text{RHS}=3}(XYI...I) = 0$, etc. The \textit{right-boundary weight}:
\begin{equation}
    \Delta(O(t),x) \equiv \sum_k p(P_k, O(t)) \delta_{\text{RHS}=x}(P_k)
\end{equation}
quantifies the average weight of $O(t)$ with right-boundary terminating at site $x$, forms a density, and probes the hydrodynamical spread of operators in time. Viewed as a superoperator expectation value, its moments can be computed within our framework by again sampling from $\kett{O(t)}_\mathcal{P}$. This is efficient since $\abs{\delta_{\text{RHS}=x}(P_k)} \leq 1$ (via Hoeffding's inequality). \\

Further examples of geometrical properties are boundary size, velocity, curvature/anisotropy \cite{nahum2018operator}, spatial range \cite{vanovac2024finite}, presence of voids \cite{liu1912void}, their higher order moments \cite{von2018operator,schuster2023operator}, and so on; their estimation is efficient as long as they have bounded variances, and their values (evaluated on Pauli operators $f(P_k)$) can be determined efficiently (on either a classical/quantum computer). Notably, diagonal superoperators that do not admit a sparse operator-sum decomposition Eq.~(\ref{eq:superop_diagonal}) (such as the right-boundary weight) cannot be straightforwardly computed using methods that boil down to sampling from $\mathcal{C}$ (or any unitarily equivalent basis), such as $2n$-qubit approaches that do not employ Bell sampling, and $n$-qubit approaches such as the one proposed in Section~\ref{sec:n_qubit_rand}. Further discussions and numerical experiments are delegated to Appendix~\ref{app:superop_comparison}.

\section{Algorithms for two-point correlators} \label{sec:2pc_main}

Two-point correlators are useful probes of operator growth \cite{von2018operator,parker2019universal,von2022operator}, scrambling, and chaos \cite{cotler2017chaos,fisher2023random,moudgalya2024symmetries,yoshimura2025theory}, and also appear in the computation of other physical quantities. For instance, as discussed in Section~\ref{sec:ose}, the Pauli amplitudes $c_k = \tr(P_k O(t))/2^n$ are involved in the computation of the OSE. Dynamical correlators characterize low-energy excitations of correlated matter via the dynamical structure factor \cite{baez2020dynamical}, the transport and spreading of correlations via the diffusion constant \cite{von2022operator,rakovszky2022dissipation}. Meanwhile, autocorrelators -- i.e., when the inner product is taken with $O$ at $t=0$, i.e., $\langle O(0) O(t) \rangle_\rho$ -- appear in computations of the spectral form factor \cite{cotler2017chaos,fisher2023random} (or its low-weight approximation \cite{yoshimura2025theory}), which relate the spectral and dynamical properties of a quantum system.

In the following, we describe two approaches for their computation. The first is based on the observation that vectorization maps amplitudes to two-point correlators. The second is based instead on working with $\kett{U}$, the vectorization of the unitary $U$; it has the advantage of fully inheriting the superoperator/transfer matrix formalism discussed in Section~\ref{sec:superop}, enabling the simultaneous estimation of many two-point correlators as expectation values of commuting observables. 

\subsection{Two-point correlators as amplitudes} \label{sec:2pc_amplitude}

The vectorization map Eq.~(\ref{eq:vect_map}) maps wavefunction amplitudes to two-point correlators Eq.~(\ref{eq:op_amp_def}), and more generally inner products between vectorized operators $\kett{O(t)} = \kett{U O U^\dagger}$, $\kett{O'(t')} = \kett{U' O' U'^\dagger}$ to operator inner products:
\begin{equation*}
    c(O ,O') \equiv \frac{1}{2^n} \tr(O^\dagger(t) O'(t')) = ~_{\mathcal{Q}} \langle \! \langle O(t) \kett{O'(t')}_\mathcal{Q}
\end{equation*}
Our framework therefore allows them to be extracted from quantum computers, using protocols for evaluating amplitudes and inner products. 

A naive approach for their computation involves applying the standard Hadamard test with the controlled-$(U^\dagger \otimes U^\mathsf{T})(U'^\dagger \otimes U'^\mathsf{T})$ unitary, mapping $c(O ,O')$ to an expectation value on an ancilla qubit. This can be improved by exploiting the unitality -- the identity-preserving nature -- of Heisenberg evolution. Through the circuit shown in Fig.~\ref{fig:2pc_cirq}, detailed in Appendix~\ref{app:2pc_amplitude}, $c(O ,O')$ is again encoded as an expectation value, without needing to query controlled time evolution unitaries, greatly simplifying implementation.

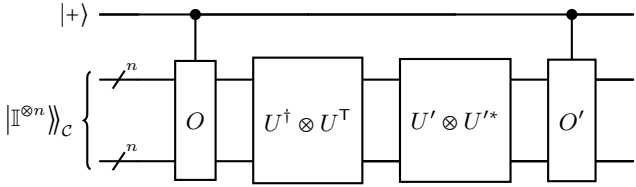
\begin{figure}[H]
\centering
\begin{quantikz}
\lstick{\(\ket{+}\)}     & \qw       & \ctrl{1}     & \qw  & \qw & \ctrl{1}  & \qw \\
\lstick[wires=2]{\(\kett{\mathbb{I}^{\otimes n}}_\mathcal{C}\)}   & \qwbundle{n}   & \gate[2]{O}  & \gate[2]{U^\dagger \otimes U^\mathsf{T}}  & \gate[2]{U' \otimes U'^*} & \gate[2]{O'}  & \qw \\
   & \qwbundle{n}            &              &  & \qw & \qw & \qw 
\end{quantikz}
\caption{\textbf{Circuit for computing two-point correlators.} Exploiting unitality allows controlled time-evolution to be avoided. Measurement of the ancilla qubit in the $\{\ket{\pm}\}$ basis yields $c(O ,O')$.}
\label{fig:2pc_cirq}
\end{figure}

\subsection{Two-point correlators as superoperator expectation values over $U$} \label{sec:2pc_expt}

Consider the task of simultaneously estimating many dynamical two-point correlators $\{\tr(P_i Q_i(t))/2^n\}_{i=1}^M$ at once, where $P_i, Q_i$ are Pauli operators. We switch our focus from Heisenberg operators to $U$, the unitary carrying out dynamics, and consider the vectorization of $U$, the $2n$-qubit state $\kett{U}_\mathcal{Q}$. 

Firstly, note that our framework fully carries over for the preparation of the state $\kett{U}_\mathcal{Q}$. Making full use of the discussion of Section~\ref{sec:preparation} and Appendix~\ref{app:details_state_prep}, they can be prepared as $R_{\mathcal{C} \rightarrow \mathcal{Q}} (U \otimes \mathbb{I}) \kett{\mathbb{I}}_\mathcal{C}$ on a quantum computer. Additionally, the application of $U$ can be parallelized; splitting it as $U = U_L U_R$, the Ricochet identity Eq.~(\ref{eq:ricochet}) leads to $\kett{U}_\mathcal{C} = (U \otimes \mathbb{I}) \kett{\mathbb{I}}_\mathcal{C} = (U_L \otimes U_R^\mathsf{T}) \kett{\mathbb{I}}_\mathcal{C}$, allowing $\kett{U}_\mathcal{C}$ to be prepared in half the depth of $U$, given the ability to implement $U_R^\mathsf{T}$. Note also that access to the adjoint $U^\dagger$ is unnecessary, in contrast to the preparation of $\kett{O(t)}$.

Next, observe that the superoperator and transfer matrix formalisms also carry over to this context. Considering the superoperator $\mathcal{A}_{kl}(\cdot) = Q_k (\cdot) Q^\dagger_l$, its expectation value on $U$ is given by:
\begin{equation*}
    \langle \mathcal{A}_{kl} \rangle_{U} = \tr(U^\dagger Q_k U Q_l^\dagger )/2^n = \tr(Q_l^\dagger Q_k(t))/2^n.
\end{equation*}
In other words, two-point correlators correspond to (self-adjoint) superoperator expectation values over $U$, allowing us to fully make use of the discussion of Section~\ref{sec:superop} and Appendix~\ref{app:superop}. Equivalent to how the simultaneous computability of OTOCs/superoperator expectation values on Heisenberg operators hinge on the commutation between superoperators, the criteria for the simultaneous computability of $\{\langle \mathcal{A}_{i} \rangle_{U}\}_i$ are determined by the commutativity between the $\mathcal{A}_{i}$'s (which directly follow from Lemmas~\ref{app:commutation_condition} and \ref{lemma:anticomm_basis} of Appendix~\ref{app:general_commutation}). 

Following the above discussion, we immediately recover two-point correlator analogues of our results on OTOCs/superoperator expectation values on Heisenberg operators. Firstly, we have the analogue of Theorem~\ref{theorem:simul_otoc}:
\begin{theorem}[Two-point correlators, commuting superoperators] \label{theorem:simul_2pc}
Consider the task of simultaneously estimating $M$ correlators $\left\{ \tr(P_{i} Q_{i}(t))/2^n \right\}_{i=1}^M$ to additive error $\epsilon$ and success probability at least $1-\delta$. There is an algorithm that achieves this using $\bigo{\log(M/\delta)/\epsilon^2}$ samples of $\kett{U}_\mathcal{C}$, measured in the common eigenbasis of $\{P_i \otimes Q_i\}_{i=1}^M$.
\end{theorem}

Next, observe that autocorrelators taking the form $\tr(P^\dagger P(t))/2^n$ -- necessary for e.g. the computation of the spectral form factor \cite{cotler2017chaos,fisher2023random} -- are expectation values $\langle \mathcal{D} \rangle_U$ of diagonal superoperators $\mathcal{D}(\cdot) \equiv P (\cdot) P^\dagger$. Their treatment in Section~\ref{sec:geometrical_prop} carries over fully: Lemma~\ref{lemma:diag_superop} implies that computational basis measurements on $\kett{U}_{\mathcal{P}}$ suffice to simultaneously compute all such correlators. That is, we have the analogue of Corollary~\ref{corr:diagonal_superop}:
\begin{corollary}[Diagonal two-point correlators]
Consider the task of estimating $\langle \mathcal{D} \rangle_U$ to additive error $\epsilon$ and success probability at least $1-\delta$, where $\mathcal{D}$ is a superoperator diagonal in the Pauli basis. There is an algorithm that achieves this using $\bigo{\norm{\mathcal{D}}^2 \log(1/\delta)/\epsilon^2}$ samples of $\kett{U}_\mathcal{P}$, measured in the computational basis.
\end{corollary}

The above results complement existing approaches for the simultaneous estimation problem based on classical shadows \cite{kunjummen2023shadow,levy2024classical}, with sample complexities that scale exponentially with the locality of the two-point correlators. In contrast, the sample complexities of our algorithms are independent of locality, and are based instead on their commutation, enabling up to exponentially many correlators to be simultaneously estimated with a sample complexity linear in $n$. It also generalizes algorithms for the estimation of diagonal PTM elements based on Bell sampling \cite{chen2022quantum,caro2024learning}, which interprets $\kett{U}$ as the Choi state of $U$.

Finally, note that the structure and classification of the eigenbases of commuting (and anticommuting) superoperators also carry over from Appendix~\ref{app:general_commutation}. In particular, in the `all simultaneously commuting' case, an $n$-qubit variant of our results hold, allowing this class of two-point correlators to be computed simultaneously with the same sample complexity, using only $n$ qubits and classical randomization, which effectively encodes $U$ as a $n$-qubit mixed state $\rho(U)$ (c.f. Appendix~\ref{app:n_qubit_algorithms}). Analogous to Table~\ref{tab:features_Ot}, a summary of properties encoded in $\kett{U}$ is provided in Table~\ref{tab:features_U} in the appendix.

\begin{figure*}
    \centering
    \includegraphics[width=.99\linewidth]{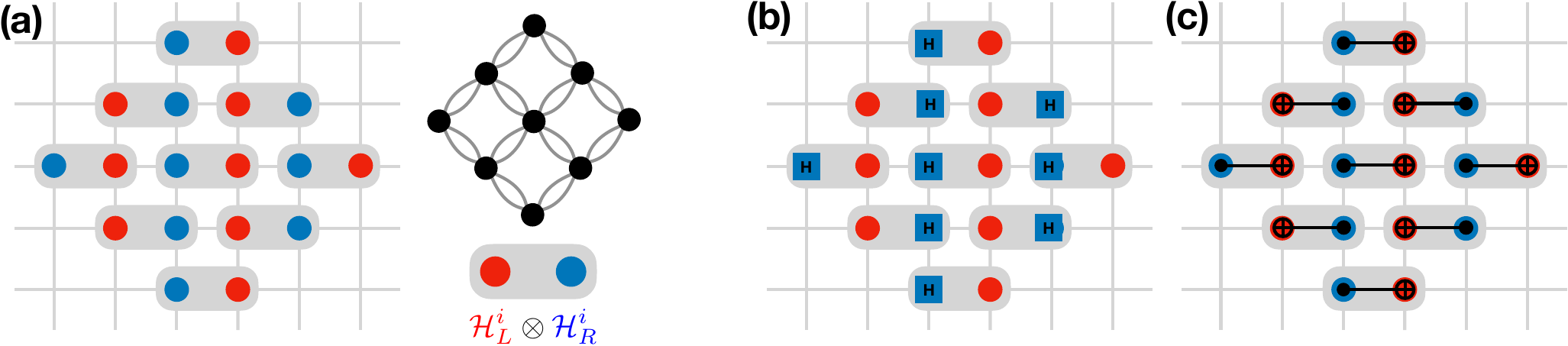}
    \caption{\textbf{Implementation of our framework in a square-lattice device.} (a) Initial embedding of the $2n = 18$ qubits of $\kett{O}$ onto a square grid. Red (blue) qubits occupy one half of the Hilbert space $\mathcal{H}_L$ ($\mathcal{H}_R$). Each red-blue qubit pair (delimited by gray capsules) encodes a single operator located at spatial position $i$. The effective propagator $e^{-iHt} \otimes e^{iH^\mathsf{T}t}$, once Trotterized, contains entangling gates that act between nearest and next-nearest-neighbors in a square grid. This is represented as edges between black nodes (each representing a local operator) of the graph at the right. Within a single Trotter step, there will be 24 entangling gates, representing by the 24 (curved) edges in this graph. (b) and (c) shows the transversal Hadamard and CNOT gates carrying out the basis transformation $R_{\mathcal{P} \rightarrow \mathcal{C}}$.}
    \label{fig:2d_embedding}
\end{figure*}

\section{Implementation proposal on quantum hardware} \label{sec:implementation_proposal}

Our algorithms are amenable to implementation on current-generation quantum computers and analog simulators, since it inherits much of the structure present in conventional quantum simulation in the Schrödinger picture. We demonstrate this by concretely outlining experiments to study operator growth and scrambling.

We consider the problem of simulating Heisenberg-picture dynamics $O \rightarrow O(t) = e^{iHt} O e^{-iHt}$ under Hamiltonians defined on a square lattice such as the mixed-field quantum Ising model, with:
\begin{equation} \label{eq:2d_mfim}
    H = h_x \sum_k \sigma_x^{(k)} + h_z \sum_k \sigma_z^{(k)} - J \sum_{\langle i,j \rangle} \sigma_z^{(i)} \sigma_z^{(j)},
\end{equation}
where $\langle i,j \rangle$ denotes summation over nearest neighbors on a square lattice, and $h_x,h_z,J$ are real parameters. The longitudinal field $h_z$ tunes the model between non-interacting and non-integrable regimes. The initial operator is taken to be a spatially local Pauli operator. We consider a prototypical task in studies of operator hydrodynamics: the mapping of the space-time profile of OTOCs/operator size, involving the computation of diagonal OTOCs for spatially local Pauli operators.

We consider implementation on digital quantum computers featuring a square grid connectivity graph (i.e. two-qubit gates are natively executable on a square lattice), which are featured in existing devices such as those based on superconducting qubits \cite{google2025quantum,google2025observation}. Samples of $\kett{O(t)}_\mathcal{P}$ are prepared, following the initialize-and-evolve protocol of Section~\ref{sec:preparation}:
\begin{enumerate}
    \item Prepare initial operator $O$ as the product state $\kett{O}_\mathcal{C}$.

    \item Implement Heisenberg evolution in $\mathcal{C}$ by time-evolving $\kett{O}_{\mathcal{C}}$ under the `super-Hamiltonian':
    \begin{equation} \label{eq:super-ham}
        \mathbb{H} \equiv -H \otimes \mathbb{I} + \mathbb{I} \otimes H^\mathsf{T}
    \end{equation}
    due to Eq.~(\ref{eq:tm_unitary}).
    
    \item Switch to $\mathcal{P}$ by applying $R_{\mathcal{P} \rightarrow \mathcal{C}}^\dagger$.

    \item Measure $\kett{O(t)}_\mathcal{P}$ in the computational basis.
\end{enumerate}

Crucially, the $2n$ qubits of $\kett{O}$ are arranged in a way that both respects the 2D geometry of the Hamiltonian $H$ (so that local terms in $H$ translate to local operations on the quantum computer), and the tensor product structure of the vectorization map (so that vectorization basis changes are also local). That is, we embed $\kett{O}$ onto the square grid of the quantum computer as shown in Fig.~\ref{fig:2d_embedding}.(a). Implementing Hamiltonian simulation of $\mathbb{H}$ via Trotterization, each term in $\mathbb{H}$ (graph with black nodes in Fig.~\ref{fig:2d_embedding}.(a)) thus translates to entangling gates between spatially local qubits on the quantum computer, regardless of lattice size. Similarly, $R_{\mathcal{P} \rightarrow \mathcal{C}}$ can be carried out in constant depth as shown in Fig.~\ref{fig:2d_embedding}.(b),(c).

For the 2D Ising model of Eq.~(\ref{eq:2d_mfim}), a single Trotter step can be compiled to a depth-5 circuit, schematically shown in Fig.~\ref{fig:2d_schedule}. It involves a sequence of gates $e^{-i\sigma_x^{(k)} \Delta t}$, $e^{-i\sigma_z^{(k)} \Delta t}$, and $e^{-i\sigma_z^{(i)} \sigma_z^{(j)} \Delta t}$, together with an additional layer of SWAP gates. 

Finally, to access the Pauli distribution of $O(t)$, it suffices to perform Bell basis measurements, i.e. transform to the Pauli basis via $R^\dagger_{\mathcal{P} \rightarrow \mathcal{C}}$. Notably, for each value of $t$, a single set of samples from $\kett{O(t)}$ suffices; the resulting (empirical) Pauli distribution allows the recovery of all geometrical properties of $O(t)$ such as the OTOC via Eq.~(\ref{eq:computable_general_form}). This enables the entire space-time profile of $O(t)$ to be obtained efficiently in practice.

Assuming local gate noise (i.e. only affecting manipulated qubits), correlated errors between $\mathcal{H}_L$ and $\mathcal{H}_R$ only occur during basis changes and possibly initialization, which prevents the appearance of globally correlated noise, rendering the noise structure similar to the Schrödinger-picture Hamiltonian simulation of $H$. In Appendix~\ref{app:details_state_prep}, we discuss additional details and error mitigation strategies unique to the Heisenberg picture, such as the possibility of fast-forwarding initial Clifford gates, and the cancellation of gates outside the lightcone of $O(t)$. We also discuss implementation considerations on analog quantum simulators that implement dynamics by tuning the parameters of a Hamiltonian, on, e.g., neutral atom arrays \cite{schuster2022many,garttner2017measuring,mi2021information,sanchez2022emergent,brown2023quantum}, in which case the state preparation protocol of Section~\ref{sec:preparation} can more naturally be phrased as a quantum quench involving an adiabatically prepared initial state of a local Hamiltonian, followed by the `super-Hamiltonian' Eq.~(\ref{eq:super-ham}).

\begin{figure}[h]
    \centering
    \includegraphics[width=1\linewidth]{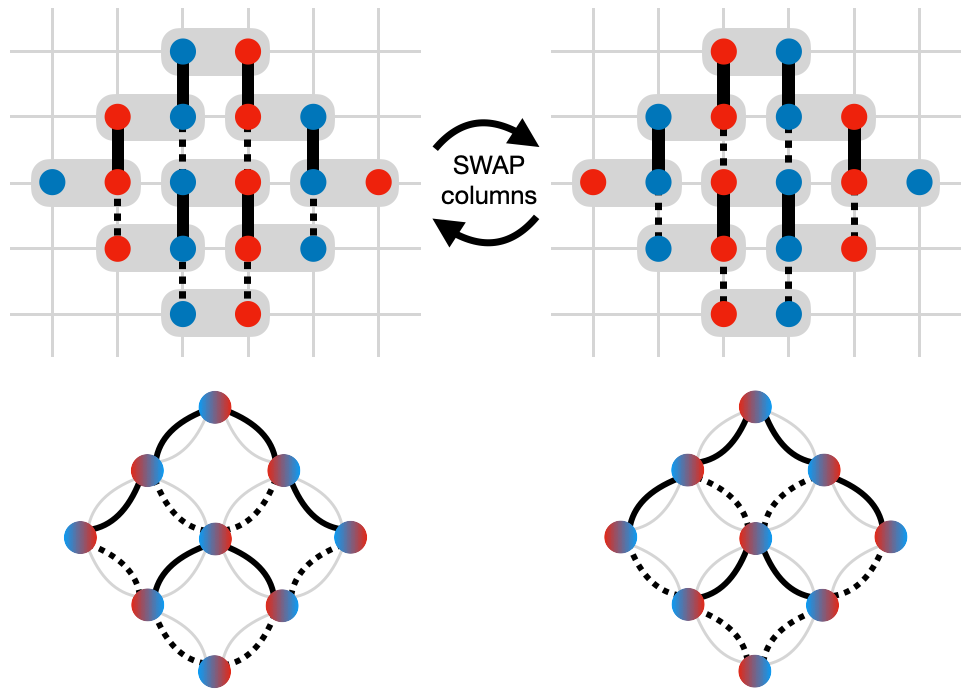}
    \caption{\textbf{Constant-depth schedule executing a single Trotter step.} Gates executed in the device (top), containing all edges in the square connectivity graph of the 2D Ising model (bottom). Solid or dotted black lines indicate entangling gates.
    The sequence of gates are: (1) solid black lines (top left) $\longrightarrow$ (2) dotted black lines (top left) $\longrightarrow$ (3) SWAP gates between qubits in odd and even columns $\longrightarrow$ (4) solid black lines (top right) $\longrightarrow$ (5) dotted black lines (top right).}
    \label{fig:2d_schedule}
\end{figure}

\section{\texorpdfstring{$\lowercase{n}$}{n}-qubit randomized algorithms} \label{sec:n_qubit_rand}
Here, we apply our formalism to a class of alternative, but closely related protocol that makes use of only $n$ qubits \cite{swingle2016measuring,mi2021information,schuster2022many,schuster2023operator,schuster2023learning,cotler2023information,google2025observation,algorithmiq}, in contrast to the main vectorization approach involving $2n$-qubit states $\kett{O}_\mathcal{Q}$. This comes at the expense of a doubled circuit depth, fixing part of the basis of measurement, and requires introducing classical randomization. Detailed discussions are delegated to Appendix~\ref{app:n_qubit_details}.

We begin by describing how one can make use of $n$ qubits of quantum memory and the ability to implement the unitary $V$ to sample from it in the computational operator basis $\mathcal{C} = \{ \ketbra{i}{j} \}$, i.e., sampling from the probability distribution $p(i,j) = \abs{\tr(\ketbra{i}{j}V)}^2/2^n = \abs{\bra{j}V\ket{i}}^2/2^n$:
\begin{enumerate}
    \item Randomly sample a bitstring $i$ with probability $p_1(i)=1/2^n$.
    \item Prepare the state $V\ket{i}$ and measure it in the computational basis, yielding $j$ with probability $p_2(j|i) = \abs{\bra{j}V\ket{i}}^2$.
\end{enumerate}
The resulting joint probability distribution of this process is $p(i,j) = p_1(i)p_2(j|i) = \abs{\bra{j}V\ket{i}}^2/2^n$. Generalizing the initial and final bases to $\{\ket{\phi_i}\}$ (with $\ket{\phi_i} \equiv U_\phi \ket{i}$) and $\{\ket{\psi_i}\}$ (with $\ket{\psi_i} \equiv U_\psi \ket{i}$) respectively, we immediately obtain:
\begin{theorem}[Operator sampling with $n$ qubits] \label{theorem:1_copy_sampling}
Given the ability to implement the unitaries $V$, $U_\psi$, and $U_\phi$, $n$ qubits of quantum memory, and classical randomization, there is an algorithm to sample from the probability distribution $\abs{\bra{\psi_j} V \ket{\phi_i}}^2/2^n$. 
\end{theorem}

The ensemble of $n$-qubit states -- which we denote as $\rho(V)$ -- can be viewed as encoding the unitary $V$ in the operator basis $\{\ketbra{\phi_i}{\psi_j}\}_{i,j}$. Setting $V = U^\dagger OU$ enables sampling from the operator distribution $\abs{\bra{\psi_j}O(t)\ket{\phi_i}}^2/2^n$, which simulates measurements on $\kett{O(t)}_\mathcal{C}$ in the basis $\{\ket{\psi_i} \otimes \ket{\phi^*_j}\}$, separable across $\mathcal{H}_L$ and $\mathcal{H}_R$. Compared to vectorization where each coherent preparation of $\kett{O(t)}_\mathcal{C}$ involves the application of the $2n$-qubit unitary $(U^\dagger \otimes U^\mathsf{T})(O \otimes \mathbb{I})$, this approach involves executing $U^\dagger OU$ on $n$ qubits, halving qubit count at the expense of doubling depth. Note also that initialization fixes the choice of $\{\ket{\phi_i}\}$.

An important consequence of Theorem~\ref{theorem:1_copy_sampling} is that samples of $n$-qubit states $\rho(O(t))$ suffice for tasks that only involve measuring $\kett{O(t)}_\mathcal{C}$ in bases separable across $\mathcal{H}_L$ and $\mathcal{H}_R$. An example of a \textit{non-separable} basis is the Bell basis (equivalent to sampling from the state $\kett{O(t)}_\mathcal{P}$ in the computational basis), where Theorem~\ref{theorem:1_copy_sampling} does not apply. However, separable measurements suffice for the estimation of individual OTOCs and superoperators via Eq.~(\ref{eq:superop_term_by_term}), and their simultaneous computation in the `all simultaneously commute' subcase (Fig.~\ref{fig:comutation}(a)). A small modification further leads to the ability to compute simultaneously estimate many commuting and local two-point correlators. In other words, we have:
\begin{corollary}[$n$-qubit analogues of Theorems~\ref{theorem:simul_otoc} and \ref{theorem:simul_2pc}] \label{theorem:n_qubit_main}
Let $P_i,Q_i$ be Pauli operators such that for all pairs $(i,j)$, $[P_i, P_j] = 0$ and $[Q_i, Q_j] = 0$.
\begin{enumerate}
    \item There is an algorithm that simultaneously estimates $M$ OTOCs $\{ \tr(O P_{i} O Q_{i})/2^n\}_{i=1}^M$ to $\epsilon$ error using $\bigo{\log(M/\delta)/\epsilon^2}$ samples of $\rho(O(t))$.

    \item There is an algorithm that simultaneously estimates $M$ correlators $\left\{ \tr(P_{i} Q_{i}(t))/2^n \right\}_{i=1}^M$ to $\epsilon$ error using $\bigo{\log(M/\delta)/\epsilon^2}$ samples of $\rho(U)$.
\end{enumerate}
\end{corollary}

The two approaches can be naturally related through graphical/tensor network notation \cite{schuster2022many}. The equality of Fig.~\ref{fig:penrose_1copy}.(a) -- obtained by `unfolding' the elbow ($\supset$) representing $n$ Bell pairs at the LHS -- suggests that one can interpret $\mathcal{H}_L$ (associated with the index $j$) and $\mathcal{H}_R$ (associated with the index $i$) as representing the output and input parts respectively of the $n$-qubit computation.
\begin{figure}[H]
    \centering
    \includegraphics[width=1\linewidth]{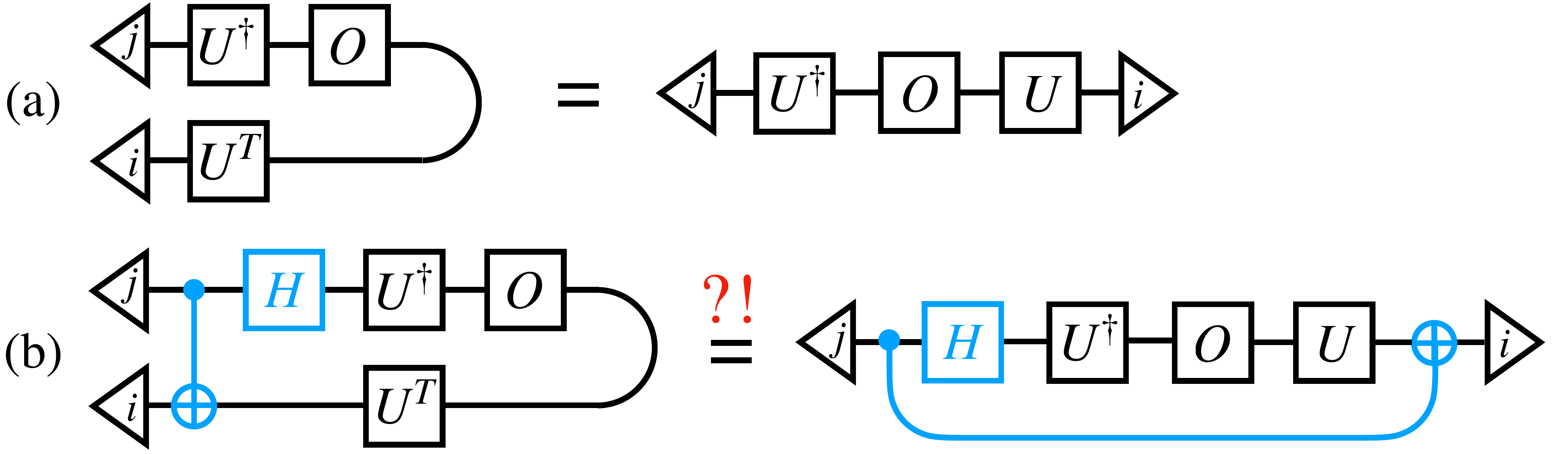}
    \caption{\textbf{Relation between the $2n$-qubit vectorization (LHS) and $n$-qubit (RHS) approaches.} Each wire contains $n$ qubits. (a) represents the expression $\langle j, i \kett{O(t)}_\mathcal{C} = \bra{j}O(t)\ket{i}$, while the LHS of (b) represents $\langle j, i \kett{O(t)}_\mathcal{P}$. The ?! indicates a hypothetical equality that would need to hold for the corresponding $n$-qubit algorithm to exist.}
    \label{fig:penrose_1copy}
\end{figure}

This relationship, together with Lemma~\ref{lemma:anticomm_basis}, hints at why it may be difficult to sample from operator distributions that require entanglement between $\mathcal{H}_L$ and $\mathcal{H}_R$ using only $n$ qubits. Using the Pauli basis -- involving the Bell basis transformation Eq.~(\ref{eq:rotation_cp}) -- as an example, the $n$-qubit computation would require entangling the input and output in time as indicated in Fig.~\ref{fig:penrose_1copy}.(b), which does not appear naively realizable. The vectorization approach is therefore expected to be advantageous (in sample complexity, beyond the memory-depth tradeoff) when operator bases that entangle $\mathcal{H}_L$ and $\mathcal{H}_R$ are considered. The computation of operator geometrical properties (described in Section~\ref{sec:geometrical_prop}) is a concrete example, where the vectorization approach is able to access the Pauli basis and therefore achieve the lowest sample complexity, matching Corollary~\ref{corr:diagonal_superop}. 

\section{Discussion}

In this work we introduced a Heisenberg-picture quantum simulation framework based on vectorization, mapping time-evolved operators $O(t)$ to states $\kett{O(t)}$ on a doubled Hilbert space. This unifies a range of protocols to probe operator dynamics as standard state-property estimation problems and yields concrete new or improved protocols for a number of physically motivated operator properties, as summarized in Tables~\ref{tab:features_Ot} and \ref{tab:algorithms_index}. Natural generalizations such as to finite-temperature, open system dynamics, and qudits further illustrate the versatility of our framework.

Our framework yields new algorithms for the computation of the operator stabilizer entropy (OSE) and the local operator entanglement (LOE), by reducing them to tasks involving sampling from the Pauli distribution of $O(t)$, and multi-copy purity/overlap estimations. Together with the experimental amenability of our approach, this opens up new directions for quantifying the complexity of the dynamics of quantum systems. 

More generally, the framework formalizes prior OTOC estimation approaches expressed via vectorization/superoperators, and it makes clear how to generalize them to simultaneously estimate many commuting superoperator observables from shared measurement data. This is especially relevant to operator-hydrodynamic questions, where one typically wants space-time resolved operator-growth diagnostics rather than a single correlator. Thus our viewpoint provides a practical toolkit for extracting such families of observables within one experimental/simulation workflow.

It is intuitive that simultaneously learning/sampling many operator properties without $2n$ qubits should be inefficient in general. Formalizing this intuition into rigorous learning separations is left open for future work. One promising direction is to cast families of Heisenberg-picture inference tasks (e.g., estimating broad classes of Pauli-diagonal superoperator expectation values, or learning coarse-grained operator hydrodynamic profiles) into an information-theoretic model where restrictions on quantum memory, access to time reversal, or allowable measurement structure provably limit achievable sample complexity. These questions are addressed in a follow-up work \cite{paper2}.

% Below removes entry from TOC
\begingroup
\makeatletter
\let\addcontentsline\@gobblethree % don't write TOC/minitoc entry
\makeatother
\begin{acknowledgments}
SC, AA and GC thank Silvia Pappalardi and Xhek Turkeshi for useful discussions on magic monotones and out-of-time-ordered correlators. 
ZH thanks Sergey Filippov and Guillermo García-Pérez for highlighting the power of $n$-qubit randomized methods in this context.  
SC, ZH and GC acknowledge the funding support from NCCR SPIN, a National Centre of Competence in Research, funded by the Swiss National Science Foundation (grant number 225153). 
AA and ZH acknowledge support from the Sandoz Family Foundation-Monique de Meuron program for Academic Promotion.
\end{acknowledgments}

\bibliography{quantum}

\endgroup

\clearpage

\onecolumngrid

\counterwithin*{equation}{section}
\renewcommand\theequation{\thesection\arabic{equation}}

\renewcommand\partname{} 
\appendix
\begin{center}
 {\Large \textbf{Appendix} }   
\end{center}

\part{}
\parttoc 

\vspace{15pt}

\begin{table*}[ht!]
\renewcommand{\arraystretch}{1.4}
\setlength{\tabcolsep}{6pt}
\begin{tabular}{p{6.5cm} p{3.45cm} p{3.25cm}}

\textbf{Algorithm} & \textbf{$2n$-qubit encoding} & \textbf{$n$-qubit encoding} \\
\Xhline{1pt}

Operator sampling & Sec.~\ref{sec:main_vect} & App.~\ref{app:n_qubit_details} \\

Operator Stabilizer Entropy (OSE) & Sec.~\ref{sec:ose}, App.~\ref{app:ose} & -- \\

Local Operator Entropy (LOE) & Sec.~\ref{sec:loe} & --  \\

OTOCs/superoperator expectation values & Sec.~\ref{subsec:self-adjoint}, \ref{sec:geometrical_prop}, App.~\ref{app:superop} & Sec.~\ref{sec:n_qubit_rand}, App.~\ref{app:n_qubit_algorithms}  \\

Two-point correlators &   &   \\

\hspace{5mm} by commutation &  Sec.~\ref{sec:2pc_expt} &   App.~\ref{app:n_qubit_algorithms} \\

\hspace{5mm} by amplitude estimation & Sec.~\ref{sec:2pc_amplitude}, App.~\ref{app:2pc_amplitude} & -- \\

\Xhline{1pt}
\end{tabular}
\caption{Index of algorithms introduced in this work.}
\label{tab:algorithms_index}
\end{table*}

\clearpage

\section{Background and preliminaries} \label{app:background_prelim}

\subsection{Relation to existing work} \label{app:relation_existing}

Our approaches share connections to certain existing encodings and protocols, generalizing and unifying them via the vectorization, superoperator, and transfer matrix formalisms. We provide a non-exhaustive list of related work below.

~ \\ \noindent \textbf{Relation to existing protocols for OTOCs and operator size:} Our algorithms for OTOCs (more generally, superoperator expectation values) extend prior methods to enable simultaneous estimation. Recasting them as expectation values of Hermitian transfer matrices enables us to exploit mutual commutativity to directly bound their sample complexities, and exposes conditions under which they are jointly computable.
\begin{itemize}
    \item The vectorization encoding (Section~\ref{sec:superop}) is closely connected to $2n$-qubit approaches that require access to time-reversal $U^\dagger$ and transposition $U^\mathsf{T}$, or $U$ and complex-conjugation $U^*$ -- see e.g. \cite{yoshida2019disentangling,landsman2019verified,schuster2022many,sundar2022proposal,green2022experimental,xu2024scrambling}. When considering the computation of a single OTOC $\tr(O^\dagger(t) P^\dagger O(t) P)/2^n$, our approach yields many existing experimental protocols \cite{green2022experimental,sundar2022proposal,schuster2022many}; due to the equality: $\tr(O^\dagger(t) P^\dagger O(t) P)/2^n = \tr(O^\dagger P^\dagger(-t) O P(-t))/2^n$, when computing a single OTOC term, queries to $U^\dagger$ and $U^\mathsf{T}$ can be swapped for queries to $U$ and $U^*$ (in the vectorization language, this is equivalent to vectorizing and forward-time-evolving $P$ (instead of backward-time-evolving $O$), and measuring in the eigenbasis of the transfer matrix $O^\dagger \otimes O^\mathsf{T}$; c.f. also Corollary~\ref{corr:backward_otoc}). To simultaneously compute the OTOCs $\left\{ \tr(O(t)^\dagger P^\dagger_{i} O(t) Q_{i})/2^n \right\}$, it is necessary to vectorize $O(t)$ as we have done. For the operator size, an existing Bell-sampling approach \cite{hu2023quantum} is an instance of our approach based on sampling from $O(t)$ in the Pauli basis.

    \item The $n$-qubit encoding (Section~\ref{sec:n_qubit_rand}) is related to existing $n$-qubit approaches (for individual estimation) that implement $O(t)$ as a unitary, and possibly employ an additional ancilla qubit -- see e.g. \cite{swingle2016measuring,mi2021information,schuster2022many,schuster2023learning,cotler2023information,google2025observation,algorithmiq}. In particular, \cite{schuster2022many} also discusses related connections between $n$- and $2n$-qubit approaches. \cite{algorithmiq} discusses the simultaneous estimation of OTOCs in the same sense of Corollary~\ref{theorem:n_qubit_main}.1.

    An alternative class of protocols involve encoding $O$ as a highly mixed $n$-qubit state $\propto \mathbb{I} + O$; see e.g. \cite{li2017measuring,garttner2017measuring,schuster2023learning,moulik2024dqc1}. They are more naturally realized in e.g. NMR quantum computers based on the DQC1 model of quantum computation \cite{knill1998power}. \cite{schuster2023operator} describes a protocol for the operator size within this class of methods.
\end{itemize}

~ \\ \noindent \textbf{Relation to existing protocols for two-point correlators:} Our framework supplements and generalizes existing protocols related to learning the PTM elements of quantum channels \cite{kunjummen2023shadow,levy2024classical,chen2022quantum,caro2024learning}. For example, \cite{caro2024learning} describes protocols based on working with the $2n$-qubit Choi state of general channels using an existing Pauli-tomography algorithm \cite{huang2021information}, which in our language corresponds to estimating all two-point correlators $\{\tr(P_i U^\dagger P_j U)/2^n\}_{i,j=1}^{4^n}$ when the channel is taken to be unitary. Whereas \cite{huang2021information} requires jointly entangling multiple copies of the Choi state but is able to obtain all possible two-point correlators, our approach only requires measuring individual copies of $\kett{O(t)}$, when there is (commutation) structure to the set $\{P_i \otimes P_j\}$. Similarly, the approach of \cite{chen2022quantum} based on randomized measurements are limited to accessing low-weight Pauli operators.

~ \\ \noindent \textbf{Encoding linear operators as quantum states:} The encoding of quantum channels, Hamiltonians, or density matrices as quantum states so that they may be studied on quantum simulators has been explored in different contexts. We expect our results to also be useful and generalizable to these contexts.

\vspace{5pt}

Quantum channels directly encoded as quantum states via the Choi map -- which we discuss in Section~\ref{sec:2pc_main} and Table~\ref{tab:features_U} for the unitary case -- allows the learning of their structure \cite{low2009learning, montanaro2010quantum, garcia2021quantum, chen2023testing, angrisani2023learning, kunjummen2023shadow,caro2024learning,levy2024classical} and properties such as the operator entanglement \cite{zanardi2001entanglement} (of a unitary $U$; which is different from the LOE) \cite{mcginley2022quantifying}. Similarly, Hamiltonians $H$ encoded via the vectorization map (where $H$ is unknown and the encoding is prepared via queries to $e^{-iHt}$) have been studied for Hamiltonian learning \cite{low2009learning,zhao2024learning,castaneda2025hamiltonian}; we anticipate their results to be useful in our context, for the learning of Heisenberg operators. \cite{somma2025shadow} proposes reduced representations of quantum states that can be regarded as marginalizations of the vectorization map; applications to Heisenberg-picture simulation are briefly discussed. \cite{kamakari2022digital} encodes density matrices as $2n$-qubit quantum states via the vectorization map for open-system simulations.

~ \\ \noindent \textbf{Double-copy algorithms:} The $2n$-qubit approaches are also generally related to \textit{double copy algorithms} requiring access to conjugate queries $U^*$ (or complex-conjugated states). Bell sampling \cite{montanaro2017learning,gross2021schur,hangleiter2024bell} and related quantum learning protocols \cite{huang2021information,king2024exponential} can be regarded as vectorizing a quantum state in the Pauli basis, $\rho \rightarrow \kett{\rho}_\mathcal{P}$.

\subsection{Details on vectorization, superoperators, and transfer matrices} \label{app:details_vect}
We begin by discussing general properties of the vectorization map, superoperators, and transfer matrices in the context of our framework, following the brief description and notation of Section~\ref{sec:prelim}. They are direct consequences of the linear-algebraic structure of the vectorization and transfer matrix formalisms; basic correspondences are summarized in Table~\ref{tab:comparison}. Details on specific properties will be introduced in the following Appendices as needed. 

\begin{table}[h!]
\centering
\renewcommand{\arraystretch}{1.55} 
\begin{tabular}{lll} 
\thickhline
 & Heisenberg picture & Vectorization \\
\hline
State space & $\mathcal{L}(\mathcal{H}) = \mathcal{L}(\mathbb{C}^{2^n})$ & $\mathcal{H}' =\mathbb{C}^{4^n}$ \\
States & $O(t)$ & $\kett{O(t)}_\mathcal{Q}$ \\
Dual linear form & $\tr({O}(t)^\dagger(\cdot))$ & $ \leftindex_{\mathcal{Q}}{\bbra{O(t)}}$\\
Dynamics (in $\mathcal{C}$) & $U^\dagger(\cdot)U$ & $U^\dagger \otimes U^\mathsf{T}$ \\
Generator of dynamics (in $\mathcal{C}$)& $H$ & $-H \otimes \mathbb{I} + \mathbb{I} \otimes H^\mathsf{T}$ \\
Linear operator (over state space)~~~ & $\mathcal{A}$ & $M^\mathcal{A}_\mathcal{Q}$ \\ 
Inner product & $\frac{1}{2^n} \tr(A^\dagger B)$ & $\leftindex_{\mathcal{Q}}{\bbrakett{A}{B}}_\mathcal{Q}$ \\
Statistical moments & $\frac{1}{2^n} \tr(O(t)^\dagger \mathcal{A}^k(O(t)))$ & $\leftindex_{\mathcal{Q}}{\bbra{O(t)} (M^{\mathcal{A}}_\mathcal{Q})^k \kett{O(t)}}_\mathcal{Q}$~  \\ 
Projectors & $\frac{1}{2^n} Q_k \tr(Q_k^\dagger(\cdot))$ & $\kett{Q_k}_{\mathcal{Q} ~ \mathcal{Q}} \! \bbra{Q_k}$ \\
Completeness & $\frac{1}{2^n} \sum_k Q_k \tr(Q_k^\dagger(\cdot)) = (\cdot)$ ~~~~ & $\sum_k \kett{Q_k}_{\mathcal{Q} ~ \mathcal{Q}} \! \bbra{Q_k} = \mathbb{I}$ \\
Basis change& $\frac{1}{2^n} \sum_k Q'_{k} \tr(Q^\dagger_{k} (\cdot))$ & $\sum_k \kett{Q'_k}_{\mathcal{Q}' ~ \mathcal{Q}} \! \bbra{Q'_k}$  \\
Commutation & $[\mathcal{A}, \mathcal{B}] = 0$ & $[M^\mathcal{A}_\mathcal{Q}, M^\mathcal{B}_\mathcal{Q}] = 0$    \\
\thickhline
\end{tabular}
\caption{Several objects and identities that hold for $\mathcal{L}(\mathcal{H})$, and their equivalent in the doubled Hilbert space $\mathcal{H}'$ via the vectorization map.}
\label{tab:comparison}
\end{table}

\subsubsection{Basic properties of superoperators and transfer matrices}

The transfer matrix (TM) of a linear superoperator $\mathcal{A}$ in the orthogonal operator basis $\mathcal{Q} = \{Q_k\}_{k=1}^{4^n}$, denoted as $M^\mathcal{A}_\mathcal{Q}$, is the unique complex matrix that satisfies $\kett{\mathcal{A}(O)}_\mathcal{Q} = M^\mathcal{A}_\mathcal{Q} \kett{O}_\mathcal{Q} ~ \forall O \in \mathcal{L}(\mathcal{H})$. Conversely, given a $4^n \times 4^n$ matrix $M$, there is a unique superoperator $\mathcal{A}_\mathcal{Q}^M$ satisfying $M \kett{O}_\mathcal{Q} = \kett{\mathcal{A}_\mathcal{Q}^M (O)}_\mathcal{Q}$.  

An arbitrary superoperator $\mathcal{A}$ can always be written in the operator-sum-representation Eq.~(\ref{eq:superop_def}). The adjoint of $\mathcal{A}$, defined as the superoperator $\mathcal{A}^\dagger(\cdot)$ such that $\tr((\mathcal{A}^\dagger(X))^\dagger Y) = \tr(X^\dagger \mathcal{A}(Y)) ~\forall X,Y \in \mathcal{L}(\mathcal{H})$, has the decomposition:
\begin{equation}
    \mathcal{A}^\dagger(\cdot) = \sum_{kl} f^*_{kl} Q_k^\dagger(\cdot)Q_l.
\end{equation}
The following basic properties can be shown, following the definition of the transfer matrix:
\begin{align}
    & \text{Identity superoperator $\mathcal{I}$ maps to identity matrix $\mathbb{I}$:} && M_\mathcal{Q}^{\mathcal{I}} = \mathbb{I} \\
    & \text{TM of composed superoperators = Composition of TMs:} && M^{\mathcal{A}_1 \circ \mathcal{A}_2}_\mathcal{Q} = M^{\mathcal{A}_1}_\mathcal{Q} M^{\mathcal{A}_2}_\mathcal{Q} \\
    & \text{Superoperator of composed matrices = Composition of superoperators:} ~~ && \mathcal{A}^{MN}_\mathcal{Q} = \mathcal{A}^M_\mathcal{Q} \circ \mathcal{A}^N_\mathcal{Q} \\
    & \text{Adjoint of TM = TM of adjoint:} && \left(M^{\mathcal{A}}_\mathcal{Q}\right)^\dagger = M^{(\mathcal{A})^\dagger}_\mathcal{Q} \\
    & \text{Superoperator of adjoint matrix = adjoint of superoperator:} && \mathcal{A}^{M^\dagger}_\mathcal{Q} = \left(\mathcal{A}^M_\mathcal{Q} \right)
^\dagger
\end{align}

$\mathcal{A}$ is said to be \textit{self-adjoint} iff $\mathcal{A}^\dagger = \mathcal{A}$. Using the above properties, it can be shown that they map to Hermitian matrices via the vectorization map, i.e., $M_\mathcal{Q}^{\mathcal{A}} = \left( M_\mathcal{Q}^\mathcal{A} \right)^\dagger$ for self-adjoint $\mathcal{A}$, and conversely $\mathcal{A}^H_\mathcal{Q} = \left( \mathcal{A}^{H}_\mathcal{Q} \right)^\dagger$ for Hermitian matrix $H$. Further properties of self-adjoint superoperators are delegated to Appendix~\ref{app:superop}, where their mapping to physically measurable observables are exploited for the computation of operator properties on quantum computers, as discussed in Section~\ref{sec:superop}. They also encode two-point correlators when viewed as observables over the Choi state of the unitary carrying out dynamics; this is discussed in Appendix~\ref{app:2pc_choi_u}. \\

$\mathcal{A}$ is said to be \textit{unitary} iff $\mathcal{A}^\dagger \mathcal{A} = \mathcal{I}$, i.e., it preserves HS inner products between operators. Similarly, the above properties can be used to show that unitary superoperators map to unitary matrices, i.e., $M^{\mathcal{A}^\dagger}_\mathcal{Q} M^{\mathcal{A}}_\mathcal{Q} = \mathbb{I}$ for unitary $\mathcal{A}$, and conversely $\mathcal{A}^U_\mathcal{Q} \circ \mathcal{A}^{U^\dagger}_\mathcal{Q} = \mathcal{I}$ for unitary matrix $U$. In particular, consider the unitary:
\begin{equation} \label{eq:basis_change_unitary}
    R_{\mathcal{Q}, \mathcal{Q}'} \equiv \sum_k \kett{Q'_k}_{\mathcal{Q}' ~ \mathcal{Q}} \! \bbra{Q'_k} = \sum_k \kett{Q_k}_{\mathcal{Q} ~ \mathcal{Q}} \! \bbra{Q'_k},
\end{equation}
which carries out vectorization basis changes $R_{\mathcal{Q}, \mathcal{Q}'} \kett{O}_\mathcal{Q} = \kett{O}_{\mathcal{Q}'}$. Define the unitary superoperator:
\begin{equation}
    \mathcal{R}_{\mathcal{Q},\mathcal{Q}'}(\cdot) \equiv \frac{1}{N} \sum_k Q'_k \tr(Q_k^\dagger (\cdot))
\end{equation} 
affecting a change in operator basis from $\mathcal{Q}$ to $\mathcal{Q}'$, i.e., $\mathcal{R}_{\mathcal{Q},\mathcal{Q}'}(Q_i) = Q'_i ~\forall i$. Then:
\begin{lemma}
$R_{\mathcal{Q}', \mathcal{Q}}$ is the transfer matrix of $\mathcal{R}_{\mathcal{Q},\mathcal{Q}'}$ in the basis $\mathcal{Q}'$, i.e., $R_{\mathcal{Q}', \mathcal{Q}} = M^{\mathcal{R}_{\mathcal{Q},\mathcal{Q}'}}_{\mathcal{Q}'}$.
\end{lemma}
\begin{proof}
\begin{equation*}
    \bra{i} M^{\mathcal{R}_{\mathcal{Q},\mathcal{Q}'}}_{\mathcal{Q}'} \ket{j} = \frac{1}{N} \tr(Q_i'^\dagger \mathcal{R}_{\mathcal{Q},\mathcal{Q}'}(Q_j')) = \frac{1}{N} \sum_k \tr(Q_i'^\dagger Q_k') \tr(Q_k^\dagger Q'_j) = \bra{i} \left(\sum_k \kett{Q'_k}_{\mathcal{Q}' ~ \mathcal{Q}'} \! \bbra{Q_k} \right) \ket{j} = \bra{i} R_{\mathcal{Q}', \mathcal{Q}} \ket{j}
\end{equation*}
\end{proof}
As discussed in Section~\ref{sec:prelim}, the unitary $R_{\mathcal{Q}, \mathcal{Q}'}$ also affects a basis change for transfer matrices via conjugation (Eq.~(\ref{eq:TM_basis_change})), since:
\begin{equation*}
    \left[ M_\mathcal{Q}^\mathcal{A} \right]_{ij} = \frac{1}{N} \tr(Q_i \mathcal{A}(Q_j)) 
    = \leftindex_{\mathcal{Q}'}{\bbra{Q_i}} M^\mathcal{A}_{\mathcal{Q}'} \kett{Q_j}_{\mathcal{Q}'} 
    = \leftindex_{\mathcal{Q}}{\bbra{Q_i}} R_{\mathcal{Q}, \mathcal{Q}'}^\dagger M^\mathcal{A}_{\mathcal{Q}'} R_{\mathcal{Q}, \mathcal{Q}'} \kett{Q_j}_{\mathcal{Q}}
    = \left[ R_{\mathcal{Q}, \mathcal{Q}'}^\dagger M^\mathcal{A}_{\mathcal{Q}'} R_{\mathcal{Q}, \mathcal{Q}'} \right]_{ij}.
\end{equation*}
Straightforward algebra further yields the identities $R_{\mathcal{Q} , \mathcal{Q}'} = (R_{\mathcal{Q}' , \mathcal{Q}})^\dagger$ and $R_{\mathcal{Q}_1 , \mathcal{Q}_3} = R_{\mathcal{Q}_2 , \mathcal{Q}_3} R_{\mathcal{Q}_1 , \mathcal{Q}_2}$, where $\mathcal{Q}_1, \mathcal{Q}_2, \mathcal{Q}_3$ denote arbitrary orthonormal bases.

\subsubsection{Vectorized operators as pure states}
When the basis of vectorization $\mathcal{Q}$ is chosen appropriately, the quantum state $\kett{O(t)}_\mathcal{Q} \in \mathcal{H}'$ inherits the geometrical structure of the Heisenberg operator $O(t)$, which facilitates its preparation on quantum hardware, especially on connectivity-constrained devices.

We say that $\mathcal{Q}$ is a \textit{k-producible operator basis} (over a fixed partition $K$) if any $Q_i \in \mathcal{Q}$ can be written as a tensor product of operators that each, across the same partitions, act on at most $k$ qubits. That is, denoting $K = \{K_1,...,K_m\}$ a partition of $\{1,...,n\}$, $Q_i$ takes the form:
\begin{equation} \label{eq:k_loc_basis}
    Q_i = \bigotimes_{K_j \in K}^{\abs{K}} q_i^{K_j},
\end{equation}
where the number of partitions $\abs{K} \leq n$ and $\text{supp}(q_i^{K_j}) \leq k$, in which case we also write $\mathcal{Q} = \bigotimes_{K_j \in K} \mathcal{Q}^{K_j}$. For instance, the computational and Pauli operator bases $\mathcal{C} = \bigotimes_{j=1}^{n} \mathcal{C}^{j}$ and $\mathcal{P} = \bigotimes_{j=1}^{n} \mathcal{P}^{j}$ are both $1$-producible/fully local. 

Following the main text and as illustrated in Fig.~\ref{fig:doubled_space}, order $\mathcal{H}'$ spatially as $\mathcal{H}' =\mathcal{H}_L \otimes \mathcal{H}_R = (\bigotimes_{i=1}^n \mathcal{H}_L^i ) \otimes (\bigotimes_{i=1}^n \mathcal{H}_R^i )$. $\mathcal{H}_L$ (the left half of $\mathcal{H}'$) and $\mathcal{H}_R$ (the right half of $\mathcal{H}'$) are each of dimension $2^n$, and each $\mathcal{H}_L^i$ or $\mathcal{H}_R^i$ has local dimension 2. Choosing a mapping between local basis operators and local basis states allows the vectorized state to also retain this tensor product structure. For instance, for fully local ($k=1$) bases such as $\mathcal{C}$ and $\mathcal{P}$,
\begin{equation} \label{eq:fully_local_mapping}
    q^{(1)}_{i}\otimes ... \otimes q^{(n)}_{i} \longleftrightarrow \kett{q^{(1)}_{i} \otimes ... \otimes q^{(n)}_{i}}_\mathcal{Q} = \kett{q^{(1)}_{i}}_{\mathcal{Q}^{1}} \otimes ... \otimes \kett{q^{(n)}_{i}}_{\mathcal{Q}^{n}} = \ket{k_1} \otimes ... \otimes \ket{k_n},
\end{equation}
where each $\ket{k_i} \in \mathcal{H}^i_L \otimes \mathcal{H}^i_R$.

Consider the computational operator basis $\mathcal{C}= \{ \ketbra{i}{j}: i, j = 0,1 \}^{\otimes n}$. The choice of mapping:
\begin{equation} \label{eq:def_comp_vect}
    \ketbra{i^{1}}{j^{1}} \otimes ... \otimes \ketbra{i^{n}}{j^{n}} 
    \longleftrightarrow \left(\ket{i^{(1)}} \otimes \ket{j^{(1)}} \right) \otimes ... \otimes \left(\ket{i^{(n)}} \otimes \ket{j^{(n)}} \right)
\end{equation}
corresponds to `row-stacking' a $2^n \times 2^n$ dimensional matrix to form a $4^n$ dimensional vector. In this basis, the single-qubit identity and Pauli operators map to Bell states:
\begin{align}
    \kett{{\mathbb{I}}}_\mathcal{C} = \frac{1}{\sqrt{2}}(\ket{00} + \ket{11}), &\quad \kett{X}_\mathcal{C} = \frac{1}{\sqrt{2}}(\ket{01} + \ket{10}), \\
    \kett{Z}_\mathcal{C} = \frac{1}{\sqrt{2}}(\ket{00} - \ket{11}), &\quad \kett{Y}_\mathcal{C} = \frac{i}{\sqrt{2}}(-\ket{01} + \ket{10}),
\end{align}
so $n$-qubit Pauli operator map to tensor products of $n$ Bell states.

Next, consider the Pauli basis $\mathcal{P}$, which is also fully local, and therefore retains the tensor product structure of an operator. Single-qubit Pauli operators are mapped, up to a complex phase, to computational basis states, i.e.,
\begin{align}
    \kett{{\mathbb{I}}}_\mathcal{P} = \ket{00}, &\quad \kett{X}_\mathcal{P} = \ket{01}, \\
    \kett{Z}_\mathcal{P} = \ket{10}, &\quad \kett{Y}_\mathcal{P} = -i\ket{11}.
\end{align}
More compactly, via the symplectic representation of Pauli operators, we have:
\begin{equation} \label{eq:def_pauli_vect}
    \kett{Z^{a_z} X^{a_x}}_\mathcal{P} = \ket{a_z} \otimes \ket{a_x} \implies 
    \kett{Z^{a^1_z} X^{a^1_x} \otimes ... \otimes Z^{a^n_z} X^{a^n_x}}_\mathcal{P} = \left(\ket{a^1_z} \otimes \ket{a^1_x} \right) \otimes ... \otimes \left(\ket{a^n_z} \otimes \ket{a^n_x}\right),
\end{equation}
where the binary vector $k = (a^1_z, a^1_x, ..., a^n_z, a^n_x) \in \mathbb{Z}_2^{2n}$ indexes one of the $4^n$ Pauli operators.

Changes in the bases of vectorization are affected by the unitary $R_{\mathcal{Q} \rightarrow \mathcal{Q}'}$ of Eq.~(\ref{eq:basis_change_unitary}). The form of $R_{\mathcal{C} \rightarrow \mathcal{P}}$, determined via Eq.~(\ref{eq:basis_change_unitary}) and the fact that both bases are fully local (of the form Eq.~(\ref{eq:fully_local_mapping})), is given by Eq.~(\ref{eq:rotation_cp}), which coincides with the (state) basis change from the computational to the Bell basis. It can always be implemented as a transversal Clifford operation in depth-2. Note that the $-i$ phase in the definition of $\kett{Y}_\mathcal{P}$ is necessary for $R_{\mathcal{C} \rightarrow \mathcal{P}}$ to take this form. In Appendix~\ref{app:qudit}, we additionally show that the $\text{CNOT}$-$\text{H}$ form of $R_{\mathcal{C} \rightarrow \mathcal{P}}$ can be generalized to the case of $n$-qudit operators, with $\text{CNOT}$ and $\text{H}$ replaced by their natural qudit generalizations.

\section{Details on preparation of \texorpdfstring{$\kett{O(t)}$}{|O(t)>>}} \label{app:details_state_prep}
Here, we detail the preparation of the vectorized state $\kett{O(t)}$ discussed in Section~\ref{sec:preparation}, firstly on gate-based digital quantum computers, and secondly on analog quantum simulators based on tunable Hamiltonian parameters. Appendices~\ref{app:channel_gen} and \ref{app:qudit} discuss generalizations to time-evolution under general unital channels and $n$ qudits respectively.

\subsection{Preparation on digital quantum computers}
Consider the setting where an initial Heisenberg operator $O$ evolves to $O(t) = U^\dagger O U$ under the action of a general unitary $U$ that may correspond to a sequence $U = \prod_{l} U_l$ of implementable gates $U_l$, or the propagator $U=e^{-iHt}$ of a Hamiltonian $H$. 

We will discuss conditions for the efficient preparation of $\kett{O(t)}_\mathcal{Q}$ for more general choices of initial operators and operator bases, beyond simple initial operators (such as Pauli operators) and operator bases such as the Pauli basis $\mathcal{P}$ where the basis change unitary to/from $\mathcal{C}$ can be implemented exactly. We will show that it can be achieved under the conditions considered in the main text, i.e.:
\begin{proposition} \label{prop:vect_prep}
$\kett{O(t)}_\mathcal{Q}$ can be prepared to $\epsilon$ error in trace distance with $L_U + \bigo{\textup{poly(n)} \textup{polylog}(n/\epsilon)}$ gates if:
    \begin{enumerate}
    \item $O(t=0)$ can be expanded as a sum of $\bigo{\textup{poly}(n)}$-many terms in a $k$-producible, orthogonal operator basis $\mathcal{Q}_i$ with $k=\bigo{\log(n)}$, or is a unitary that can be implemented using $\bigo{\textup{poly}(n)}$ gates,
    \item $U^\dagger$, and $U^\mathsf{T}$ or $U$ can each be implemented to error $\epsilon/6$ with $L_U$ gates,
    \item and the final basis $\mathcal{Q}$ has a tensor product structure over partitions of at most size $\bigo{\log(n)}$. 
\end{enumerate}
\end{proposition}
\begin{proof}
Following the discussion of Section~\ref{sec:preparation}, to prove this statement, begin by writing $\kett{O(t)}_\mathcal{Q} = R_{\mathcal{C}, \mathcal{Q}} M^U_{\mathcal{C}} R_{\mathcal{Q}_i, \mathcal{C}} \kett{O}_{\mathcal{Q}_i}$ so that state preparation is composed of the following general steps:
\begin{enumerate}
    \item Preparation of initial encoded operator $\kett{O}_{\mathcal{Q}_i}$ in the basis $\mathcal{Q}_i$.
    \item Vectorization basis change between $\mathcal{Q}_i$ and $\mathcal{C}$ ($\kett{O}_{\mathcal{C}} = R_{\mathcal{Q}_i,\mathcal{C}} \kett{O}_{\mathcal{Q}_i}$), and between $\mathcal{C}$ and $\mathcal{Q}$ ($\kett{O}_{\mathcal{Q}} = R_{\mathcal{C}, \mathcal{Q}}, \kett{O(t)}_{\mathcal{C}}$).
    \item Time evolution $\kett{U^\dagger O U}_{\mathcal{C}} = M^U_{\mathcal{C}} \kett{O}_{\mathcal{C}} = U^\dagger \otimes U^\mathsf{T} \kett{O}_{\mathcal{C}}$ in $\mathcal{C}$.
\end{enumerate}

Denote the approximately prepared initial and final states as $\widetilde{\kett{O}}_{\mathcal{Q}_i}$ and $\widetilde{\kett{O(t)}}_{\mathcal{Q}}$, and the approximately implemented basis changes and time evolution as $\widetilde{R}_{\mathcal{Q}_i,\mathcal{C}}$, $\widetilde{R}_{\mathcal{C},\mathcal{Q}}$, and $\widetilde{M}^U_{\mathcal{C}}$ respectively. Using $\norm{\prod_l \tilde{U}_l - \prod_l {U}_l}_\infty \leq \sum_l \norm{\tilde{U}_l - U_l}_\infty$, i.e., that their errors accumulate linearly, the total state preparation error $\epsilon$ is:
\begin{align*}
    \epsilon = \norm{\widetilde{\kett{O(t)}}_\mathcal{Q} - \kett{O(t)}_\mathcal{Q} }_2 &= \norm{ \widetilde{R}_{\mathcal{C},\mathcal{Q}} \widetilde{M}^U_{\mathcal{C}} \widetilde{R}_{\mathcal{Q}_i,\mathcal{C}} \widetilde{\kett{O}}_{\mathcal{Q}_i} - R_{\mathcal{C}, \mathcal{Q}} M^U_{\mathcal{C}} R_{\mathcal{Q}_i, \mathcal{C}} \kett{O}_{\mathcal{Q}_i} }_2 \\
    & \leq \epsilon_1 + \epsilon_2 + \epsilon_3 + \epsilon_4,
\end{align*}
with the error due to each step defined as:
\begin{alignat*}{2}
\epsilon_1   & \equiv \norm{ \widetilde{\kett{O}}_{\mathcal{Q}_i} - \kett{O}_{\mathcal{Q}_i} }_2, && ~~~~~ \epsilon_2 \equiv \norm{ \widetilde{R}_{\mathcal{Q}_i,\mathcal{C}} - R_{\mathcal{Q}_i, \mathcal{C}} }_\infty,\\
\epsilon_3   & \equiv \norm{ \widetilde{R}_{\mathcal{C},\mathcal{Q}} - R_{\mathcal{C}, \mathcal{Q}} }_\infty, && ~~~~~ \epsilon_4 \equiv \norm{ \widetilde{M}^U_{\mathcal{C}} -  M^U_{\mathcal{C}} }_\infty,
\end{alignat*}
where $\norm{\cdot}_2$ denotes the Euclidean distance and $\norm{\cdot}_\infty$ the spectral/operator norm for operators. We will proceed to describe each step in detail in Subsections \ref{app:initial_sparse_prep}, \ref{app:vect_basis_change_err}, and \ref{app:time_evolution} below; individually bounding their errors (with the choice $\epsilon_1 = 0$, $\epsilon_2 = \epsilon_3 = \epsilon_4 = \epsilon/3$) proves Proposition~\ref{prop:vect_prep}. 

Finally, we relate the Euclidean distance $\norm{\widetilde{\kett{O(t)}}_\mathcal{Q} - \kett{O(t)}_\mathcal{Q} }_2$ to the trace distance:
\begin{align}
    D_{\mathrm{tr}}\left(\widetilde{\kett{O(t)}}_\mathcal{Q} , \kett{O(t)}_\mathcal{Q} \right)&\equiv\sqrt{1 - \abs{_\mathcal{Q}\widetilde{\llangle{O(t)}}\kett{O(t)}_\mathcal{Q} }^2}
    \\&\le \sqrt{2-2\abs{_\mathcal{Q}\widetilde{\llangle{O(t)}}\kett{O(t)}_\mathcal{Q} }} \label{eq:fid_to_l2_step1}\\
    &= \min_{\alpha\in\mathbb R}\ \Big\|\widetilde{\kett{O(t)}}_\mathcal{Q}-e^{i\alpha}\kett{O(t)}_\mathcal{Q}\Big\|_2 \label{eq:fid_to_l2_step2}\\
    &\le \Big\|\widetilde{\kett{O(t)}}_\mathcal{Q}-\kett{O(t)}_\mathcal{Q}\Big\|_2, \label{eq:fid_to_l2_step3}
\end{align}
where in \eqref{eq:fid_to_l2_step1} we used the elementary inequality $1-s^2\le 2-2s$ valid for all $s\in[0,1]$, in \eqref{eq:fid_to_l2_step2} we used the identity
$\min_{\alpha\in\mathbb R}\|\ket\psi-e^{i\alpha}\ket\phi\|_2^2=2-2|\langle\psi|\phi\rangle|$ for normalized vectors, and in \eqref{eq:fid_to_l2_step3} we used that the minimum over $\alpha$ is upper bounded by the value at $\alpha=0$.
\end{proof}
At the end of this subsection, we discuss how queries to $U^\mathsf{T}$ can be replaced by queries to $U$. 

\subsubsection{Preparation of initial state} \label{app:initial_sparse_prep}
We wish to prepare the state $\kett{O}_Q$, and we are either given (1) a classical description of the initial operator $O = \sum_k c_k Q_k$ in terms of the coefficients $c_k$ of its decomposition over basis operators $Q_k \in \mathcal{Q}$, in which only $s$ terms are non-zero, or (2) the ability to implement $O$ efficiently.

By definition of the vectorization map in the basis $\mathcal{Q}$, Eq.~(\ref{eq:vect_map}), $\kett{O}_Q \propto \sum_k c_k \ket{k}$, a linear combination of $s$ computational basis states. We can therefore directly make use of existing exact state preparation algorithms that prepare arbitrary $s$-sparse quantum states, which generically require resources that scale at most as $\bigo{ns}$. For this purpose, many constructions exist, which exhibit tradeoffs across circuit gate count, depth, ancilla qubits, and T-gate count \cite{gleinig2021efficient,zhang2022quantum,li2024nearly,vilmart2025resource}, possibly achieving $\bigo{\log(ns)}$ scaling in one or more of the metrics, and at most $\bigo{ns}$ in the rest. For instance, \cite{gleinig2021efficient} describes exact state preparation algorithms taking $\bigo{ns}$ gates, without the need for ancillas, which is sufficient for our situation. By assumption (1), i.e. that $s=\bigo{\text{poly}(n)}$, $\kett{O}_{\mathcal{Q}_i}$ can be prepared exactly using $\bigo{\text{poly}(n)}$ gates.

Alternatively, if $O$ is a unitary operator that can be efficiently implemented (such as a Pauli operator, or a linear combination of anticommuting Pauli operators \cite{zhao2020measurement}), we can make use of Eq.~(\ref{eq:ricochet}) directly to prepare $\kett{O}_C$ by first preparing $n$ Bell pairs $\kett{\mathbb{I}^{\otimes n}}$ (which can be achieved in depth-2 regardless of $n$), and applying the unitary $O \otimes \mathbb{I}^{\otimes n}$ or $\mathbb{I}^{\otimes n} \otimes O^\mathsf{T}$ to it. This can be straightforwardly extended, e.g. to sparse linear combinations of implementable unitaries using the linear combination of unitaries technique \cite{childs2012hamiltonian}.

In physical contexts, $O$ is usually a Hermitian operator that represents a physical observable, such as the local spin in the $z$ direction ($Z_i = \mathbb{I} \otimes ... \mathbb{I} \otimes Z \otimes \mathbb{I} ... \otimes \mathbb{I}$ corresponding to a one-hot computational basis state), or the total magnetization ($\sum_{i=1}^n Z_i$, corresponding to the permutationally symmetric $n$-qubit W state: $\ket{W_n} \equiv (\ket{10...0} + \ket{010...0} + ... + \ket{0...01})/\sqrt{n}$), and so can be expressed as a sum of $\text{poly}(n)$-many Pauli operators.

\subsubsection{Vectorization basis change} \label{app:vect_basis_change_err}
Vectorization basis changes $\kett{O}_{\mathcal{Q}} \rightarrow \kett{O}_{\mathcal{Q}'}$ are necessary when the bases for initial state preparation, time evolution, and final measurement differ from one another, or if there are constraints in implementation that render working in certain bases necessary or more convenient (for instance, hardware connectivity constraints that render time-evolving in local bases more convenient, or the inability to perform discrete basis changes that necessitate time-evolving in entangled bases). 

Given an encoded state $\kett{O}_\mathcal{Q}$, a change in the basis of vectorization to $\mathcal{Q}'$ can always be affected by the unitary $R_{\mathcal{Q},\mathcal{Q}'}$ of Eq.~(\ref{eq:basis_change_unitary}). The depth of $R_{\mathcal{Q},\mathcal{Q}'}$ depends on the structures of $\mathcal{Q}$ and $\mathcal{Q}'$, possibly scaling exponentially as $n$ in the worst case if it involves a global change in entanglement. However, when $\mathcal{Q}$ and $\mathcal{Q}'$ are related by a basis change that is only locally entangling, $R_{\mathcal{Q},\mathcal{Q}'}$ admits efficient implementation.

The simplest instance where this occurs is when $\mathcal{Q}$ and $\mathcal{Q}'$ both have tensor product structures over partitions of small sizes, i.e. are $k$-producible with small $k$. An important example is the basis change between the computational and Pauli bases, where the basis operators in both bases are fully separable ($k=1$), and the basis change unitary $R_{\mathcal{C},\mathcal{P}}$ is precisely the Bell basis transformation Eq.~(\ref{eq:rotation_cp}). Notably, $R_{\mathcal{C},\mathcal{P}}$ can be implemented (exactly) as a local, transversal circuit in depth 2, and is a Clifford circuit, rendering its implementation particularly easy. The CNOT-H structure of Eq.~(\ref{eq:rotation_cp}) also generalizes to the case of qudits -- c.f. Appendix~\ref{app:qudit}.

More generally, Let $\mathcal{Q}$ be a $k$-producible operator basis of the form Eq.~(\ref{eq:k_loc_basis}). Subsequently, writing $R_{\mathcal{Q},\mathcal{Q}'} = R_{\mathcal{C},\mathcal{Q}'} R_{\mathcal{Q},\mathcal{C}} $ ($\mathcal{C}$ can be replaced by any other $1$-producible basis), $R_{\mathcal{Q},\mathcal{C}}$ (and $R_{\mathcal{C},\mathcal{Q}'}$) takes the form:
\begin{equation} \label{eq:local_rotation_tensor}
    R_{\mathcal{Q},\mathcal{C}} = \bigotimes_{j \in K}^{\abs{K}} R_{\mathcal{Q}^{j}, \bigotimes_{j' \in P_j} \mathcal{C}^{j'}}
\end{equation}
where each unitary $R_{\mathcal{Q}^{j}, \bigotimes_{j' \in P_j} \mathcal{C}^{j'}}$ acts locally on the $\abs{K_j} \leq 2k$ qubits in $(\bigotimes_{i \in P_j} \mathcal{H}_L^i ) \otimes (\bigotimes_{i \in P_j} \mathcal{H}_R^i )$. 

As a result, $R_{\mathcal{Q},\mathcal{Q}'}$ can always be efficiently approximated on a quantum computer using a discrete set of universal gates as long as the partitions are small, i.e. of size $k = \bigo{\log(n)}$:
\begin{lemma} \label{thm:basis_change_approx}
Let $\mathcal{Q}$ and $\mathcal{Q}'$ be operator bases that are both $k$-producible, with $k=\bigo{\log(n)}$. Then, there exists a unitary $\widetilde{R}_{\mathcal{Q},\mathcal{Q}'}$ of depth $L=\bigo{\textup{poly}(n)\textup{polylog}(n/\epsilon)}$ consisting of 1- and 2-qubit/local gates such that:
\begin{equation}
   \norm{\tilde{R}_{\mathcal{Q},\mathcal{Q}'} - R_{\mathcal{Q},\mathcal{Q}'}}_\infty \leq \epsilon,
\end{equation}
where $\norm{\cdot}_\infty$ denotes the operator/spectral norm.
\end{lemma}
% $L=\bigo{\text{poly}(n)\text{polylog}(\text{poly}(n)/\epsilon)}$ $\bigo{k^24^k \log(k^24^k/(\epsilon/2n))}$
\begin{proof}
Let $P$ and $P'$ be the partitions of $\mathcal{Q}$ and $\mathcal{Q}'$ respectively. Writing $R_{\mathcal{Q},\mathcal{Q}'} = R_{\mathcal{C},\mathcal{Q}'} R_{\mathcal{Q},\mathcal{C}}$, each unitary has the tensor product structure of Eq.~(\ref{eq:local_rotation_tensor}). Since errors accumulate linearly over tensor and matrix products, it suffices to implement each local unitary $R_{\mathcal{Q}^{j}, \bigotimes_{j' \in P_j} \mathcal{C}^{j'}}$ (and $R_{\bigotimes_{j' \in P_j} \mathcal{C}^{j'}, \mathcal{Q}'^{j}}$) in parallel to error at most $\epsilon/2n$.

Standard circuit approximation constructions allow general $k$-qubit unitaries to be implemented with $\bigo{k^2 4^k \text{polylog}(k^2 4^k / \epsilon)}$ gates \cite{nielsen2010quantum}. Given that each local unitary acts on at most on $2k$ qubits, the total approximation error can be bounded to at most $\epsilon$ using $L = \bigo{k^2 16^k \text{polylog}(k^2 16^k n/ \epsilon)}$ gates from a universal gateset. In particular, when $k=\bigo{\textup{log}(n)}$, $L=\bigo{\textup{poly}(n)\textup{polylog}(n/\epsilon)}$, and when $k=\bigo{1}$, $L=\bigo{\text{polylog}(n/\epsilon)}$.
\end{proof}

As a result of Lemma~\ref{thm:basis_change_approx}, conditions (1) and (3) of Proposition~\ref{prop:vect_prep} respectively ensure that $\epsilon_2$ (due to the initial basis change $R_{\mathcal{Q}_i, \mathcal{C}}$) and $\epsilon_3$ (due to the final basis change $R_{\mathcal{C}, \mathcal{Q}}$) can be bounded by $\epsilon/3$ using $\bigo{\textup{poly}(n) \textup{polylog}(n/\epsilon)}$ gates.

\subsubsection{Time evolution} \label{app:time_evolution}
Following the discussion of Section~\ref{sec:preparation}, time evolution in $\mathcal{C}$ is affected by the unitary $M^\mathcal{U}_\mathcal{C} = U^\dagger \otimes U^\mathsf{T}$. On a digital quantum computer, it can be implemented by simply reversing the order of the gates in the decomposition $U= \prod_{l} U_l$, and applying $U_l^\dagger \otimes U_l^\mathsf{T}$ sequentially. Implementing $U^\dagger$ and $U^\mathsf{T}$ to error $\epsilon/6$ each then bounds the time evolution error to $\epsilon_4 = \epsilon/3$, completing the proof of Proposition~\ref{prop:vect_prep}. 

$M^\mathcal{U}_\mathcal{C}$ inherits the circuit depth and geometry of $U$/$U^\dagger$, which are the main implementation considerations on near-term, non-error-corrected quantum computers. Furthermore, it inherits their Cliffordness; this is crucial for their implementation on fault-tolerant error-corrected quantum computers, where magic resources (injected via protocols such as magic state distillation \cite{bravyi2005universal}) are often regarded as the leading resource overhead \cite{litinski2019game}. Together with the fact that state preparation and vectorization basis changes can also be fully Clifford operations when switching between $\mathcal{C}$ and $\mathcal{P}$ via Eq.~(\ref{eq:rotation_cp}), we find that the overall preparation of $\kett{O(t)}_\mathcal{P}$ is also magic-resource-preserving. In Appendix~\ref{app:channel_gen}, we describe how our framework can be extended to simulate Heisenberg time evolution under general quantum channels $O \rightarrow \mathcal{E}^\dagger(O)$, using a probabilistic construction based on block encodings.

Next, we describe several features of time evolution in the Heisenberg picture -- distinct and absent from the Schrödinger picture -- that lead to simplifications in implementation. Note that they also apply to the $n$-qubit protocol, detailed in Appendix~\ref{app:n_qubit_details}. 
\begin{itemize}
    \item \textit{Fast-forwarding initial Clifford unitaries:} Time evolution $\kett{O}_\mathcal{P} \rightarrow \kett{U^\dagger P U}_\mathcal{P}$ is generally implemented as $\kett{O(t)}_\mathcal{P} = R_{\mathcal{C} \rightarrow \mathcal{P}} (U^\dagger \otimes U^\mathsf{T}) R_{\mathcal{P} \rightarrow \mathcal{C}} \kett{O}_\mathcal{P}$. In the case where the initial operator $O$ is a Pauli operator and $U$ is a Clifford unitary decomposable into $l=\bigo{\text{poly}(n)}$ single and two-qubit Clifford gates, this transformation can be efficiently `fast-forwarded' and implemented with just two layers of single-qubit gates. 
    
    Observing that Clifford unitaries simply permute a Pauli operator $P$ to another Pauli operator $Q$ (up to a phase), it suffices to track the permutations sequentially in classical memory, and finally apply single-qubit $X$ and phase gates that permutes between the computational basis states $\kett{P}_\mathcal{P}$ and $\kett{Q}_\mathcal{P}$. The global phase can alternatively be stored in classical memory, avoiding the need to implement the phase gates.
    
    This is distinct from the Schrödinger picture; while initial Clifford gates transform stabilizer states to other stabilizer states, it may be difficult to prepare arbitrary highly entangled stabilizer states on quantum computers. This feature can be taken advantage of by compilation strategies that commute Clifford gates to the beginning of the circuit, at the expense of increasing the connectivity of the remaining non-Clifford gates \cite{litinski2019game}.

    \item \textit{Preservation of lightcone structure:} Vectorization fully preserves the lightcone structure of Heisenberg dynamics. Dynamics in the Heisenberg picture is unital (i.e. identity preserving), which manifests in via vectorization as the invariance of the Bell state $\kett{\mathbb{I}}_\mathcal{C} = \frac{1}{\sqrt{2}} (\ket{00} + \ket{11})$ to unitary transformations of the form $U_1 \otimes U_2$, where $U_1U_2^\mathsf{T} = \mathbb{I}$, which is satisfied by $U^\dagger \otimes U^\mathsf{T}$. Due to this property, unitaries outside the lightcone of an initial local operator act trivially on the vectorized state, allowing them to be removed from the quantum circuit. In other words, the circuit $U_L ... U_0 \kett{O}$ can be simplified by commuting $U_i$'s located outside the lightcone of $O$ to the beginning of the circuit and removed via $U_i \kett{\mathbb{I}} = \kett{\mathbb{I}}$. 
    
    In practice, the non-participating qubits can be left idling or initialized only when they intersect the lightcone, both of which reduce gate errors. This feature can also be applied to the $n$-qubit approach, which was also implemented in e.g. \cite{mi2021information,google2025observation}, where idling qubits were subjected to dynamical decoupling gates to reduce idling error.

\end{itemize}

\subsubsection{Avoiding queries to \texorpdfstring{$U^\mathsf{T}$}{U T}} \label{app:avoiding_transpose}

We briefly discuss how queries to the transposed unitary $U^\mathsf{T}$ can be replaced with $U$, at the price of having to implement $U$, $O$, and $U^\dagger$ in series, effectively doubling the depth of the time evolution step in Proposition~\ref{prop:vect_prep}.

This is a direct consequence of the ricochet identity Eq.~(\ref{eq:ricochet}), which implies that $\kett{O(t)}_\mathcal{Q} =R_{\mathcal{C},\mathcal{Q}} \left( U^\dagger O U \otimes \mathbb{I} \right) \kett{\mathbb{I}}_\mathcal{C}$. That is, it suffices to apply the unitary $U^\dagger O U$ to $\mathcal{H}_L$, followed by $R_{\mathcal{C},\mathcal{Q}}$ to $n$ Bell pairs $\kett{\mathbb{I}}_\mathcal{C}$. Note that this requires $O$ to be unitary (e.g. Pauli operators, or linear combinations of anticommuting Pauli operators). The more general case can be achieved using e.g. the linear combination of unitaries technique \cite{childs2012hamiltonian}.

\subsection{Considerations for analog quantum simulators}
The above section assumes the ability to implement discrete gates from a universal gateset (as is the natural case for digital quantum computers), in which case $\kett{O}_\mathcal{Q}$ can be efficiently prepared, as long as the conditions of Proposition~\ref{prop:vect_prep} are satisfied. In the case of analog quantum simulators and certain experimental platforms that implement dynamics via tunable Hamiltonian parameters, it is more natural to reduce the state preparation problem to a Hamiltonian simulation problem $e^{-i H' t}$ that is tractable for the analog simulator \cite{glaetzle2017coherent,notarnicola2023randomized}. 

Firstly, making use of the identity Eq.~(\ref{eq:ricochet}), it can be straightforwardly shown that $\kett{O(t)}_\mathcal{C}$ satisfies the Schrödinger equation $ih \partial_t \kett{O(t)}_\mathcal{C} = \overline{H} \kett{O(t)}_\mathcal{C}$, where $\overline{H} \equiv -H \otimes \mathbb{I} + \mathbb{I} \otimes H^\mathsf{T}$. Its solution is thus given by:
\begin{equation}
    \kett{O(t)}_\mathcal{C} = e^{-i \overline{H} t} \kett{O}_\mathcal{C} = e^{iHt} \otimes e^{-iH^\mathsf{T} t} \kett{O}_\mathcal{C},
\end{equation}
which amounts to independently evolving the subsystem $\mathcal{H}_L$ backward in time with $H$, and $\mathcal{H}_R$ forward in time with $H^\mathsf{T}=H^*$. Introducing an intermediate operator basis $\mathcal{Q}$, Eq.~(\ref{eq:state_prep}) can then be written as:
\begin{alignb}
    \kett{O(t)}_\mathcal{Q} &= R_{\mathcal{Q}, \mathcal{C}}^\dagger e^{-i \overline{H} t} R_{\mathcal{Q}, \mathcal{C}} \kett{O}_\mathcal{Q} \\
    &= e^{-i(R_{\mathcal{Q}, \mathcal{C}}^\dagger \overline{H} R_{\mathcal{Q}, \mathcal{C}})t} \kett{O}_\mathcal{Q} \\
    & \equiv e^{-i \overline{H}_\mathcal{Q} t} \kett{O}_\mathcal{Q} \\
    &= e^{i(R_{\mathcal{Q}, \mathcal{C}}^\dagger (H \otimes \mathbb{I}) R_{\mathcal{Q}, \mathcal{C}})t} e^{-i(R_{\mathcal{Q}, \mathcal{C}}^\dagger (\mathbb{I} \otimes H^\mathsf{T}) R_{\mathcal{Q}, \mathcal{C}})t} \kett{O}_\mathcal{Q}, \label{eq:analog_ham_sim_separated}
\end{alignb}
where in the third line we defined the rotated Hamiltonian:
\begin{equation*}
    \overline{H}_\mathcal{Q} \equiv R_{\mathcal{Q}, \mathcal{C}}^\dagger \overline{H} R_{\mathcal{Q}, \mathcal{C}}
    = -R_{\mathcal{Q}, \mathcal{C}}^\dagger (H \otimes \mathbb{I}) R_{\mathcal{Q}, \mathcal{C}} + R_{\mathcal{Q}, \mathcal{C}}^\dagger (\mathbb{I} \otimes H^\mathsf{T}) R_{\mathcal{Q}, \mathcal{C}},
\end{equation*}
and in the fourth line made use of $[H \otimes \mathbb{I}, \mathbb{I} \otimes H^\mathsf{T}] = 0$. We conclude that the preparation of $\kett{O(t)}_\mathcal{Q}$ hinges on whether the simulation of $\overline{H}_\mathcal{Q}$ and the preparation of $\kett{O}_\mathcal{Q}$ are efficient in some basis $\mathcal{Q}$. \\

Several considerations for the simulation of $\overline{H}_\mathcal{Q}$ are as follows:
\begin{itemize}
    \item \textit{Locality of $\overline{H}_\mathcal{Q}$}: The conjugation of $H \otimes \mathbb{I}$ and $\mathbb{I} \otimes H^\mathsf{T}$ by a non-local basis transformation $R_{\mathcal{Q}, \mathcal{C}}$ generally leads to $\overline{H}_\mathcal{Q}$ having increased non-locality. However, if $R_{\mathcal{Q}, \mathcal{C}}$ has a tensor product structure over local partitions of constant size $p$ -- which, by Proposition~\ref{thm:basis_change_approx}, occurs precisely when $\mathcal{Q}$ has a tensor product structure over the same sized partitions -- this increase is only by the same constant $p$. For example, if $H$ (and therefore $\overline{H}$) is $k$-local (i.e., it can be written as a sum $H = \sum_k c_k P_k$ of Pauli operators $P_k$ that are supported non-trivially on at most $k$ sites), then $\overline{H}_\mathcal{Q}$ is at most $(k+p)$-local. For example, a 2-local Hamiltonian $H$ is transformed into a 4-local Hamiltonian $H'_\mathcal{P}$ in the Pauli basis. The increased connectivity imposes implementational difficulties, but can nonetheless be accounted by e.g. resolving them into tractable lower-order terms \cite{glaetzle2017coherent}.

    \item \textit{Time-reversal}: Simulating the first term in Eq.~(\ref{eq:analog_ham_sim_separated}) generally requires time-reversal. As discussed in the previous section, on digital quantum computers this simply amounts to reversing the order of the gates and implementing their inverses. The same maneuver in analog simulators require the ability to invert the signs of all terms in $\overline{H}_\mathcal{Q}$. Whether this is possible generally depends on the form of $\overline{H}_\mathcal{Q}$ and the tunability of the physical platform, and has been achieved in many experimental platforms based on superconducting qubits, trapped ions, neutral atoms, and nuclear spins \cite{schuster2022many,garttner2017measuring,mi2021information,sanchez2022emergent,brown2023quantum}.
    
    \item \textit{Transpose}: Simulating the second term in Eq.~(\ref{eq:analog_ham_sim_separated}) generally requires simulating time evolution under $H^\mathsf{T}$. Similar to time-reversal, this is straightforward on digital quantum computers, but on analog simulators require inverting the signs of terms in $H$ that contain an odd number of $Y$ Pauli operators. Notably, time-reversal symmetric Hamiltonians with $H=H^\mathsf{T}$ (and thus $U=U^\mathsf{T}$) -- e.g. real-entried Hamiltonians, or those containing only Pauli operators with $X$ and $Z$ terms -- are unaffected. Note that the discussion in Appendix~\ref{app:avoiding_transpose} also holds here, allowing the transposition to be avoided under certain conditions.
\end{itemize}

The initial state $\kett{O}_\mathcal{Q}$ can be prepared e.g. by adiabatic state preparation, as long as $\kett{O}_\mathcal{Q}$ corresponds to the ground state of a simulatable Hamiltonian \cite{notarnicola2023randomized} (e.g. the preparation of the computational basis state $\kett{O}_\mathcal{P}$ where $O\in\mathcal{P}$, in which case a local Hamiltonian consisting of only $Z$ terms suffice).

\section{Details on algorithms for OTOCs / superoperator expectation values} \label{app:superop}
This appendix extends the discussion of Section~\ref{sec:superop} on the computation of the superoperator expectation values $\tr(O^\dagger \mathcal{A}(O))/2^n$ of self-adjoint superoperators (satisfying $\mathcal{A}=\mathcal{A}^\dagger$). As discussed there, they map to Hermitian transfer matrices, and therefore can be interpreted as the Heisenberg-picture analogue of physically measurable observables in the Schrödinger picture. This observation -- which is another consequence of the bijective isometry realized by the vectorization map -- not only enables the individual computation of the expectation values of self-adjoint superoperators on quantum computers via Eq.~(\ref{eq:super_expt_vect}), but also their simultaneous estimation via commutation.

Throughout this section, we assume the operator $O$ to be Hermitian for simplicity, allowing us to drop the adjoints. We also take the Hilbert space dimension to be $N = 2^n$.

\subsection{Diagonal superoperators}
We begin with superoperators diagonal in the Pauli basis, of the form Eq.~(\ref{eq:superop_diagonal}), which encode \textit{all} geometrical properties of $O$. Firstly, we have the following Lemma:
\begin{lemma} \label{lemma:diag_superop}
Let $\mathcal{D}$ be a superoperator. The following statements are equivalent:
\begin{enumerate}
    \item $M^\mathcal{D}_\mathcal{P}$ is a diagonal matrix with diagonal entries $[M^\mathcal{D}_\mathcal{P}]_{kk} = \lambda_k$.
    \item $M^\mathcal{D}_\mathcal{C} = \sum_k f_k P_k \otimes P_k^\mathsf{T}$.
    \item $\mathcal{D}$ has a diagonal operator-sum decomposition, i.e., $\mathcal{D}(\cdot) = \sum_k f_k P_k (\cdot) P_k$.
    \item $\mathcal{D}$ can be written as a sum of projectors onto $P_k$'s, i.e., $\mathcal{D}(\cdot) = \frac{1}{2^n} \sum_k \lambda_k P_k \tr(P_k (\cdot))$.
\end{enumerate}
The vectors $\vec{f}$ and $\vec{\lambda}$ are related by a Walsh-Hadamard transform $\vec{\lambda} = K \vec{f}$ and its inverse $\vec{f} = \frac{1}{4^n} K \vec{\lambda}$, with entries $K_{ij} \in \{-1,1\}$ taking value 1 when $[P_i,P_j]=0$ and $-1$ otherwise.
\end{lemma}

\begin{proof}
    \noindent ($1 \implies 2$): Write the diagonal transfer matrix as $M^\mathcal{D}_\mathcal{P} = \sum_k \lambda_k \ketbra{k}{k} = \sum_k \lambda_k \kett{P_k}_\mathcal{P} \leftindex_{\mathcal{P}}{\bbra{P_k}}$. Conjugating it by the basis change unitary $R_{\mathcal{P} \rightarrow \mathcal{C}} = R_{\mathcal{C} \rightarrow \mathcal{P}}^\dagger$ yields:
    \begin{equation}
        M^\mathcal{D}_\mathcal{C} = R^\dagger_{\mathcal{C} \rightarrow \mathcal{P}} M^\mathcal{D}_\mathcal{P} R_{\mathcal{C} \rightarrow \mathcal{P}} 
        = \sum_k \lambda_k R^\dagger_{\mathcal{C} \rightarrow \mathcal{P}} \kett{P_k}_\mathcal{P} \leftindex_{\mathcal{P}}{\bbra{P_k}} R_{\mathcal{C} \rightarrow \mathcal{P}} 
        = \sum_k \lambda_k \kett{P_k}_\mathcal{C} \leftindex_{\mathcal{C}}{\bbra{P_k}}.
    \end{equation}
    Expanding $M^\mathcal{D}_\mathcal{C}$ in the Pauli basis $\{P_i \otimes P_j\}_{i,j=1}^{4^n,4^n}$ and substituting the previous expression, we have:
    \begin{align*}
        M^\mathcal{D}_\mathcal{C} &= \frac{1}{4^n}\sum_{ij} \tr((P_i \otimes P_j) M^\mathcal{D}_\mathcal{C}) P_i \otimes P_j \\
        &= \frac{1}{4^n}\sum_{ij} \tr((P_i \otimes P_j) \sum_k \lambda_k \kett{P_k}_\mathcal{C} \leftindex_{\mathcal{C}}{\bbra{P_k}}) P_i \otimes P_j \\
        &= \frac{1}{4^n}\sum_{ijk} \lambda_k \leftindex_{\mathcal{C}}{\bbra{P_k}} P_i \otimes P_j \kett{P_k}_\mathcal{C} P_i \otimes P_j \\
        &= \frac{1}{4^n}\sum_{ijk} \lambda_k \tr(P_k^\dagger P_i P_k P_j^*)/2^n P_i \otimes P_j \\
        &= \frac{1}{4^n}\sum_{ijk} \lambda_k \left( K_{ik} \delta_{ij} \psi(P_j) \right) P_i \otimes P_j \\
        &= \sum_i \left( \frac{1}{4^n} \sum_k \lambda_k K_{ik} \right) P_i \otimes P_i^* \\
        &= \sum_i f_i P_i \otimes P_i^*,
    \end{align*}
    where $\psi(P_j)$ is the phase acquired by complex-conjugating/transposing, $P_j^*=\psi(P_j) P_j$, and $K_{ij} \in \{-1,1\}$ are entries of the Walsh-Hadamard transform as defined in the Lemma statement.

    ~\newline
    \noindent $(2 \implies 3)$: Immediate via Eq.~(\ref{eq:tm_cb_identity}) relating $M_\mathcal{C}^\mathcal{D}$ with the operator-sum decomposition of $\mathcal{D}$.
    
    ~\newline
    \noindent ($3 \implies 4$): Inserting the completeness relation for operators $\mathcal{I}(\cdot) = \sum_i P_i \tr(P_i (\cdot))/2^n$, we have:
    \begin{align*}
        \mathcal{D}(\cdot) &= \frac{1}{2^n} \sum_k f_k P_k \left(\sum_i P_i \tr(P_i (\cdot)) \right) P_k \\
        &= \frac{1}{2^n} \sum_{k,i} f_k \tr(P_i (\cdot)) P_k P_i P_k \\
        &= \frac{1}{2^n} \sum_{k,i} f_k \tr(P_i (\cdot)) (K_{ik} P_i P_k) P_k \\
        &= \frac{1}{2^n} \sum_i \left(\sum_k f_k K_{ik} \right) P_i \tr(P_i (\cdot)) \\
        &= \frac{1}{2^n} \sum_i \lambda_i P_i \tr(P_i (\cdot)).
    \end{align*}
    
    ~\newline
    \noindent ($4 \implies 1$): By definition of the transfer matrix elements Eq.~(\ref{eq:tm_elements}), we have:
    \begin{equation*}
        [M^\mathcal{D}_\mathcal{P}]_{kl} = \frac{1}{4^n} \tr(P_k \sum_i \lambda_i P_i \tr(P_i P_l)) = \lambda_l \delta_{kl}.
    \end{equation*}
\end{proof}

Lemma~\ref{lemma:diag_superop} provides a connection between $\vec{\lambda}$ and $\vec{f}$, the decompositions of $\mathcal{D}$ in the operator-sum and projective representations respectively. The vectors $\vec{f} \equiv (f_1,...,f_{4^n})^\intercal$ and $\vec{\lambda} \equiv (\lambda_1,...,\lambda_{4^n})^\intercal$ are related by a Walsh-Hadamard transform, where $K$ is symmetric and proportional to an orthogonal matrix, with entries $K_{ij} \in \{-1,1\}$ taking value 1 when $[P_i,P_j]=0$ and $-1$ otherwise. Related discussions on the Walsh-Hadamard transform and the relationship between the projective and operator-sum decomposition of superoperators can be found e.g. in \cite{flammia2020efficient,nambu2005matrix}. 

Alternatively, that the PTM of diagonal superoperators are diagonal in the computational basis can be inferred directly from the fact that the basis transformation unitary $R_{\mathcal{C} , \mathcal{P}}$ only conjugates diagonal operators of the form $P \otimes P^*$ to Pauli operators containing only $I$ or $Z$ terms:
\begin{table}[h!]
\renewcommand{\arraystretch}{1.4}
\centering
\begin{tabular}{|l|l|l|l|}
\hline
\textbf{II $\rightarrow$ $\,\text{II}$} & IX $\rightarrow$ $\text{IX}$ & IZ $\rightarrow$ $\text{XZ}$ & IY $\rightarrow$ $\,\text{XY}$ \\
XI $\rightarrow$ $\text{ZX}$ & \textbf{XX $\rightarrow$ $\text{ZI}$} & XZ $\rightarrow$ $\text{YY}$ & XY $\rightarrow$ $-\,\text{YZ}$ \\
ZI $\rightarrow$ $\text{XI}$ & ZX $\rightarrow$ $\text{XX}$ & \textbf{ZZ $\rightarrow$ $\text{IZ}$} & ZY $\rightarrow$ $\,\text{IY}$ \\
YI $\rightarrow$ $-\,\text{YX}$ & YX $\rightarrow$ $-\,\text{YI}$ & YZ $\rightarrow$ $\,\text{ZY}$ & \textbf{YY $\rightarrow$ $-\text{ZZ}$} \\
\hline
\end{tabular}
\caption{Conjugation of Pauli operators by $R_{\mathcal{C} , \mathcal{P}} (\cdot) R^\dagger_{\mathcal{C} , \mathcal{P}}$.}
\label{tab:conjugation}
\end{table}

Note that Lemma~\ref{lemma:diag_superop} does not hold for arbitrary operator bases $\mathcal{Q}$ -- the 2-qubit superoperator $Z\mathbb{I}(\cdot)ZZ$ does not admit a diagonal operator-sum expansion in the computational operator basis $\mathcal{C}$, but $M^\mathcal{A}_\mathcal{C} = Z\mathbb{I} \otimes ZZ$ is diagonal. Nevertheless, a straightforward partial generalization to bases that satisfy the (looser) condition $Q_i Q_j = \phi_{ij} Q_j Q_i$ (more generally known as q-commutation, satisfied by e.g. Majorana operators \cite{xu2025dynamics,miller2025simulation}) can be made:
\begin{lemma} \label{app:tm_diagonal}
Let $\mathcal{Q} = \{ Q_i\}$ be an orthogonal operator basis such that $Q_i Q_j = \phi_{ij}Q_j Q_i$ for $i \neq j$, and $\tr(Q_k^\dagger Q_l)= \delta_{kl} N$. If a superoperator $\mathcal{A}$ can be expanded as $\mathcal{A}(\cdot) = \sum_k f_k Q_k(\cdot)Q_k^\dagger$, then its transfer matrix in the same basis $M^\mathcal{A}_\mathcal{Q}$ is diagonal, with entries $\lambda_i \equiv [M^\mathcal{A}_\mathcal{Q}]_{ii} = \sum_{k} f_{k} \phi_{ki}$.
\end{lemma}
\begin{proof}
The matrix elements of $M^{\mathcal{A}}_\mathcal{Q}$ are:
\begin{equation*} \label{eq:walsh-had-sum}
    \bra{i} M^{\mathcal{A}}_\mathcal{Q} \ket{j} = \frac{1}{N} \trace(Q_i^\dagger \mathcal{A}(Q_j)) 
    = \frac{1}{N} \sum_{k} f_{k} \tr(Q_i^\dagger Q_k Q_j Q_k^\dagger) 
    = \frac{1}{N} \sum_{k} f_{k} \phi_{kj} \tr(Q_i^\dagger Q_j Q_k Q_k^\dagger) 
    = \delta_{ij} \sum_{k} f_{k} \phi_{ki}. 
\end{equation*}
\end{proof}

\subsection{Commuting superoperators and the structure of their common eigenbasis} \label{app:general_commutation}
Transfer matrices inherit the notion of commutation, Eq.~(\ref{eq:commutation_inheritance}); this is because $[M^\mathcal{A}, M^\mathcal{B}] \kett{O} = \kett{[\mathcal{A}, \mathcal{B}](O)} = M^{[\mathcal{A}, \mathcal{B}]} \kett{O}$ for arbitrary $O$, so $M^{[\mathcal{A}, \mathcal{B}]} = [M^\mathcal{A}, M^\mathcal{B}]$, which vanishes if and only if $[\mathcal{A}, \mathcal{B}] = 0$. Consequently, as discussed in Section~\ref{sec:superop}, the expectation values $\{\tr(O^\dagger \mathcal{A}_i(O))/2^n\}_i$ of commuting superoperators $\{\mathcal{A}_i \}$ can be simultaneously estimated, by sampling from $\kett{O(t)}_\mathcal{C}$ in the (not necessarily unique) common eigenbasis of their transfer matrices $\{M_\mathcal{C}^{\mathcal{A}_i}\}_i$. Given samples of $\kett{O(t)}_\mathcal{C}$, this procedure is efficient if the basis transformation can be carried out efficiently, which generally depends on the set $\{\mathcal{A}_i \}_i$.

To investigate the structure of this common eigenbasis, begin by observing that the pairwise commutation of $\{\mathcal{A}_i \}_i$ can be expressed in terms of local commutation/anticommutation relations, i.e.:
\begin{lemma}[Necessary and sufficient condition for transfer matrix commutation, from local commutation/anticommutation]\label{app:commutation_condition}~\\
Let $\{\mathcal{A}_1, ..., \mathcal{A}_M \}$ be a set of superoperators where $\mathcal{A}_i(\cdot) = P_{i} (\cdot) Q_{i}^\dagger$, $P_i,Q_i \in \mathcal{P}$. Then all of them mutually commute, i.e.:
\begin{equation} \label{eq:commutation_superop}
    [P_i \otimes Q_i^*, P_j \otimes Q_j^*] = 0
\end{equation}
for all pairs $(i,j)$, if and only if for all pairs $(i,j)$, either:
\begin{align}
    \{P_i, P_j\} = 0 ~\text{and}~ \{Q_i, Q_j\} = 0,& ~~\text{or} \label{eq:lemma_all_anticomm} \\ 
    [P_i, P_j] = 0 ~\text{and}~ [Q_i, Q_j] = 0.& \label{eq:lemma_all_commuting}
\end{align}
\end{lemma}
\begin{proof}
For any pair $(i,j)$, applying the identity:
\begin{equation}
    [P_i \otimes Q_i, P_j \otimes Q_j] = \frac{1}{2}([P_i, P_j] \otimes \{Q_i, Q_j \} + \{P_i, P_j \} \otimes [Q_i, Q_j])
\end{equation}
to the left-hand-side of Eq.~(\ref{eq:commutation_superop}) yields:
\begin{equation}
    [P_i \otimes Q_i^*, P_j \otimes Q_j^*] = \frac{1}{2}([P_i, P_j] \otimes \{Q_i, Q_j \}^* + \{P_i, P_j\} \otimes [Q_i, Q_j]^*).
\end{equation}
Since Pauli operators either commute or anticommute, this expression is zero if and only if either Eq.~(\ref{eq:lemma_all_anticomm}) or Eq.~(\ref{eq:lemma_all_commuting}) holds.
\end{proof}

In other words, the left-acting operator pairs $P_i, P_j$ and right-acting operator pairs $Q_i, Q_j$ must simultaneously commute or anticommute. This allows us to classify sets of commuting superoperators $\{\mathcal{A}_i \}$ into those that commute because:
\begin{enumerate}[label=(\roman*)]
    \item all left factors commute pairwise and all right factors commute pairwise (i.e. $\forall ~ i \neq j$, $[P_i,P_j]=0$ and $[Q_i,Q_j]=0$), or 
    \item there exists at least one pair $\mathcal{A}_i$, $\mathcal{A}_j$ where their left factors anticommute and their right factors anticommute (i.e. $\exists ~ i \neq j$ such that $\{P_i,P_j \}=0$ and $\{Q_i,Q_j \}=0$).
\end{enumerate}
Fig.~\ref{fig:comutation} illustrates some examples. Further examples for these two cases and their common eigenbases are as follows.

\begin{figure}[h]
    \centering
    \includegraphics[width=.45\linewidth]{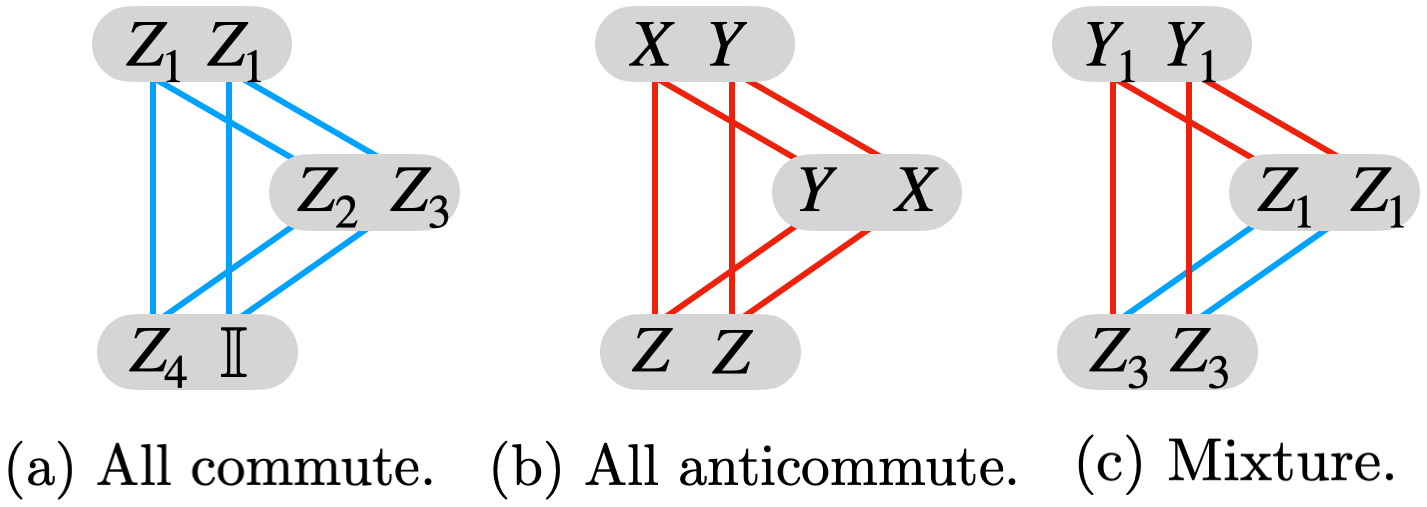}
    \caption{\textbf{Examples of superoperator sets that commute due to different reasons.} Grey capsules denote superoperators $P_i(\cdot)Q_i$, while blue (red) lines denote commutation (anticommutation) between the left or right-acting operators. (a) $\{Z_1(\cdot)Z_1, Z_2(\cdot)Z_3, Z_4(\cdot)\mathbb{I}\}$ commute due to simultaneous commutation. (b) $\{X(\cdot)Y, Y(\cdot)X, Z(\cdot)Z\}$ commute due to simultaneous anticommutation. (c) $\{Y_1(\cdot)Y_1, Z_1(\cdot)Z_1, Z_3(\cdot)Z_3\}$ commute due to a mixture of both cases.}
    \label{fig:comutation}
\end{figure}

\noindent \textbf{All simultaneously commute (Fig.~\ref{fig:comutation}(a)):} 
\begin{itemize}
    \item $\{\tr(O Z_i O Z_j)/2^n\}_{i,j=1}^{2^n}$ where $Z_i \in \{\mathbb{I},Z \}^{\otimes n}$, with common eigenbasis $\{\ket{0}, \ket{1}\}^{\otimes 2n}$.
    \item $\{\tr(O X_i O X_j)/2^n\}_{i,j=1}^{2^n}$ where $X_i \in \{\mathbb{I},X \}^{\otimes n}$, with common eigenbasis $\{\ket{+}, \ket{-}\}^{\otimes 2n}$.
    \item $\{\tr(O Z_i O X_j)/2^n\}_{i,j=1}^{2^n}$ with common eigenbasis $\{\ket{0}, \ket{1}\}^{\otimes n} \otimes \{\ket{+}, \ket{-}\}^{\otimes n}$.
\end{itemize}

Notably, this includes off-diagonal OTOC terms (where left- and right-acting terms differ). In each case, all $4^n$ terms are diagonal in the same basis. \\

\noindent \textbf{At least one pair simultaneously anticommute (Fig.~\ref{fig:comutation}.(b) \& (c)):}
\begin{itemize}
    \item $\{\tr(O P_i O P_i)/2^n\}_{i=1}^{4^n}$ where $P_i \in \mathcal{P}$, with common eigenbasis being the Bell basis which entangles $\mathcal{H}_L$ and $\mathcal{H}_R$. Contains both simultaneously commuting and anticommuting pairs.
    
    \item $\{\tr(O X O Y)/2^n, \tr(O Y O X)/2^n, \tr(O Z O Z)/2^n \}$, with common eigenbasis entangling $\mathcal{H}_L$ and $\mathcal{H}_R$. All pairs simultaneously anticommute.
\end{itemize}

From the above examples, we observe that the common eigenbasis in case (a) is separable across $\mathcal{H}_L$ and $\mathcal{H}_R$, while for cases (b) and (c) are entangled across $\mathcal{H}_L$ and $\mathcal{H}_R$. Indeed, this observation is general :
\begin{itemize}
    \item On one hand, any eigenstate of the eigenbasis in cases (b) and (c) above must be entangled across $\mathcal{H}_L$ and $\mathcal{H}_R$. Thus, simultaneous computation requires measuring $\kett{O(t)}_\mathcal{C}$ in a basis that is entangled across them. An important example is again the case of the diagonal superoperators discussed in Section~\ref{sec:geometrical_prop}, where their common eigenbasis -- the Bell basis -- entangles $\mathcal{H}_L$ and $\mathcal{H}_R$ via the Bell basis unitary Eq.~(\ref{eq:rotation_cp}).

    \item On the other hand, for case (a) there exists a common eigenbasis that is separable across $\mathcal{H}_L$ and $\mathcal{H}_R$, which is simply the tensor product of the common eigenbasis of the left-acting operators $\{P_i\}$ and that of the right-acting operators $\{Q_i\}$. As a consequence, this tensor product form allows an `unfolding' of the current algorithm to a $n$-qubit algorithm, allowing this class of quantities to be computed simultaneously with the same sample complexity, but using only $n$ qubits and classical randomization (c.f. Section~\ref{sec:n_qubit_rand}, App.~\ref{app:n_qubit_algorithms}).
\end{itemize}
In other words, we have the following:
\begin{lemma}[Entanglement structure of common eigenbasis] \label{lemma:anticomm_basis}
Let $S = \{P_i\otimes Q_i\}_{i=1}^M$ be a set of pairwise commuting Pauli operators in $\mathcal{L}(\mathcal{H}_L \otimes \mathcal{H}_R)$. 
\begin{enumerate}
    \item If there exists $i\neq j$ with $\{P_i,P_j\}=0$ and $\{Q_i,Q_j\}=0$, then any simultaneous eigenvector of $S$ is entangled across $\mathcal{H}_L$ and $\mathcal{H}_R$.
    \item On the other hand, if $\forall ~i\neq j$, $[P_i,P_j]=0$ and $[Q_i,Q_j]=0$, then $S$ admits a common eigenbasis where all eigenstates are separable across $\mathcal{H}_L \otimes \mathcal{H}_R$, namely the basis $\{\ket{p_k} \otimes \ket{q_l}\}_{k,l=1}^{2^n}$ where $\{\ket{p_k} \}_{k=1}^{2^n}$ ($\{\ket{q_l} \}_{l=1}^{2^n}$) is the common eigenbasis of $\{P_i\}$ ($\{Q_i\}$).
\end{enumerate}
\end{lemma}
\begin{proof}
(1): Suppose, for the sake of contradiction, that $\ket{\psi} = \ket{\psi_L}\otimes\ket{\psi_R}$ is a simultaneous eigenvector of $\{P_i\otimes Q_i\}_{i=1}^M$. Then:
\begin{align}
    (P_i \otimes Q_i) \bigl(\ket{\psi_L}\otimes\ket{\psi_R}\bigr) &= \lambda_i\,\ket{\psi_L}\otimes\ket{\psi_R}, \\
    (P_j \otimes Q_j )\bigl(\ket{\psi_L}\otimes\ket{\psi_R}\bigr) &= \lambda_j\,\ket{\psi_L}\otimes\ket{\psi_R}.
\end{align}
For the equalities to hold, $\ket{\psi_L}$ ($\ket{\psi_R}$) must be the common eigenstate of $P_i$ and $P_j$ ($Q_i$ and $Q_j$). This contradicts the assumption that $\{P_i,P_j\}=0$ and $\{Q_i,Q_j\}=0$, since anticommuting Pauli operators have no common eigenstates.

(2): The states $\ket{p_k} \otimes \ket{q_l}$ form a complete basis, and are common eigenstates of all operators in $S$ since $(P_i \otimes Q_i) \ket{p_k} \otimes \ket{q_l} \propto \ket{p_k} \otimes \ket{q_l}$.
\end{proof}

A consequence of Lemma~\ref{lemma:anticomm_basis} is that $n$-qubit approaches are unable to simultaneously estimate OTOCs resulting from sets of Pauli operators satisfying the first (pairwise anticommutation) condition by sampling from their common eigenbasis, while on the contrary those satisfying the second (pairwise commutation) condition are jointly computable; c.f. Appendix~\ref{app:n_qubit_details}. 

\vspace{5pt}

Finally, as mentioned in the main text, the anticommutation between transfer matrices can analogously be used to group terms together as unitaries for their simultaneous estimation, enabling the use of unitary grouping strategies for measurement reduction \cite{izmaylov2019unitary,zhao2020measurement}. We have the negation of Lemma~\ref{app:commutation_condition}:
\begin{lemma}[Necessary and sufficient condition for transfer matrix anticommutation, from local commutation/anticommutation]\label{app:anticommutation_condition}~\\
Let $\{\mathcal{A}_1, ..., \mathcal{A}_M \}$ be a set of superoperators where $\mathcal{A}_i(\cdot) = P_{i} (\cdot) Q_{i}^\dagger$, $P_i,Q_i \in \mathcal{P}$. Then all distinct pairs anticommute, i.e.:
\begin{equation} \label{eq:anticommutation_superop}
    \{P_i \otimes Q_i^*, P_j \otimes Q_j^* \} = 0
\end{equation}
for all pairs $(i,j)$, if and only if for all pairs $(i,j)$, either:
\begin{align}
    \{P_i, P_j\} = 0 ~\text{and}~ [Q_i, Q_j] = 0,& ~~\text{or}  \\ 
    [P_i, P_j] = 0 ~\text{and}~ \{Q_i, Q_j \} = 0.& 
\end{align}
\end{lemma}
\begin{proof}
The result follows by negating Lemma~\ref{app:commutation_condition}. Alternatively, the identity:
\begin{equation*}
    \{P_i \otimes Q_i, P_j \otimes Q_j \} = \frac{1}{2}(\{P_i, P_j \} \otimes \{Q_i, Q_j \} + [P_i, P_j ] \otimes [Q_i, Q_j])
\end{equation*}
can be used; it vanishes if and only if the above conditions are satisfied.
\end{proof}

Subsequently, given a superoperator of the form $\mathcal{A} = \sum_i c_i \mathcal{A}_i$ where $c_i \in \mathbb{R}$ and $\sqrt{\sum_i \abs{c_i}^2} = 1$, if all $\mathcal{A}_i$'s satisfy the condition of the above Lemma, $\mathcal{A}$ (and its transfer matrix) will be unitary, enabling $\langle \mathcal{A} \rangle_O$ to be measured as an overlap \cite{izmaylov2019unitary} or diagonalized \cite{zhao2020measurement}.

\subsection{Additional numerics and comparison of sample complexities} \label{app:superop_comparison}

This appendix discusses the number of samples needed to estimate general superoperator expectation values $\langle \mathcal{A} \rangle_{O}$, where $\mathcal{A}(\cdot) = \sum_{kl} f_{kl} P_k (\cdot) P_l^\dagger$, and $P_k, P_l \in \mathcal{P}$. We compare between the $n$- and $2n$-qubit approaches by making use of existing results on measurement reduction for energy/observable estimation, together with simple numerical experiments.

As discussed in Section~\ref{sec:superop}, the mapping of this problem to the estimation of the expectation values of Hermitian operators enables measurement optimization strategies for the latter to be applied. Examples are term grouping by commutation (as discussed in Section~\ref{subsec:self-adjoint}) and unitary term grouping by anticommutation (as discussed in the previous subsection); see e.g. \cite{tilly2022variational}, Section 5 in the context of energy computation in variational quantum eigensolvers. 

% The optimal allocation of $n_i$'s -- obtained via Lagrange multipliers \cite{rubin2018application,crawford2021efficient} -- is proportional to the variances, i.e. $n_i \propto \sqrt{\text{Var}_O{(\mathcal{A}^i)}}$, for a total of $N = \left(\sum_{i=1}^K \sqrt{\text{Var}_O(\mathcal{A}^i)} \right)^2/\epsilon^2$ samples.

As an illustration, we apply known results to our situation. Consider estimation to root-mean-squared error $\epsilon_{\text{RMS}}$. Suppose that $\mathcal{A}$ has been grouped into $M$ groups of mutually commuting superoperators $\mathcal{A}^1,...,\mathcal{A}^M$ as in Eq.~(\ref{eq:superoperator_term_group}) where each $\mathcal{A}^i = \sum_{j=1}^{K_i} \mathcal{A}_j^i$ contains $K_i$ OTOC terms $ \mathcal{A}_j^i (\cdot) = \norm{\mathcal{A}_j^i} P^i_j(\cdot)Q^i_j$ that mutually commute, i.e. $[\mathcal{A}_k^i, \mathcal{A}_l^i]=0 ~ \forall k,l$. Each term $\langle \mathcal{A}^i\rangle_O$ is then computed by taking $n_i$ samples from $\kett{O}_\mathcal{C}$ in the common eigenbasis of $\{P^i_j \otimes Q^i_j\}_{j=1}^{K_i}$, for a total of $N_\text{comm} \equiv \sum_i n_i$ samples. The optimal allocation of $n_i$'s -- given by $n_i \propto \sqrt{\text{Var}(\mathcal{A}_i)}$ via Lagrange multipliers \cite{rubin2018application,crawford2021efficient}, and achieves $N_{\text{comm}} = (\sum_i \text{Var}(\mathcal{A}_i)_O)^2/\epsilon^2_{\text{RMS}}$ -- require variance information. In the more practical situation where this is absent, each group can be allocated shots proportional to their weights \cite{arrasmith2020operator}, i.e.:
\begin{equation} \label{eq:wds_alloc}
    n_i = N_\text{comm} \frac{\sum_j \norm{\mathcal{A}^i_j}}{\sum_{i,j} \norm{\mathcal{A}^i_j}} \implies N_\text{comm} = \frac{\sum_{i,j} \norm{\mathcal{A}^i_j}}{\epsilon_{\text{RMS}}^2} \mathlarger{\sum}_{i=1}^M \frac{\text{Var}(\mathcal{A}^i)_O}{\sum_{j=1}^{K_i} \norm{\mathcal{A}^i_j}}.
\end{equation}
Similar analyses for estimators based on naive term-by-term estimation that do not exploit commutativity (denoted $N_{\text{naive}}$) and one that directly samples in the eigenbasis of $\mathcal{A}$ (denoted $N_{\text{full comm}}$) yields:
\begin{equation}
    N_{\text{naive}} = \frac{\sum_{i,j} \norm{\mathcal{A}^i_j}}{\epsilon_{\text{RMS}}^2} \mathlarger{\sum}_{i,j} \frac{\text{Var}(\mathcal{A}^i_j)_O}{\norm{\mathcal{A}^i_j}}, ~~~ N_{\text{full comm}} = \frac{\text{Var}(\mathcal{A})_O}{\epsilon_{\text{RMS}}^2}
\end{equation}
respectively. Applications of the subadditivity of the standard deviation and the triangle and Cauchy-Schwarz inequalities then yield $N_{\text{naive}} \geq N_\text{comm} \geq N_{\text{full comm}}$, as one would intuitively expect. 

The ratios between the different sample complexities quantify their relative performance. To remove the dependence on $O$, the number of samples needed can be averaged over the uniform spherical measure to provide an estimate of the average performance of the above estimators \cite{crawford2021efficient}. This allows the variances in their expressions to be replaced by $ \text{Var}(\mathcal{A})_O \rightarrow \sum_{ij} \norm{\mathcal{A}^i_j}^2$, $\text{Var}(\mathcal{A}^i)_O \rightarrow \sum_j \norm{\mathcal{A}^i_j}^2$, and $\text{Var}(\mathcal{A}^i_j)_O \rightarrow \norm{\mathcal{A}^i_j}^2$ for more tractable analysis. We use overbars ($\overline{N}_{\text{naive}}, \overline{N}_\text{comm}, \overline{N}_{\text{full comm}}$) to indicate the averaged number of samples. 

\begin{center}
    ~ \\ \textbf{Application to size superoperator} ~\\
\end{center}

\begin{figure}[h]
    \centering
    \includegraphics[width=1\linewidth]{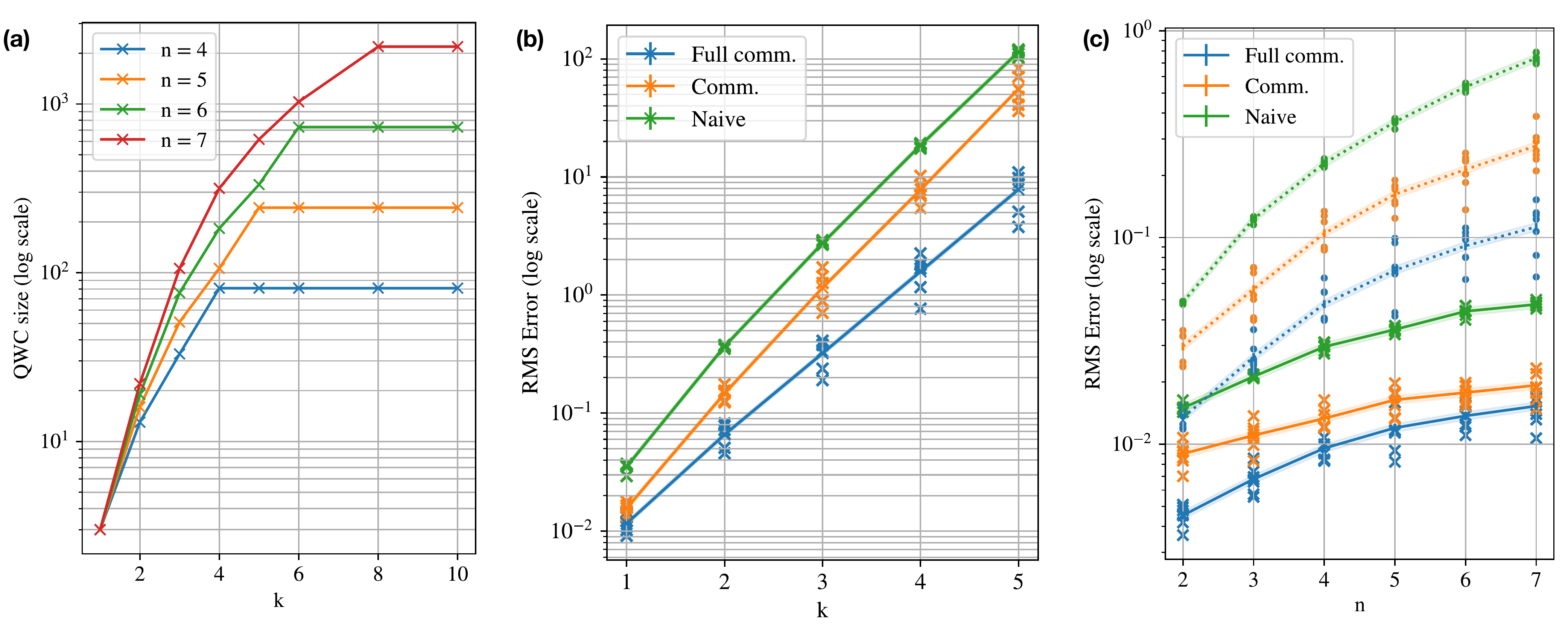}
    \caption{\textbf{Numerical experiments on the computation of the size superoperator.} (a) Number of qubit-wise-commuting groups in $\mathcal{S}$ (in log scale) as a function of $k$ for the size superoperator for different problem sizes $n=4,5,6,7$ (in different colors). Notably, from $k=2$ onwards the number of QWC groups increases with $n$. (b) RMS error $\epsilon$ (in log scale) as a function of $k$ for the size superoperator for problem size $n=5$. (c) RMS error $\epsilon_{\text{RMS}}$ (in log scale) as a function of problem size $n$ for $k=1$ (solid lines) and $k=2$ (dotted lines). In (b) and (c), curves for full commutation representing the vectorization approach (blue), limited commutation grouping representing the $n$-qubit approach (yellow), and naive term-by-term estimation (green) are displayed. Crosses represent different values of $t$, while the solid line represents their average. Each point is obtained with $N=10^3$ samples, averaged over 500 realizations; the shaded regions indicate the standard error.}
    \label{fig:size_numerics}
\end{figure}

To illustrate, we compare between the $n$- and $2n$-qubit approaches by applying the above discussion to the size superoperator $\mathcal{S}$ of Eq.~(\ref{eq:size_superop}) and its higher order moments $\langle \mathcal{S}^k \rangle_{O}$.

Since $\mathcal{S}$ is diagonal in the Pauli basis, the $2n$-qubit approach achieves the best scaling in sample complexity of $N_{\text{full comm}} = \text{Var}(\mathcal{S})_{O}/\epsilon_{\text{RMS}}^2$ samples (by directly sampling from $\kett{O}_\mathcal{P}$). 

For the $n$-qubit approach, first separate the $3n$ non-identity terms of Eq.~(\ref{eq:size_superop}) into three groups $\mathcal{S}= \frac{3n}{4}\mathbb{I} - \frac{1}{4}(\mathcal{S}_x + \mathcal{S}_y + \mathcal{S}_z)$ with:
\begin{equation*}
    \mathcal{S}_x = \sum_{j=1}^n X^{(j)}(\cdot)X^{(j)}, ~~ \mathcal{S}_y = \sum_{j=1}^n Y^{(j)}(\cdot)Y^{(j)}, ~~ \mathcal{S}_z = \sum_{j=1}^n Z^{(j)}(\cdot)Z^{(j)},
\end{equation*}
where $X^{(j)}, Y^{(j)}, Z^{(j)}$ are the weight-1 Pauli $X,Y,Z$ operators at position $j$. $\mathcal{S}_x, \mathcal{S}_y, \mathcal{S}_z$ can then be separately measured via the $n$-qubit approach by applying Algorithm~\ref{alg:n_qubit_sampler} with the operator bases:
\begin{equation*}
    \{H^{\otimes n}\ketbra{i}{j} H^{\otimes n}\}, ~~ \{(S^\dagger H)^{\otimes n}\ketbra{i}{j}(S^\dagger H)^{\otimes n}\}, ~~ \{\ketbra{i}{j}\}
\end{equation*}
respectively. Following the discussion of Appendix~\ref{app:n_qubit_algorithms}, this procedure is equivalent to separately estimating the observables $M^{\mathcal{S}_z}_\mathcal{C}$, $M^{\mathcal{S}_x}_\mathcal{C}$, and $M^{\mathcal{S}_y}_\mathcal{C}$ by sampling from $\kett{O}_\mathcal{C}$ in the bases:
\begin{equation*}
    \{H^{\otimes n}\ket{i}\otimes H^{\otimes n}\ket{j}\}, ~~ \{(S^\dagger H)^{\otimes n}\ket{i} \otimes (HS^\dagger)^{\otimes n}\ket{j}\}, ~~\{\ket{i}\otimes\ket{j}\}
\end{equation*}
respectively, where each group belongs to the `all simultaneously commute' case (Fig.~\ref{fig:comutation}(a)). We allocate samples to each group based on their weights, following Eq.~(\ref{eq:wds_alloc}).

Averaged over the uniform spherical measure, the ratio of the number of samples needed by the $n$-qubit approach to that of the $2n$-qubit approach (to estimate $\langle \mathcal{S} \rangle_{O(t)}$ fixed error $\epsilon$) is $\overline{N}_{\text{comm}}/\overline{N}_{\text{full comm}} = 3$. This implies that the sample complexity of the $2n$-qubit approach is smaller by a constant factor over the $n$-qubit approach on average. 

As one considers higher order moments $\langle \mathcal{S}^k \rangle_{O}$ -- relevant, for example, in the presence of dissipation \cite{schuster2023operator} and for the characterization of hydrodynamic behavior \cite{von2018operator} -- the number of terms and qubit-wise-commuting (QWC) groups generally increases exponentially with $k$, leading to an increase in the average performance ratio $\overline{N}_{\text{comm}}/\overline{N}_{\text{full comm}}$; c.f. Fig.~\ref{fig:size_numerics}.(a). Indeed, for $k \geq 2$, the number of qubit-wise-commuting groups of $\mathcal{S}^k$ begins to depend on system size $n$, leading $\overline{N}_{\text{comm}}/\overline{N}_{\text{full comm}}$ to no longer be constant in $n$. The $2n$-qubit approach is then significantly favorable compared to the $n$-qubit approach for large problem sizes. 

Fig.~\ref{fig:size_numerics}.(b) and (c) displays numerical experiments that illustrate this behaviour. We consider time-evolving the Heisenberg operator $O(t=0) = Z^{(1)}$ under the 1D TFIM Hamiltonian $H = \sum_{i=1}^n Z^{(i)} + \sum_{i=1}^{n-1} X^{(i)} X^{(i+1)}$ to different times, and plot the RMS error $\epsilon_{\text{RMS}}$ as a function of either $k$ (Fig.~\ref{fig:size_numerics}.(b)) or $n$ (Fig.~\ref{fig:size_numerics}.(c)) for a fixed number of $N=10^3$ shots. This is done for all three approaches: naive term-by-term estimation, full commutation representing the vectorization approach, and limited commutation grouping representing the $n$-qubit approach.

Finally, we remark that diagonal superoperators without a sparse operator-sum decomposition cannot be straightforwardly treated via the $n$-qubit approach (as we have done for the $\mathcal{S}$), unless it is possible to sample from the normalized distribution proportional to the coefficients $\abs{f_{kl}}$. This appears to be true for operators such as the right boundary superoperator: they admit a simple projector decomposition $\mathcal{D}(\cdot) = \frac{1}{2^n} \sum_k \lambda_k P_k \tr(P_k (\cdot))$ with easily computable diagonal entries $\lambda_k$ (so that $M^{\mathcal{D}}_\mathcal{P} = \sum_k \lambda_k \ketbra{k}{k}$), but has exponentially many terms in its operator-sum decomposition in the Pauli basis.

\section{Details on \texorpdfstring{$\kett{U}$}{|U>>} and algorithm for two-point correlators}

\subsection{Algorithms involving \texorpdfstring{$\kett{U}$}{|U>>}} \label{app:2pc_choi_u}
As discussed in Section~\ref{sec:2pc_main}, our approach based on the encoding $O(t) \rightarrow \kett{O(t)}$ also applies fully for the encoding $U \rightarrow \kett{U}$, allowing properties of $U$ to be extracted, using the same framework and algorithms we have developed for general Heisenberg operators $O(t)$. Table~\ref{tab:features_U} below -- obtained by specializing Table~\ref{tab:features_Ot} to the case where $O(t) = U$ -- summarizes this operator-state mapping. See also Appendix~\ref{app:n_qubit_algorithms} for $n$-qubit encodings of $U$, i.e. the application of Section~\ref{sec:n_qubit_rand} to $U$. 
\begin{table*}[ht]
\renewcommand{\arraystretch}{1.4}
\setlength{\tabcolsep}{6pt}
\begin{tabular}{P{4.5cm}|P{5.75cm}|P{5.7cm}}
\thickhline
\textbf{Properties of $U$} & \textbf{Properties of $\kett{U}_\mathcal{Q}$} & \textbf{Method of extraction} \\
\thickhline

Operator distribution in $\mathcal{Q}$ & Probability distribution in computational basis & Measurement in computational basis; Appendix~\ref{app:n_qubit_details} \\

\hline

Two-point correlators/transfer matrix elements, $\tr(U^\dagger P U Q)$ & Expectation value of $P \otimes Q$, $\leftindex_{\mathcal{Q}}{\bbra{U} P \otimes Q \kett{U}}_\mathcal{Q}$ & Measurement in eigenbasis of $P \otimes Q$ and evaluation of Monte Carlo estimator; Section~\ref{sec:2pc_expt}, \cite{caro2024learning} \\

\hline

Operator amplitudes $\tr(Q_k U)/2^n$, inner products $\tr(U_1 U_2)/2^n$ & Amplitude in computational basis, inner product between states & Amplitude/inner-product estimation \\

\hline

Operator entanglement \cite{zanardi2001entanglement} & Subsystem entanglement entropy of $\kett{U}_\mathcal{Q}$ & Multivariate trace estimation, classical shadows \cite{mcginley2022quantifying} \\

\thickhline
\end{tabular}
\caption{\textbf{Summary of properties of $U$ that can be computed under our framework.} Summary of relevant properties of the unitary $U$, corresponding properties encoded in $\kett{U}_\mathcal{Q}$ as a result of the operator-state mapping, and methods for their extraction on a quantum computer. The analogue of Table~\ref{tab:features_Ot} for $\kett{U}_\mathcal{Q}$.}
\label{tab:features_U}
\end{table*}

As mentioned in Appendix~\ref{app:relation_existing}, $\kett{U}$ can also be viewed as the Choi state of the unitary $U$. The Choi state of an $n$-qubit quantum channel $\mathcal{N}$ is defined as the $2n$-qubit quantum state $J(\mathcal{N}) \equiv (\mathbb{I}^{\otimes n} \otimes \mathcal{N})(\kettbbra{\mathbb{I}^{\otimes n}}{\mathbb{I}^{\otimes n}})$. Setting $\mathcal{N}(\cdot) = \mathcal{U}(\cdot) \equiv U^\dagger (\cdot) U$ and making use of the ricochet identity Eq.~(\ref{eq:ricochet}) then yields $J(\mathcal{U}) = \kett{U}_\mathcal{C} \leftindex_{\mathcal{C}}{\bbra{U}}$. The extraction of the properties of the channel $\mathcal{N}$ from $J(\mathcal{N})$ is studied in works such as \cite{caro2022learning,mcginley2022quantifying,kunjummen2023shadow,levy2024classical}.

\subsection{Algorithms for estimating two-point correlators} \label{app:2pc_amplitude}

Here, we describe a method to estimate $c(O(t),O'(t'))$ and $c_k$ using $2n$ qubits. As discussed in Section~\ref{sec:2pc_main}, a naive approach based on the Hadamard test requires access to controlled time evolution unitaries. The algorithm below -- exploiting the unitality of Heisenberg time evolution -- estimates $c(O(t), O'(t'))$ using only access to non-controlled time evolution unitaries. This `interferometric' approach is related to the \textit{Pseudo-Choi state} recently introduced in the context of Hamiltonian learning \cite{castaneda2025hamiltonian}. The circuit diagram is shown in Fig.~\ref{fig:2pc_cirq}.

Define $c(V)$ as the unitary $V \in \mathcal{L}(H')$ controlled on the ancilla qubit, $c(V) \equiv \mathbb{I} \otimes \ketbra{0}{0} + V \otimes \ketbra{1}{1}$. Firstly, the $2n+1$ qubit state:
\begin{equation} \label{eq:pseudo_choi}
    \ket{\Psi(O)} \equiv \frac{1}{\sqrt{2}}\!\left(\kett{O' U' O(t) U'^\dagger}_{\mathcal C} \otimes \ket{1} + \kett{\mathbb{I}}_{\mathcal C} \otimes \ket{0}\right)
\end{equation}
is prepared using the following sequence of unitaries, as shown in the circuit diagram Fig.~\ref{fig:2pc_cirq}:
\begin{align*}
    \ket{0}^{\otimes(2n+1)}
    \xrightarrow{\ (R_{\mathcal P \rightarrow \mathcal C}^{\otimes n} \otimes H)\ } &
    \kett{\mathbb{I}}_{\mathcal C} \otimes \ket{+}
    \\
     \xrightarrow{\ c(O \otimes \mathbb{I})\ } &
    \frac{1}{\sqrt{2}} \left(\kett{O}_{\mathcal C} \otimes \ket{1} + \kett{\mathbb{I}}_{\mathcal C} \otimes \ket{0}\right)
    \\
     \xrightarrow{\ (U^\dagger \otimes U^{\mathsf T}) \otimes \mathbb{I}\ } &
    \frac{1}{\sqrt{2}} \left(\kett{O(t)}_{\mathcal C} \otimes \ket{1} + \kett{\mathbb{I}}_{\mathcal C} \otimes \ket{0}\right)
    \\
    \xrightarrow{\ (U' \otimes U'^* ) \otimes \mathbb{I}\ } &
    \frac{1}{\sqrt{2}} \left(\kett{U' O(t) U'^\dagger}_{\mathcal C} \otimes \ket{1} + \kett{U' U'^\dagger}_{\mathcal C} \otimes \ket{0}\right) 
    \\
    = & \frac{1}{\sqrt{2}} \left(\kett{U' O(t) U'^\dagger}_{\mathcal C} \otimes \ket{1} + \kett{\mathbb{I}}_{\mathcal C} \otimes \ket{0}\right) 
    \\
     \xrightarrow{\ c(O' \otimes \mathbb{I})\ } &
    \frac{1}{\sqrt{2}} \left(\kett{O' U' O(t) U'^\dagger}_{\mathcal C} \otimes \ket{1} + \kett{\mathbb{I}}_{\mathcal C} \otimes \ket{0}\right) = \ket{\Psi(O)}.
\end{align*}
Subsequently, measuring the ancilla qubit in the $\{\ket{+},\ket{-}\}$ basis to compute the expectation value of the Pauli $X$ operator yields the desired result as:
\begin{align*}
    \bra{\Psi(O)} \mathbb{I} \otimes X \ket{\Psi(O)} &= \frac{1}{2} \left(\bbra{O' U' O(t) U'^\dagger} \otimes \bra{0} + \bbra{\mathbb{I}} \otimes \bra{1} \right) \left(\kett{O' U' O(t) U'^\dagger} \otimes \ket{1} + \kett{\mathbb{I}} \otimes \ket{0} \right) \\
    &= \frac{1}{2} \left(\bbrakett{O' U' O(t) U'^\dagger}{\mathbb{I}} + \bbrakett{\mathbb{I}}{O' U' O(t) U'^\dagger} \right) \\
    &= \frac{1}{2^n} \tr(O'(t')^\dagger O(t)).
\end{align*}

This procedure enables $c(O(t),O'(t'))$ to be estimated to $\epsilon$ accuracy with $\bigo{1/\epsilon^2}$ samples of the $2n+1$ qubit state $\ket{\Psi(O)}$. It can be modified to compute the amplitude $\tr(P_k O(t))/2^n$ by removing the unitaries $U'$ and $U'^*$, and substituting the controlled-$O'$ with a controlled-$P_k$.

Note that it can also be straightforwardly modified to estimate Pauli probabilities $p_k = c(O(t),O'(t'))^2$; by modifying the final empirical averaging step to multiply together the results of two independent measurements on the ancilla, an unbiased estimator of $p_k$ is obtained. This avoids the need to implement a controlled-SWAP operation on $4n+1$ qubits via the SWAP test on $\kett{O(t)} \otimes \kett{O'(t')}$, but requires stricter access to $O(t)$ (querying time evolution unitaries), and doubles the number of samples needed (due to the $2\alpha-2$ samples required for each term in the empirical average, as opposed to only $\alpha-1$).

Finally, an additional algorithm for computing overlaps (and therefore fidelities) via the $n$-qubit encoding is described in Appendix~\ref{app:n_qubit_details}; it takes the same number of samples, but requires double the circuit depth (due to having to implement $U$ and $U^\dagger$ sequentially).

\section{Details on algorithms for Operator Stabilizer Entropies} \label{app:ose}

We prove Theorem~\ref{theorem:ose} by showing that the algorithm described in Section~\ref{sec:ose} leads to an estimate of $P^{(\alpha)}(O(t))$ to additive error $\epsilon$ and success probability at least $1-\delta$, i.e.:
\begin{equation} \label{eq:eps_acc_ose}
    \Pr\left(\abs{\hat{P}^{(\alpha)} - P^{(\alpha)}} \geq \epsilon\right) \leq \delta,
\end{equation}
using $M+N \geq \frac{2\alpha\log(4/\delta)}{\epsilon^2}$ samples of $\kett{O(t)}_\mathcal{P}$. 
\begin{proof}
Begin by defining the following:
\begin{itemize}
    \item $Y \equiv \frac{1}{M} \sum_{i=1}^M p^{\alpha-1}_{k_i}$ is an unbiased estimator of $P^{(\alpha)}$, without accounting for errors originating from estimating $p^{\alpha-1}_{k_i}$'s.

    \item Recall that in the SWAP test, the fidelity $p_{k_i} = \abs{\!\! ~_{\mathcal{P}} \langle \! \langle P_{k_i} \kett{O(t)}_\mathcal{P}}^2$ is proportional to expectation value of the $Z$ Pauli operator on an ancilla qubit. We will estimate $p_{k_i}$ by performing $\alpha-1$ independent measurements on the ancilla qubit, multiplying the results together, and taking the empirical average. That is, let $z^{l}_{j}(k_i) \in\{0,1\}$ denote the (shifted) measurement outcome of the $j^{\text{th}}$ sample of the $l^{\text{th}}$ factor in the multiplication, such that $\mathbb{E}[z^{l}_{j}(k_i)] = p_{k_i}$. Defining:
        \begin{equation} \label{eq:product_unbiasing}
        \hat{p}^{\alpha-1}_{k_i} \equiv \frac{1}{m_i} \sum_{j=1}^{m_i} z^{1}_{j}(k_i) ... z^{\alpha-1}_{j}(k_i),
    \end{equation}
    it is an unbiased estimator of $p^{\alpha-1}_{k_i}$.
    
    \item $\hat{P}^{(\alpha)} \equiv \frac{1}{M} \sum_{i=1}^M \hat{p}^{\alpha-1}_{k_i}$ is an unbiased estimator of $P^{(\alpha)}(O(t))$, \textit{including} errors originating from estimating $p^{\alpha-1}_{k_i}$'s.
\end{itemize}

The union bound yields:
\begin{equation}\label{eq:union_ose}
      \Pr\left(\abs{\hat{P}^{(\alpha)} - P^{(\alpha)}} \geq \epsilon\right)
      \leq \Pr\left(\abs{\hat{P}^{(\alpha)} - Y} \geq \frac{\epsilon}{2} \right) + \Pr\left(\abs{ Y - P^{(\alpha)}} \geq \frac{\epsilon}{2} \right),
\end{equation}
where the first term is due to errors in approximating the fidelities and the second term is due to empirical averaging.

Begin by bounding the second term. Applying Hoeffding's inequality with $0 \leq p^{\alpha-1}_{k_i} \leq 1$, we conclude that taking $M \geq 2\log(4/\delta)/\epsilon^2$ samples suffices to upper-bound the second term by $\delta/2$. Similarly for the first term, applying Hoeffding's inequality with $0 \leq z^{1}_{j}(k_i) ... z^{\alpha-1}_{j}(k_i) \leq 1$, we conclude that taking $m_i = m \geq 2\log(4/\delta)/(M \epsilon^2)$ suffices to upper-bound it by $\delta/2$. Combining the two bounds via Eq.~(\ref{eq:union_ose}) yields Eq.~(\ref{eq:eps_acc_ose}). The total number of samples/measurements needed for the procedure is thus:
\begin{equation}
    M + (\alpha-1)\sum_{i=1}^M m_i \geq \frac{2 \alpha \log(4/\delta)}{\epsilon^2}.
\end{equation}

Taking the logarithm of the estimated value of $P^{(\alpha)}$ and dividing by $1-\alpha$ yields the OSE $M^{(\alpha)}$. Performing a first-order Taylor expansion of $\frac{1}{1-\alpha} \log(P^{(\alpha)} + \epsilon)$, we find that the constant error $\epsilon$ translates to an error of $\left(\frac{e^{-M^{(\alpha)}(1-\alpha)}}{1-\alpha}\right) \epsilon$ in $M^{(\alpha)}$.
\end{proof}

The estimator of $p_{k_i}^{\alpha-1}$ above, based on applying a SWAP test between the $4n+1$ qubits $\kett{O(t)}_\mathcal{P} \otimes \kett{P_k}_\mathcal{P} \otimes \ket{+}$, requires a controlled unitary involving $4n+1$ qubits. It can be replaced with other estimators of $p_{k_i}^{\alpha-1}$ such as those discussed throughout our work, which exhibit tradeoffs between required quantum memory, sample complexity, or access to $O(t)$. Several possible replacements with different requirements are:
\begin{itemize}
    \item Overlap/amplitude estimation algorithm of Appendix~\ref{app:2pc_amplitude} with the state $\ket{\Psi(O)}$ (Eq.~\ref{eq:pseudo_choi}) involving an ancilla, modified for $p_{k_i}^{\alpha-1}$ (in the sense of Eq.~\ref{eq:product_unbiasing}); the number of qubits can be reduced to $2n+1$, at the expense of requiring the more stringent access to the state-preparation unitary $U^\dagger \otimes U^\mathsf{T}$, and doubled number of samples (from squaring the amplitudes).

    \item Amplitude estimation algorithm of Section~\ref{sec:2pc_main} based on directly encoding the unitary $U$ as $\kett{U}$ or $\rho(U)$. For instance, with the $n$-qubit encoding $\rho(U)$, the number of qubits can be reduced to $n$, at the expense of a doubled circuit depth (due to having to implement $U$ and $U^\dagger$ sequentially) and doubled number of samples.

    \item Overlap estimation via Direct Fidelity estimation \cite{flammia2011direct} applied to the vectorized state $\kett{O(t)}$, again modified to compute $p_{k_i}^{\alpha-1}$. Compared to the above approaches, this involves Bell-sampling from the $2n$-qubit state $\kett{O}^{\otimes 2}$, which only involves applying the (transversal) unitary Eq.~(\ref{eq:rotation_cp}), and evaluating Pauli expectation values on $\kett{O}$. 
\end{itemize}

\section{Details on \texorpdfstring{$n$}{n}-qubit approach} \label{app:n_qubit_details}

Here, we detail the $n$-qubit approach discussed in Section~\ref{sec:n_qubit_rand}, which can be viewed as a generalization of certain $n$-qubit approaches for the computation of OTOCs \cite{swingle2016measuring,mi2021information,schuster2022many,schuster2023operator,google2025observation,algorithmiq}. 

We begin by restating the protocol of Theorem~\ref{theorem:1_copy_sampling} that enables sampling from the operator distribution:
\begin{equation} \label{eq:app_op_dist}
    p(i,j) = \frac{1}{2^n}\abs{\bra{\psi_j}V\ket{\phi_i}}^2,
\end{equation}
where $V$ is an implementable unitary (which we mainly regard as the Heisenberg operator $O(t) = U^\dagger O U$, or the propagator $U$ itself), and $\{\ket{\psi_i}\}$ (with $\ket{\psi_i} \equiv  U_\psi \ket{i}$) and $\{\ket{\phi_i}\}$ (with $\ket{\phi_i} \equiv U_\phi \ket{i}$) are arbitrary orthonormal bases efficiently reachable from the computational basis via the basis change unitaries $ U_\psi$ and $U_\phi$:

\vspace{5pt}
\begin{algorithm}[H]
\linespread{1.15}\selectfont
\caption{Sampling from the operator distribution of $V$ in the basis $\{\ketbra{\phi_i}{\psi_j}\}$} \label{alg:n_qubit_sampler}
\KwData{Initial state $\ket{0}^{\otimes n}$; unitaries $V$, $U_\phi$, $U_\psi^\dagger$}
\KwResult{A tuple $(i,j)$ with $p(i,j)= |\bra{\psi_j} V\ket{\phi_i}|^2/2^n$, where $\ket{\phi_i} \equiv U_\phi \ket{i}$ and $\ket{\psi_j} \equiv U_\psi \ket{j}$} 
  \vspace*{2pt}
  \Indp Classically sample $i \sim \mathrm{Unif}(\{0,1\}^n)$ and record outcome $i$\;
  Prepare $\ket{i}$\;
  Apply $U_\phi$, $V$, $U^\dagger_\psi$ sequentially: $\ket{\chi} \gets U^\dagger_\psi V U_\phi\ket{i} $\;
  Measure $\ket{\chi}$ in the computational basis and record outcome $j$\;
  
  \textbf{Return} the tuple $(i,j)$\;
\end{algorithm}
\vspace{5pt}

Sampling from $V=O(t)$ in the basis $\mathcal{C}$ corresponds to sampling a (column, row) tuple $(i,j)$ from the matrix representation of $O(t)$ in the standard computational basis, with probabilities $p(i,j) = \abs{\bra{j}O(t)\ket{i}}^2/2^n$. In this case, Algorithm~\ref{alg:n_qubit_sampler} can be understood as first randomly selecting a column index $i$ for the column $\ket{c} \equiv O(t) \ket{i}$ (normalized due to the unitarity of $O(t)$), followed by sampling a row index $j$ according to the Born rule applied to $\ket{c}$. The equivalence of this protocol to sampling from $\kett{O(t)}_\mathcal{C}$ in the separable basis $\{\ket{\psi_i} \otimes \ket{\phi^*_j}\}$ stems from the equality $\bra{\psi_i}O(t)\ket{\phi_j} = \langle \psi_i ,\phi^*_j \kett{O(t)}_\mathcal{C}$, which can also be seen via graphical notation, shown in Fig.~\ref{fig:penrose_1copy_sampling}.

\begin{figure}[h]
    \centering
    \includegraphics[width=.6\linewidth]{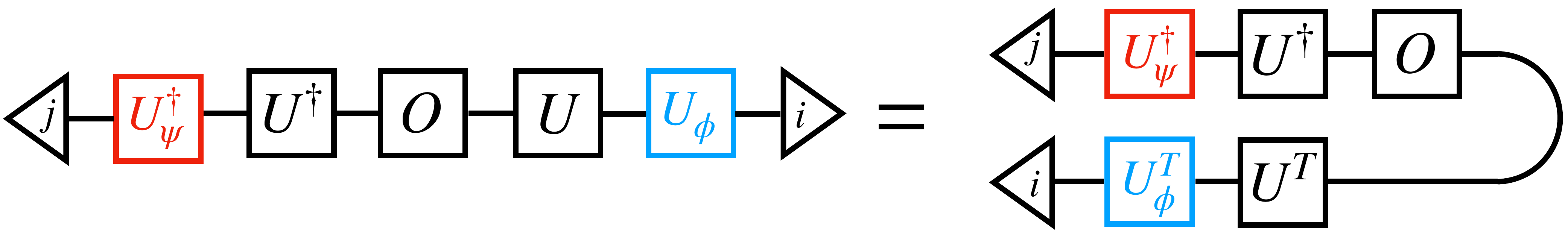}
    \caption{Diagram representing circuit for operator sampling with $n$ qubits, in a separable basis $\{\ket{\psi_i} \otimes \ket{\phi^*_j}\}$.}
    \label{fig:penrose_1copy_sampling}
\end{figure}

We remark on the fact that given samples of $\kett{O(t)}$, the basis in which we measure this state can be changed by simply appending unitaries to it before measurement, while in contrast, given samples of $\rho(V)$, the basis associated with the $\{\ket{\phi_k}\}$ part (equivalent to the $\mathcal{H}_R$ half of $\kett{O(t)}$) cannot be changed (because they are already fixed during the random initial state preparation phase).

Note that the Heisenberg-picture specific features that simplify implementations of dynamics -- discussed in Appendix~\ref{app:time_evolution}, such as gate cancellation due to lightcone structure and the ability to fast-forward initial Clifford unitaries -- also applies to this situation, which has been exploited for implementation in e.g. \cite{google2025observation}.

\subsection{Representation as classical-quantum state}
Generally, Algorithm~\ref{alg:n_qubit_sampler} can be viewed as encoding the operator distribution $p(i,j)$ of Eq.~(\ref{eq:app_op_dist}) as a classical-quantum state consisting of $n$ qubits and $n$ classical bits, of the form:
\begin{equation} \label{eq:classical_quantum_state}
    \rho(V, \{\ketbra{\phi_i}{\psi_j}\}_{i,j}) \equiv \frac{1}{2^n} \sum_{k=1}^{2^n}  \underbrace{U^\dagger_\psi V \ketbra{\phi_k}{\phi_k} V^\dagger U_\psi}_\text{quantum register} ~\otimes \!\!\!\!\!\!\! \underbrace{\ketbra{k}{k} \vphantom{U^\dagger_\psi}}_\text{classical register} \!\!\!\!\!\!\!\!\!\!\! .
\end{equation}
The form of the above state is obtained by first representing the initial random sampling and conditioned state preparation as the classical-quantum state $\frac{1}{2^n} \sum_k \ketbra{\phi_k}{\phi_k} \otimes  \ketbra{k}{k}$, applying the unitary $V \otimes \mathbb{I}$ to yield $\frac{1}{2^n} \sum_k V\ketbra{\phi_k}{\phi_k}V^\dagger \otimes  \ketbra{k}{k}$, and finally applying the basis rotation unitary $U_\psi^\dagger \otimes \mathbb{I}$. We just write $\rho(V)$ for brevity when there is no ambiguity or the choice of basis $\{\ketbra{\phi_i}{\psi_j}\}_{i,j})$ is irrelevant.

Measurements on this state in the basis $\{\ket{i} \otimes \ket{j}\}$ yields $\tr( \ketbra{j}{j}\otimes\ketbra{i}{i}\rho(V) ) = p(i,j)$, recovering the measurement statistics of Algorithm~\ref{alg:n_qubit_sampler}. Similarly, given operators $A$ and $B$ diagonal in the bases $\{\ket{\psi_i}\}$ and $\{\ket{\phi_j}\}$ respectively, we have:
\begin{equation*}
    \tr((U^\dagger_\psi A U_\psi \otimes U^\dagger_\phi B U_\phi) \rho(V, \{\ketbra{\phi_i}{\psi_j}\}_{i,j})) = \tr(V^\dagger A V B)/2^n,
\end{equation*}
where $U^\dagger_\psi A U_\psi$ and $U^\dagger_\phi B U_\phi$ are diagonal in the computational basis.

The relationship between Algorithm~\ref{alg:n_qubit_sampler} and the $2n$-qubit vectorization approach can be made more explicit via the state $\rho(V)$ as follows. Through the vectorization approach, the state that yields the statistics $p(i,j)$ when measured in the basis $\{\ket{i} \otimes \ket{j}\}$ -- which can be identified diagrammatically via Fig.~\ref{fig:penrose_1copy_sampling} -- is $\ket{W} \equiv (U_\psi^\dagger \otimes U_\phi^{T}) \kett{V}_\mathcal{C}$. Firstly, observe that given any orthonormal basis $\{\ket{\phi_j}\}$, inserting a resolution of identity and using $\kett{\ketbra{a}{b}}_\mathcal{C} = \ket{a} \otimes (\bra{b})^\mathsf{T} = \ket{a} \otimes \ket{b^*}$ yields $\left(\sum_j \ketbra{\phi_j}{\phi_j} \otimes \mathbb{I} \right) \kett{\mathbb{I}}_\mathcal{C} = \sum_k \ket{\phi_j} \otimes \ket{\phi_j^*}$, and therefore:
\begin{equation} \label{eq:app_id_expansion}
    \kett{\mathbb{I}}_\mathcal{C} \leftindex_{\mathcal{C}}{\bbra{\mathbb{I}}} = \frac{1}{2^n}\sum_{k=1}^{2^n} \ketbra{\phi_k}{\phi_k} \otimes \ketbra{\phi_k^*}{\phi_k^*} + \frac{1}{2^n} \sum_{k \neq l} \ketbra{\phi_k}{\phi_l} \otimes \ketbra{\phi_k^*}{\phi_l^*}.
\end{equation}
$\ketbra{W}{W}$ can subsequently be written as:
\begin{align}
    \ketbra{W}{W} &= \frac{1}{2^n} \sum_{k=1}^{2^n} U_\psi^\dagger V \ketbra{\phi_k}{\phi_k} V^\dagger U_\psi \otimes U_\phi^{T} \ketbra{\phi^*_k}{\phi^*_k} U_\phi^{*} +  \frac{1}{2^n} \sum_{k \neq l} U_\psi^\dagger V \ketbra{\phi_k}{\phi_l} V^\dagger U_\psi \otimes U_\phi^{T} \ketbra{\phi^*_k}{\phi^*_l} U_\phi^{*} \\
    &= \underbrace{\frac{1}{2^n} \sum_{k=1}^{2^n} U_\psi^\dagger V \ketbra{\phi_k}{\phi_k} V^\dagger U_\psi \otimes \ketbra{k}{k}}_\text{$=\rho(V)$} +  \underbrace{\frac{1}{2^n} \sum_{k \neq l} U_\psi^\dagger V \ketbra{\phi_k}{\phi_l} V^\dagger U_\psi \otimes \ketbra{k}{l}}_\text{$\equiv \sigma_{\text{CQ}}$},
\end{align}
where we used Eq.~(\ref{eq:app_id_expansion}) in the first equality and $U_\phi^* \ket{i} = \ket{\phi^*_i}$ in the second equality, finding that $\ketbra{W}{W}$ can be written as a sum of the classical-quantum state $\rho(V)$ and another `off-diagonal' state $\sigma_{CQ}$. Equivalently, under the action of the completely dephasing map on $\mathcal{H}_R$, $\Delta_R(\cdot) \equiv \sum_{i=1}^{2^n} (\mathbb{I} \otimes \ketbra{i}{i}) (\cdot)(\mathbb{I} \otimes \ketbra{i}{i})$, we recover $\Delta_R(\ketbra{W}{W}) = \rho(V)$. 

When measured in the basis $\{\ket{i} \otimes \ket{j}\}$, $\tr(\ketbra{i}{i}\otimes\ketbra{j}{j} \sigma_{CQ})$ vanishes; $\rho(V)$ therefore fully encodes the probability distribution $\tr(\ketbra{j}{j}\otimes\ketbra{i}{i} \ketbra{W}{W}) = p(i,j)$, which can be prepared via Algorithm~\ref{alg:n_qubit_sampler}. In contrast, when $\kett{O(t)}_\mathcal{C}$ involves measurements in a basis $\{\Pi_k\}$ entangled across $\mathcal{H}_R$ and $\mathcal{H}_L$, this no longer suffices since $\tr(\Pi_k \sigma_{CQ})$ is generally non-vanishing.

Compared to the $2n$-qubit vectorization encoding, Algorithm~\ref{alg:n_qubit_sampler} requires only $n$ qubits, at the expense of a doubled circuit depth (due to needing to execute $U$ and $U^\dagger$ sequentially), and requires introducing classical randomization during state preparation. Note also that the $2n$-qubit approach requires queries to $U^\dagger$ and $U^\mathsf{T}$ (if they are to be implemented in parallel; c.f. Appendix~\ref{app:avoiding_transpose}), while the $n$-qubit approach requires queries to $U^\dagger$ and $U$. Similarly, to sample from the operator basis $\{\ketbra{\psi_j}{\phi_i}\}$, the $2n$-qubit approach requires querying $U_\psi^\dagger \otimes U^\mathsf{T}_\phi$, while the $n$-qubit approach requires querying $U_\phi$ and $U^\dagger_\psi$. 

\subsection{\texorpdfstring{$n$}{n}-qubit analogues of main results} \label{app:n_qubit_algorithms}

Due to Theorem~\ref{theorem:1_copy_sampling}/Algorithm~\ref{alg:n_qubit_sampler}, the sampling part of the estimation of quantities $F(O(t))$ or $F(U)$ via the Monte Carlo estimator Eq.~(\ref{eq:computable_general_form}), restricted to operator bases of the form $\{\ketbra{\psi_i}{\phi_j}\}$ (or equivalently, that do not require sampling from $\kett{O(t)}_\mathcal{C}$ or $\kett{U}_\mathcal{C}$ in a basis that entangles $\mathcal{H}_L$ and $\mathcal{H}_R$), can be achieved using only the $n$-qubit quantum states $\rho(O)$ or $\rho(U)$. While the resources required to draw each sample exhibit a memory-depth tradeoff compared to the vectorization approach, these tasks take the same amount of samples, since they are based on the same estimators. In the following, we re-state these results for completeness, following the presentation of Sections~\ref{sec:superop}
 and \ref{sec:2pc_main}.
 
\begin{center}
    ~ \\ \textbf{Expectation values of self-adjoint superoperators/OTOCs via $\rho(O)$} ~\\
\end{center}

\begin{itemize}
    \item {Estimating a single OTOC:} to estimate $\langle \mathcal{A}_{PQ} \rangle_{O} = \tr(OPOQ)/2^n$, the $n$-qubit approach requires $N=\bigo{1/\epsilon_{\text{RMS}}^2}$ samples of $\rho(O)$ (in any operator basis).

    \item {Estimating a single superoperator expectation value $\langle \mathcal{A} \rangle_{O}$ with $M$ terms in its decomposition:} the $n$-qubit approach can estimate $\langle \mathcal{A} \rangle_{O}$ term-by-term via Eq.~(\ref{eq:superop_term_by_term}). This requires $N_{\text{naive}} = \left( \sum_{kl} \abs{f_{kl}} \sqrt{\text{Var}(\mathcal{A}_{kl})_{O(t)}} \right)^2/\epsilon_{\text{RMS}}^2$ samples of $\rho(O)$ (in any operator basis).

    If, additionally, terms in $\mathcal{A}$ can be grouped such that within each group, all superoperators $\mathcal{A}_i = P_i (\cdot) Q_i^\dagger$ satisfy $[P_i,P_j]=[Q_i,Q_j]=0$, the number of samples can be reduced to $N_{\text{comm}}$; see Eq.~(\ref{eq:wds_alloc}).

    \item {Simultaneously estimating many OTOCs:} the $n$-qubit approach can also exploit commutation between superoperators for their simultaneous estimation. However, it is restricted to the `all simultaneously commute' case (Fig.~\ref{fig:comutation}.(a)) where all superoperators $\mathcal{A}_i = P_i (\cdot) Q_i^\dagger$ satisfy $[P_i,P_j]=[Q_i,Q_j]=0$ (Eq.~(\ref{eq:lemma_all_commuting}) of Lemma~\ref{app:commutation_condition}). This yields Theorem~\ref{theorem:n_qubit_main}.(a), which requires $N=\bigo{1/\epsilon_{\text{RMS}}^2}$ samples of $\rho(O, \{\ketbra{\phi_i}{\psi_j}\}_{i,j})$, where $\{\ket{\phi_i}\}_{i}$ is the common eigenbasis of $\{P_i\}_i$ and $\{\ket{\psi_i}\}_{i}$ is the common eigenbasis of $\{Q_i\}_i$.

    Notably, as the Pauli basis is not of the form $\{\ketbra{\phi_i}{\psi_j}\}$ (which is more generally due to Lemma~\ref{lemma:anticomm_basis}), we do not immediately have an $n$-qubit analogue of Corollary~\ref{corr:diagonal_superop}, where terms are allowed to simultaneously anticommute, i.e. $\{P_i, P_j\} = \{Q_i, Q_j\} = 0$.
\end{itemize}

\begin{center}
    ~ \\ \textbf{Two-point correlators via $\rho(U)$} ~\\
\end{center}

Alternatively, we can encode the unitary $U$ using $n$ qubits by simply setting $V=U$ in Algorithm~\ref{alg:n_qubit_sampler}. This enables sampling from the probability distribution $p(i,j)= |\bra{\psi_j} U \ket{\phi_i}|^2/2^n$ (or equivalently, encode this probability distribution as the state $\rho(U, \{\ketbra{\phi_i}{\psi_j}\}_{i,j})$). Mirroring the discussion above, we obtain two-point correlator analogues of our results in Section~\ref{sec:2pc_main} using $n$ qubits.
\vspace{5pt}

\begin{itemize}
    \item {Estimating a single two-point correlator:} to estimate $\tr(P Q(t))/2^n$, the $n$-qubit approach requires $N=\bigo{1/\epsilon_{\text{RMS}}^2}$ samples of $\rho(U)$ (in any operator basis).

    \item {Estimating a linear combination of two-point correlators:} the quantity $\langle \mathcal{A} \rangle_{U} = \sum_{kl} f_{kl} \tr(P_k Q_l(t))$, where $\mathcal{A}$ is a superoperator with $M$ terms, can be estimated term-by-term using $N_{\text{naive}} = \left( \sum_{kl} \abs{f_{kl}} \sqrt{\text{Var}(\mathcal{A}_{kl})_{U}} \right)^2/\epsilon_{\text{RMS}}^2$ samples of $\rho(U)$ (in any operator basis).

    If, additionally, terms in $\mathcal{A}$ can be grouped such that within each group, all superoperators $\mathcal{A}_i = P_i (\cdot) Q_i^\dagger$ satisfy $[P_i,P_j]=[Q_i,Q_j]=0$, the number of samples can be reduced to $N_{\text{comm}}$; see Eq.~(\ref{eq:wds_alloc}).

    \item {Simultaneously estimating many two-point correlators:} similarly, the $n$-qubit approach can exploit commutation between superoperators for simultaneous estimation for the `all simultaneously commute' case. This yields Theorem~\ref{theorem:n_qubit_main}.(b), which requires $N=\bigo{1/\epsilon_{\text{RMS}}^2}$ samples of $\rho(U, \{\ketbra{\phi_i}{\psi_j}\}_{i,j})$, where $\{\ket{\phi_i}\}_{i}$ is the common eigenbasis of $\{P_i\}_i$ and $\{\ket{\psi_i}\}_{i}$ is the common eigenbasis of $\{Q_i\}_i$.
\end{itemize}

\section{Extension to finite-temperature quantities} \label{app:finite_temp}

In this appendix, we briefly discuss how our framework can be extended to the estimation of correlators at finite temperature. We focus on protocols that have access to $2n$ qubits, are able to query $U^\dagger$ and $U$ or $U^\mathsf{T}$, and are able to implement the imaginary time propagator $\propto e^{-\beta H}$ (to be defined later), which yield analogues of our results of Section~\ref{sec:superop}, Appendix~\ref{app:superop}, and Appendix~\ref{app:2pc_amplitude}. Nonetheless, most of our developed formalism holds, illustrating the generality of our framework. For instance, they can be extended to yield $n$-qubit algorithms via the discussion of Section~\ref{sec:n_qubit_rand} and Appendix~\ref{app:n_qubit_details}. The resulting protocols generalize existing methods for the computation of regularized OTOCs \cite{swingle2018resilience,sundar2022proposal,green2022experimental}.

The (infinite-temperature) OTOCs that we have mainly considered can be generalized to finite temperatures -- also called thermally regulated OTOCs \cite{viswanath1994recursion,maldacena2016bound,xu2024scrambling} -- as:
\begin{equation} \label{eq:t_otoc_general}
    \text{OTOC}_\beta(O(t),A,B; \alpha_1,\alpha_2,\alpha_3,\alpha_4) \equiv \frac{1}{Z}\tr(\rho_\beta^{\alpha_1} O(t) \rho_\beta^{\alpha_2} A \rho_\beta^{\alpha_3} O(t) \rho_\beta^{\alpha_4} B),
\end{equation}
where $O(t)=e^{iHt} O e^{-iHt}$, $Z = \tr(e^{-\beta H})$, and $\rho_\beta \equiv e^{- \beta H}$ is proportional to the thermal density matrix at inverse temperature $\beta$. $\alpha_i \in [0,1]$ can be chosen in many ways, as long as $\alpha_1 + \alpha_2 + \alpha_3 + \alpha_4 = 1$. For example, the case where only a single term is 1 and $\beta=0$ corresponds to the OTOCs discussed in the main text. In the following, we consider choices of $\alpha_i$ where any two of the four $\alpha_i$'s are equal to $1/2$.

Similarly, a finite temperature generalization of the two-point correlation function between Heisenberg operators $O_1,O_2$ is:
\begin{equation} \label{eq:2ptcorr_general}
    C_\beta(O_1, O_2) \equiv \frac{1}{Z} \tr(\rho_{\beta/2} O_1 \rho_{\beta/2} O_2),
\end{equation}
also known as the Wightman inner product \cite{parker2019universal}, appearing in e.g. finite temperature dynamical structure factors \cite{baez2020dynamical}. 

Given access to samples of the state $\kett{O(t)}$ -- where the basis of vectorization is taken to be $\mathcal{C}$ and suppressed in this section -- and the ability to apply the normalized imaginary time propagator:
\begin{equation} \label{eq:normalized_ite}
    \ket{\psi} \rightarrow \frac{e^{-\beta H/2}}{\sqrt{\bra{\psi} e^{-\beta H} \ket{\psi}}}\ket{\psi} \equiv A^{\beta/2}_\psi(H) \ket{\psi}
\end{equation}
for arbitrary states $\ket{\psi}$, we describe here how the finite-temperature quantities Eq.~(\ref{eq:t_otoc_general}) and Eq.~(\ref{eq:2ptcorr_general}) can be estimated using our framework. The imaginary-time transformation Eq.~(\ref{eq:normalized_ite}) -- involving the application of the non-unitary operator $A^{\beta/2}_\psi(H)$ -- can be implemented on a quantum computer probabilistically \cite{leadbeater2304non,xie2024probabilistic} in conjunction with amplitude amplification, or variationally \cite{mcardle2019variational}.

Firstly, consider $A^{\beta/2}_{O(t)}(H \otimes \mathbb{I})$, the action of the imaginary time propagator on $\kett{O(t)}$ with the Hamiltonian $H \otimes \mathbb{I}$. Through the ricochet identity Eq.~(\ref{eq:ricochet}), the normalization factor in the denominator resolves to:
\begin{align*}
    \bbra{O(t)} e^{-\beta H} \otimes \mathbb{I} \kett{O(t)}^{1/2} &= \tr( O(t)^2  e^{-\beta H})^{1/2} \\
    &= \tr( O^2 e^{-iHt} e^{-\beta H} e^{iHt})^{1/2} \\
    &= \tr( O^2 e^{-\beta H})^{1/2} \\
    &= \sqrt{Z},
\end{align*}
implying that $A^{\beta/2}_{O(t)}(H \otimes \mathbb{I}) = \rho_{\beta/2}/\sqrt{Z} \otimes \mathbb{I}$, where we made use of the cyclicity of the trace, the Hermiticity of $O(t)$, that the imaginary and real time propagators due to the same Hamiltonian commute, and $O^2 = \mathbb{I}$. Similarly $A^{\beta/2}_{O(t)}(\mathbb{I} \otimes H^\mathsf{T}) = \mathbb{I} \otimes \rho_{\beta/2}^\mathsf{T}/\sqrt{Z}$. States resulting from applications of $A^{\beta/2}_{O(t)}(H \otimes \mathbb{I})$ and/or $A^{\beta/2}_{O(t)}(\mathbb{I} \otimes H^\mathsf{T})$ therefore read:
\begin{align*}
    \ket{O(t)}_L &\equiv A^{\beta/2}_{O(t)}(H \otimes \mathbb{I}) \kett{O(t)} = \kett{\rho_{\beta/2} O(t)}/\sqrt{Z}, \\
    \ket{O(t)}_R &\equiv A^{\beta/2}_{O(t)}(\mathbb{I} \otimes H^\mathsf{T}) \kett{O(t)} = \kett{O(t) \rho_{\beta/2}}/\sqrt{Z},
\end{align*}
where the tensor product in the third equality emphasizes that the imaginary-time transformation can be parallelized.

Overlaps between the above states yield the two-point correlators of Eq.~(\ref{eq:2ptcorr_general}):
\begin{equation}
    C_\beta(O_1, O_2) \equiv \frac{1}{Z} \tr(O_1 \rho_{\beta/2} O_2 \rho_{\beta/2}) = {}_L\!\braket{O_1|O_2}_R.
\end{equation}
Following Appendix~\ref{app:2pc_amplitude}, they can be obtained via amplitude estimation algorithms such as the Hadamard test.

On the other hand, to evaluate the regulated OTOC Eq.~(\ref{eq:t_otoc_general}), one can compute the expectations values of the above states, i.e.:
\begin{align}
    \text{OTOC}_\beta(O(t), A, B, 0,1/2,1/2,0) &= \frac{1}{Z} \tr(O(t)\rho_{\beta/2} A \rho_{\beta/2} O(t) B) = {}_L\bra{O(t)} A \otimes B^\mathsf{T} \ket{O(t)}_L, \\
    \text{OTOC}_\beta(O(t), A, B, 1/2,0,0,1/2) &= \frac{1}{Z} \tr(\rho_{\beta/2} O(t) A O(t) \rho_{\beta/2} B) = {}_R\bra{O(t)} A \otimes B^\mathsf{T} \ket{O(t)}_R.
\end{align}
Importantly, as the regulated OTOCs are equivalent to expectation values, the discussions of simultaneous computation via commutation of Section~\ref{sec:superop} and Appendix~\ref{app:superop} fully apply here. This allows, for instance, the simultaneous computation of the set of all diagonal regulated OTOCs, i.e. $\{\text{OTOC}_\beta(O(t), P, P, 1/2,0,0,1/2)\}_{P \in \mathcal{P}_n}$ and $\{\text{OTOC}_\beta(O(t), P, P, 0,1/2,1/2,0)\}_{P \in \mathcal{P}_n}$, by sampling from the states $\ket{O(t)}_L$ and $\ket{O(t)}_R$ in the Bell basis. This is the finite temperature analogue of Corollary~\ref{corr:diagonal_otoc}.

Further choices of $\alpha_i$ can be accessed by applying the operators $A,B^\mathsf{T}$ as unitaries before computing overlaps, i.e.:
\begin{align}
    \text{OTOC}_\beta(O(t), A, B, 0,1/2,0,1/2) &= \frac{1}{Z} \tr( O(t) \rho_{\beta/2} A O(t) \rho_{\beta/2} B) = {}_L\bra{O(t)} A \otimes B^\mathsf{T} \ket{O(t)}_R.
\end{align}
The above two methods cover $3$ out of $4$ possible types of OTOCs (from ${4 \choose 2} = 6$ choices for the coefficients $\alpha_i$, where some choices result in the same OTOCs -- up to a phase depending on $A,B$ -- due to the cyclicity of the trace).

\section{Extension to open-system dynamics} \label{app:channel_gen}

In this appendix, we describe an approach to generalize our framework beyond unitary evolution, to also include the simulation of open-system dynamics in the Heisenberg picture \cite{swingle2018resilience,zhang2019information,zanardi2021information,schuster2023operator}. 

Let $\mathcal{E}$ denote a CPTP map (carrying out Schrödinger-picture evolution $\rho \rightarrow \mathcal{E}(\rho)$), and $\mathcal{E}^\dagger$ its unital dual (carrying out Heisenberg-picture evolution $O \rightarrow \mathcal{E}^\dagger(O)$, where $O$ is Hermitian; in some works also denoted $\mathcal{E}^*$). One approach to study $\mathcal{E}^\dagger(O)$ via the vectorization map on quantum computers is by preparing and extracting information from the pure state $\kett{\mathcal{E}^\dagger(O)}_\mathcal{C} \equiv \frac{\vecc{\mathcal{E}^\dagger(O)}}{\sqrt{\tr((\mathcal{E}^\dagger(O))^2)}}$, where $\vecc{\cdot}$ denotes the unnormalized vectorization map in the computational basis (used here to track normalization factors more carefully).

\vspace{5pt}

The following result provides a probabilistic algorithm for its preparation:
\begin{lemma}
Let $\mathcal{E}$ be a CPTP map acting on $k$ qubits, and $U_{SE}$ its Stinespring dilation. Given $\kett{O}_\mathcal{C}$ and the ability to query $U_{SE}^\dagger$, and $U_{SE}$ or $U_{SE}^\mathsf{T}$, there exists an algorithm that implements the map $\kett{O}_\mathcal{C} \rightarrow \kett{\mathcal{E}^\mathrm{\dagger}(O)}_\mathcal{C}$ with success probability:
\begin{equation}
    p = \frac{\tr((\mathcal{E}^\mathrm{\dagger}(O))^2)}{2^{n+n_E}} \geq \frac{\tr((\mathcal{E}^\mathrm{\dagger}(O))^2)}{2^{n+2k}}
\end{equation}
using $2 n_E \leq 4k$ ancilla qubits.
\end{lemma}

\begin{proof}
Let $S$ denote the set of $k$ qubits $\mathcal{E}$ is acting on. Its action is equivalent to that of a larger unitary $U_{SE}$ (its Stinespring dilation) acting jointly on the same $k$ qubits, along with $n_E \leq 2k$ ancillary qubits taken to be $\ketbra{e}{e} \equiv \ketbra{0...0}{0...0}$ (denoted $E$), followed by tracing out the ancillary subsystem, i.e.:
\begin{equation} \label{eq:rho_dilation}
    \mathcal{E}(\rho) = \tr_E(U_{SE} \rho \otimes \ketbra{e}{e} U_{SE}^\dagger).
\end{equation}
Via the equality $\tr(\mathcal{E}(\rho)O) = \tr(\rho \mathcal{E}^\dagger(O))$, an analogous relation for the unital dual $\mathcal{E}^\dagger$ holds, i.e.:
\begin{equation} \label{eq:op_dilation}
    \mathcal{E}^\dagger(O) = (\mathbb{I}_S \otimes \bra{e}) U_{SE}^\dagger (O \otimes \mathbb{I}_E) U_{SE} (\mathbb{I}_S \otimes \ket{e} ).
\end{equation}

Viewing $\mathcal{E}^\dagger$ as a superoperator on $O$, we can write:
\begin{equation}
    \vecc{\mathcal{E}^\dagger(O)} = M^{\mathcal{E}^\dagger}_\mathcal{C} \vecc{O},
\end{equation}
where $M^{\mathcal{E}^\dagger}_\mathcal{C} = (M^{\mathcal{E}}_\mathcal{C})^\dagger$ denotes the transfer matrix of $\mathcal{E}^\dagger$ in the computational basis. It can be determined by repeatedly applying the ricochet property as:
\begin{align*}
    \vecc{\mathcal{E}^\dagger(O)} &= \vecc{(\mathbb{I}_S \otimes \bra{e}_E) U_{SE}^\dagger (O \otimes \mathbb{I}_E) U_{SE} (\mathbb{I}_S \otimes \ket{e}_E )} \\
    &= ((\mathbb{I}_S \otimes \bra{e}_E) U_{SE}^\dagger) \otimes (U_{S'E'} (\mathbb{I}_S \otimes \ket{e}_{E'} ))^\mathsf{T} \vecc{O \otimes \mathbb{I}_E} \\
    &= (\mathbb{I}_S \otimes \mathbb{I}_{S'} \otimes \bra{e}_E \otimes \bra{e}_{E'}) (U_{SE}^\dagger \otimes U_{S'E'}^\mathsf{T})~ \vecc{O} \otimes \vecc{\mathbb{I}_E},
\end{align*}
where $S'$ and $E'$ denote the replicas of $S$ and $E$ respectively.

The final equality implies that $\kett{\mathcal{E}^\dagger(O)}_\mathcal{C}$ can be prepared by applying the unitary $U_{SE}^\dagger \otimes U_{S'E'}^\mathsf{T}$ to the state $\kett{O}_\mathcal{C} \otimes \kett{\mathbb{I}_E}_\mathcal{C}$, measuring the register $EE'$, and post-selecting on the outcome $\ket{e}_E \otimes \ket{e}_{E'}$, in which case the register $SS'$ will be in the desired state. The probability that this occurs is given by:
\begin{equation}
    p = \frac{\vecc{\mathcal{E}^\dagger(O)}^\dagger \vecc{\mathcal{E}^\dagger(O)}}{\tr(O^2) \tr(\mathbb{I}_E)} = \frac{\tr((\mathcal{E}^\dagger(O))^2)}{2^{n+n_E}} \geq \frac{\tr((\mathcal{E}^\dagger(O))^2)}{2^{n+2k}}.
\end{equation}
\end{proof}

\begin{figure}[h]
\centering
\begin{quantikz}[row sep={0.7cm,between origins}]
\lstick[wires=3]{\(\kett{O}\)} \qw & \qwbundle{n-2} & \qw & \qw & \qw& \qw& \qw  \rstick[wires=3]{$\kett{\mathcal{E}^\dagger(O)}$}\\
 & \qw       & \qw &      \targ{} \gategroup[wires=4,steps
=3,style={dashed,
rounded corners, inner sep=6pt}, label style={label position=below, , yshift=-0.5cm}]{$U^\dagger \otimes U^\mathsf{T}$} & \qw & \qw & \qw \\
 & \qw       & \qw &      \qw        & \targ{} & \qw & \qw &  \\
\lstick{\(\ket{0}\)}    & \gate{H}  & \ctrl{1} & \ctrl{-2}  & \qw  & \gate{R_Y(\theta^*)} & \meter{} \\
\lstick{\(\ket{0}\)}    & \qw       & \targ{} &  \qw        & \ctrl{-2}   & \gate{R_Y(\theta^*)} & \meter{}
\end{quantikz}
\caption{Circuit for simulation of bit-flip unital channel, with Kraus operators $\{ E_0 = \sqrt{1-p} ~\mathbb{I}, E_1 = \sqrt{p} ~X \}$. The dilation $U$ such that $U \ket{\psi}\ket{0} = \sqrt{p} \ket{\psi} \ket{0} + \sqrt{1-p} X \ket{\psi} \ket{1}$ is $U = R_Y(\theta) ~ CX$, where $R_Y(\theta) = e^{-iY\theta/2}$ with $\theta$ is chosen such that $\cos(\theta/2) = \sqrt{1-p}$ and $\sin{\theta/2} = \sqrt{p}$.}
\label{fig:block_encoding_example}
\end{figure}
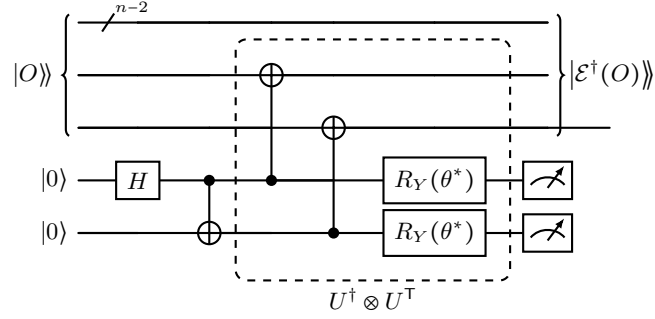

Once samples of the state $\kett{\mathcal{E}^\dagger(O)}$ have been prepared, information regarding $\mathcal{E}^\dagger(O)$ can be extracted from them via measurements, following the framework described in the main text. For instance, taking expectation values yields OTOCs in the open system case \cite{zanardi2021information}.

The above algorithm can be viewed as implementing the block encoding of the non-unitary operator $M^{\mathcal{E}^\dagger}_\mathcal{C}$, which requires implementing the transpose of the Stinespring dilation $U_{SE}$ of $\mathcal{E}$, and its adjoint $U_{SE}^\dagger$. The transpose can be avoided in the same way as in the unitary case (following the discussion of Appendix~\ref{app:avoiding_transpose}), at the expense of having to implement $U_{SE}$ and $U_{SE}^\dagger$ sequentially. Moreover, channels that admit Stinespring dilations that are Clifford can also be implemented with only Clifford operations and measurements. As an illustrative example, Fig.~\ref{fig:block_encoding_example} displays the circuit that implements $\kett{O} \rightarrow \kett{\mathcal{E}^\dagger(O)}$ for the bit-flip channel.

When applied to simulate a $k$-qubit channel acting on $S$ (a subsystem of the larger $n$-qubit Heisenberg operator $O$), only the corresponding $2k$ qubits (labelled $SS'$) of the $2n$-qubit state $\kett{O}$ participate in the algorithm (along with $2n_E$ ancilla qubits labelled $EE'$). This enables the implementation of non-overlapping channels to be parallelized. 

A drawback of probabilistic implementations in general -- shared with similar approaches to implement non-unitary operations \cite{leadbeater2304non,azses2024nonunitary} -- is the exponential suppression in total success probability in the number of block encodings applied. Nevertheless, the probabilistic implementation of $M^{\mathcal{E}^\dagger}_\mathcal{C}$ is unavoidable in general, since non-unitary operations generally cannot be implemented deterministically \cite{terashima2005nonunitary}. The success probability can also be increased via amplitude amplification.

\section{Extension to qudits} \label{app:qudit}

In this appendix, we describe how our framework can be generalized to the case of $n$ qudits, where the local Hilbert space dimension is of prime dimension $d$. Through an appropriate generalization of the vectorization map for qudits, we show that our framework extends naturally to this case, allowing the preparation and extraction of information from vectorized $n$-qudit Heisenberg operators. Similar to the qudit case, the preservation of Cliffordness and entanglement resources also holds.

Following standard definitions \cite{veitch2014resource,wang2020qudits}, begin by defining the qudit $X$ and $Z$ gates, also known as the shift and the boost operators respectively:
\begin{equation} \label{eq:qudit_pauli_def}
    X_d = \sum_{j=0}^{d-1} \ket{j + 1} \! \bra{j}, \quad Z_d = \sum_{j=0}^{d-1} \omega^j \ket{j} \! \bra{j},
\end{equation}
where $\omega \equiv e^{2 \pi i/d}$, and additions are $\text{mod d}$. Subsequently, the generalized qudit Pauli operators, also known as Heisenberg-Weyl operators, are defined as:
\begin{align*}
P(a_z, a_x) = \phi(a_z, a_x, d) Z_d^{a_z} X_d^{a_x}, \quad \phi(a_z, a_x, d) =&
    \left\{\begin{aligned}
    & i^{a_z a_x}, \quad &d &= 2 \\
    & \omega^{-a_x a_z/2}, \quad &d &\neq 2 
    \end{aligned}\right.
\end{align*} 
where $(a_z, a_x) \in \mathbb{Z}_d \times \mathbb{Z}_d$ indexes the $d^2$ possible generalized Pauli operators. The $d^{2n}$ generalized Pauli operators for $n$-qudit systems are then:
\begin{equation}
    P\left(\bigoplus_{i=1}^{n} (a^i_z, a^i_x) \right) \equiv \bigotimes_{i=1}^{n} P(a^i_z, a^i_x) \in \mathcal{P}^d_n.
\end{equation}

Next, define the qudit generalizations of the controlled-not and Hadamard gates, which are the qudit Hadamard gate $H_d$ and the qudit $\text{SUM}$ gate $\text{SUM}_d$:
\begin{align}
    H_d &= \frac{1}{\sqrt{d}} \sum_{m,n = 0}^{d-1} \omega^{mn} \ket{m}  \! \bra{n}, \\
    \text{SUM}_d &= \sum_{m,n = 0}^{d-1} \ket{m, m+n }  \! \bra{m,n}.
\end{align}
Similar to their qubit counterparts, they are (generalized) Clifford unitaries, i.e. unitaries that, up to a phase, conjugate Heisenberg-Weyl operators to other Heisenberg-Weyl operators.

We can now define the generalization of the vectorization maps Eqs.~(\ref{eq:def_comp_vect}) and (\ref{eq:def_pauli_vect}) for a single qudit as:
\begin{alignat}{2}
    \kett{\ketbra{i}{j}}_\mathcal{C} &\equiv \ket{i} \otimes \ket{j}, \quad && i,j \in \mathbb{Z}_d \label{eq:def_comp_vect_qudit} \\ 
    \kett{Z_d^{a_z} X_d^{a_x}}_{\mathcal{P}^d} &\equiv \ket{a_z} \otimes \ket{- a_x}, \quad && a_z, a_x \in \mathbb{Z}_d, \label{eq:def_pauli_vect_qudit}
\end{alignat}
which generalizes to $n$ qudits by taking tensor products (and noting that $\ket{-a_x}$ is equal to $\ket{a_x}$ for $d=2$, recovering the qubit case of Eq.~(\ref{eq:def_pauli_vect})).

Notably, the ricochet identity continues to hold, enabling the discussion on the preparation and extraction of information to carry over to vectorized qudit Heisenberg operators $\kett{O(t)}$. In particular, sampling from $\kett{O(t)}_{\mathcal{P}_n^d}$ enables sampling from the generalized Pauli distribution of $O(t)$, which enables the computation of OTOCs and geometrical properties of $O(t)$, while the amplitudes of $\kett{O(t)}_{\mathcal{P}_n^d}$ encode two-point correlators (between qudit operators). 

Furthermore, as the result below shows (also observed in \cite{somma2025shadow}), the basis change unitary between the qudit computational and Pauli bases, i.e. the unitary carrying the transformation $\kett{O}_{\mathcal{C}^d_n} \rightarrow \kett{O}_{\mathcal{P}^d_n}$, is simply the generalization of the Bell basis transformation for the qubit case of Eq.~(\ref{eq:rotation_cp}):
\begin{lemma} \label{app:qudit_basis_change}
    The vectorization basis change unitary $R_{\mathcal{C}^d_n,\mathcal{P}^d_n}$ between the $n$-qudit computational and generalized Pauli bases is given by $R_{\mathcal{C}^d_n,\mathcal{P}^d_n} = \bigotimes_{i=1}^n R_{\mathcal{C}^d,\mathcal{P}^d}$, where:
    \begin{equation} \label{eq:qubit_basis_form}
        R_{\mathcal{C}^d,\mathcal{P}^d} = \mathrm{SUM}_d \cdot \mathrm{H}_d \otimes \mathbb{I},
    \end{equation}
    which is a depth-2 local generalized Clifford unitary.
\end{lemma}
\begin{proof}
It suffices to show that Eq.~(\ref{eq:qubit_basis_form}) holds for any qudit Pauli operator $P(a_z, a_x)$. Firstly, note that:
\begin{equation*}
    (Z_d)^{a_z} = \sum_{j=0}^{d-1} \omega^{a_z j} \ket{j} \! \bra{j},
\end{equation*}
and:
\begin{align*}
    (X_d)^{a_x} &= \left( \sum_{j_1=0}^{d-1} \ket{j_1 + 1 } \! \bra{j_1} \right) \left( \sum_{j_2=0}^{d-1} \ket{j_2 + 1 } \! \bra{j_2} \right) ... \left( \sum_{j_{a_x}=0}^{d-1} \ket{j_{a_x} + 1 } \! \bra{j_{a_x}} \right) \\
    &= \left( \sum_{j=0}^{d-1} \ket{j + 1 } \! \bra{j-1 } \right) \left( \sum_{j_3=0}^{d-1} \ket{j_3 + 1 } \! \bra{j_3} \right)  ... \left( \sum_{j_{a_x}=0}^{d-1} \ket{j_{a_x} + 1 } \! \bra{j_{a_x}} \right) \\
    &= \left( \sum_{j=0}^{d-1} \ket{j + 1 } \! \bra{j+1-a_x } \right),
\end{align*}
where we repeatedly applied:
\begin{equation*}
    \left( \sum_{j=0}^{d-1} \ket{j + 1 } \! \bra{a} \right)\left( \sum_{j'=0}^{d-1} \ket{j' + 1 } \! \bra{j'} \right) =  \sum_{j=0}^{d-1} \ket{j + 1 } \! \bra{a-1 }.
\end{equation*}

On one hand, making use of the ricochet identity and the linearity of vectorization, we have:
\begin{align*}
    \kett{P(a_z, a_x)}_\mathcal{C} &= \kett{\phi(a_z, a_x, d) Z_d^{a_z} X_d^{a_x}}_\mathcal{C}\\
    &= \frac{\phi(a_z, a_x, d)}{\sqrt{d}} \kett{\left( \sum_{j=0}^{d-1} \omega^{a_z j} \ket{j} \! \bra{j} \right)  \left( \sum_{j'=0}^{d-1} \ket{j' + 1 } \! \bra{j'+1-a_x } \right)}_\mathcal{C} \\
    &= \frac{\phi(a_z, a_x, d)}{\sqrt{d}} \sum_{j=0}^{d-1} w^{j a_z} \ket{j} \ket{j-a_x}. \\
\end{align*}
On the other hand, using the above expressions for $(Z_d)^{a_z}$ and $(X_d)^{a_x}$ leads to:
\begin{align*}
    (\text{SUM}_d \cdot \text{H}_d \otimes \mathbb{I}) \kett{P(a_z, a_x)}_{\mathcal{P}^d} &= \phi(a_z, a_x, d) (\text{SUM}_d \cdot \text{H}_d \otimes \mathbb{I}) \kett{Z_d^{a_z} X_d^{a_x}}_{\mathcal{P}^d} \\
    &= \phi(a_z, a_x, d) (\text{SUM}_d \cdot \text{H}_d \otimes \mathbb{I}) \ket{a_z} \otimes \ket{- a_x} \\
    &= \phi(a_z, a_x, d) \left( \sum_{m,n = 0}^{d-1} \ket{m, m+n }  \! \bra{m,n} \right) \left( \frac{1}{\sqrt{d}} \sum_{m',n' = 0}^{d-1} \omega^{m'n'} \ket{m'}  \! \bra{n'} \otimes \mathbb{I} \right) \ket{a_z} \otimes \ket{- a_x} \\
    &= \frac{\phi(a_z, a_x, d)}{\sqrt{d}} \sum_{m=0}^{d-1} \omega^{m a_z} \ket{m, m - a_x },
\end{align*}
yielding our claim.
\end{proof}

As a consequence, the above qudit generalization also preserves the Cliffordness of time-evolution due to $U$.

\end{document}